\pdfoutput=1 
\newif\ifextended\extendedtrue

\documentclass[sigconf, nonacm]{acmart}

\usepackage{amsthm}
\usepackage[binary-units=true]{siunitx}
\usepackage{amsmath}

\newcommand*{\symDefine}[2]{\newcommand{{#1}}{{#2}}}

\symDefine{\symMVP}{\text{MVP}}
\symDefine{\symUpdateEvent}{A}
\symDefine{\symMasked}{a}
\symDefine{\symBase}{b}
\symDefine{\symCount}{c}
\symDefine{\symNumExtraBits}{d}
\symDefine{\symTransformation}{f}
\symDefine{\symSomeFunc}{f}
\symDefine{\symRegContrib}{g}
\symDefine{\symRegContribCorr}{\symRegContrib_\text{corr}}
\symDefine{\symRegMartingale}{h}
\symDefine{\symHashBit}{h}
\symDefine{\symRegAddr}{i}
\symDefine{\symIndexJ}{j}
\symDefine{\symRegAddrOther}{j}
\symDefine{\symUpdateVal}{k}
\symDefine{\symUpdateValOther}{l}
\symDefine{\symIndexBit}{l}
\symDefine{\symIndexBitSimple}{\tilde{\symIndexBit}}
\symDefine{\symNumReg}{m}
\symDefine{\symCardinality}{n}
\symDefine{\symCardinalityMax}{\symCardinality_\text{max}}
\symDefine{\symCardinalityEstimatorCorrected}{\symCardinalityEstimator_\text{corr}}
\symDefine{\symCardinalityEstimatorML}{\symCardinalityEstimator_\text{ML}}
\symDefine{\symCardinalityEstimatorLow}{\symCardinalityEstimator_\text{low}}
\symDefine{\symCardinalityEstimatorHigh}{\symCardinalityEstimator_\text{high}}
\symDefine{\symCardinalityEstimator}{\hat{\symCardinality}}
\symDefine{\symCardinalityEstimatorAlt}{\symCardinalityEstimator_\text{alt}}
\symDefine{\symCardinalityEstimatorMartingale}{\symCardinalityEstimator_\text{martingale}}
\symDefine{\symPrecision}{p}
\symDefine{\symProbability}{\Pr}
\symDefine{\symRegisterBit}{q}
\symDefine{\symBitsForMax}{q}
\symDefine{\symRegister}{r}
\symDefine{\symRegisterSimple}{\tilde{\symRegister}}
\symDefine{\symSum}{s}
\symDefine{\symIndexS}{s}
\symDefine{\symMaxUpdateVal}{u}
\symDefine{\symMaxUpdateValSimple}{\tilde{\symMaxUpdateVal}}
\symDefine{\symVarianceFactor}{v}
\symDefine{\symMaxUpdateValMax}{w}
\symDefine{\symStatisticX}{X}
\symDefine{\symX}{x}
\symDefine{\symEventX}{X}
\symDefine{\symEventY}{Y}
\symDefine{\symY}{y}
\symDefine{\symZ}{z}
\symDefine{\symZEstimate}{\hat{z}}

\symDefine{\symLikelihood}{\mathcal{L}}
\symDefine{\symFisher}{\mathcal{I}}
\symDefine{\symShannon}{\mathcal{H}}
\symDefine{\symComplexity}{\mathcal{O}}

\symDefine{\symLikelihoodFuncExponentOne}{\alpha}
\symDefine{\symAlpha}{\alpha}
\symDefine{\symLikelihoodFuncExponentTwo}{\beta}
\symDefine{\symBeta}{\beta}
\symDefine{\symGamma}{\gamma}
\symDefine{\symGammaFunc}{\Gamma}
\symDefine{\symDensity}{\rho}
\symDefine{\symDensityUpdate}{\symDensity_\text{update}}
\symDefine{\symDensityRegister}{\symDensity_\text{reg}}
\symDefine{\symDensityRegisterSimple}{\tilde{\symDensity}_\text{reg}}
\symDefine{\symSomeConstant}{\kappa}
\symDefine{\symCubicContribFunc}{\psi}
\symDefine{\symContributionCoefficient}{\eta}
\symDefine{\symEstimationConstant}{\lambda}
\symDefine{\symNewConstant}{\omega}
\symDefine{\symSmallRangeCorrFunc}{\sigma}
\symDefine{\symLargeRangeCorrFunc}{\varphi}
\symDefine{\symZetaFunc}{\zeta}
\symDefine{\symStateChangeProbability}{\mu}
\symDefine{\symGRA}{\tau}

\DeclareMathOperator*{\symBias}{Bias}
\DeclareMathOperator*{\symVariance}{Var}
\DeclareMathOperator*{\symExpectation}{\mathbb{E}}

\usepackage{luatex85}
\usepackage[a-2b]{pdfx}
\usepackage{xifthen}
\usepackage[nolist,nohyperlinks]{acronym}
\usepackage[ruled,vlined]{algorithm2e}
\usepackage{mathtools}
\usepackage[capitalise,noabbrev]{cleveref}
\usepackage{enumitem}
\usepackage{xcolor}

\makeatletter
\def\mathcolor#1#{\@mathcolor{#1}}
\def\@mathcolor#1#2#3{%
 \protect\leavevmode
 \begingroup
 \color#1{#2}#3%
 \endgroup
}
\makeatother

\newtheorem{lemma}{Lemma}

\SetKw{KwOr}{or}
\SetKwProg{Fn}{function}{}{}
\SetKwFunction{FuncUncompress}{unpack}
\SetKwFunction{FuncCompress}{pack}
\SetKwFunction{FuncNLZ}{nlz}
\SetKwComment{Comment}{$\triangleright$\ }{}

\SetCommentSty{mycommfont}

\newcommand{\myAlg}[2][]{
 \ifthenelse{\isempty{#1}}%
 {\begin{figure}}
 {\begin{figure}[#1]}
 \begingroup 
 \csname @twocolumnfalse\endcsname
 \noindent
 \resizebox{\columnwidth}{!}{%
 \begin{minipage}{1.33\columnwidth}
 \begin{algorithm}[H]
 \DontPrintSemicolon
 {#2}
 \end{algorithm}
 \end{minipage}%
 }
 \endgroup
 \end{figure}
}

\begin{acronym}
 \acro{HLL}{HyperLogLog}
 \acro{ULL}{UltraLogLog}
 \acro{EHLL}{ExtendedHyperLogLog}
 \acro{PCSA}{probabilistic counting with stochastic averaging}
 \acro{ML}{maximum likelihood}
 \acro{RMSE}{root-mean-square error}
 \acro{NLZ}{number of leading zeros}
 \acro{GRA}{generalized remaining area}
 \acro{FGRA}{further generalized remaining area}
 \acro{CPC}{compressed probability counting}
 \acro{HIP}{historic inverse probability}
 \acro{MVP}{memory-variance product}
 \acro{PMF}{probability mass function}
 \acro{FISH}{Fisher-Shannon}
 \acro{CR}{corrected raw}
 \acro{SIMD}{single-instruction multiple-data}
\end{acronym}

\allowdisplaybreaks
\DontPrintSemicolon
\begin{document}
\title{UltraLogLog: A Practical and More Space-Efficient Alternative to HyperLogLog for Approximate Distinct Counting}

\author{Otmar Ertl}
\affiliation{%
  \institution{Dynatrace Research}
  \city{Linz}
  \country{Austria}
}
\email{otmar.ertl@dynatrace.com}

\begin{abstract}
  Since its invention HyperLogLog has become the standard algorithm for approximate distinct counting. Due to its space efficiency and suitability for distributed systems, it is widely used and also implemented in numerous databases. This work presents UltraLogLog, which shares the same practical properties as HyperLogLog. It is commutative, idempotent, mergeable, and has a fast guaranteed constant-time insert operation. At the same time, it requires 28\% less space to encode the same amount of distinct count information, which can be extracted using the maximum likelihood method. Alternatively, a simpler and faster estimator is proposed, which still achieves a space reduction of 24\%, but at an estimation speed comparable to that of HyperLogLog. In a non-distributed setting where martingale estimation can be used, UltraLogLog is able to reduce space by 17\%. Moreover, its smaller entropy and its 8-bit registers lead to better compaction when using standard compression algorithms. All this is verified by experimental results that are in perfect agreement with the theoretical analysis which also outlines potential for even more space-efficient data structures. A production-ready Java implementation of UltraLogLog has been released as part of the open-source Hash4j library.
\end{abstract}

\maketitle

\pagestyle{plain}
\section{Introduction}
Many applications require counting the number of distinct elements in a data set or data stream. It is well-known that exact counting needs linear space \cite{Alon1999}. However, the space requirements can be drastically reduced, if approximate results suffice. \ac{HLL} \cite{Flajolet2007} with improved small-range estimation \cite{Heule2013, Qin2016, Zhao2016, Ertl2017, Ting2019} represents the state-of-the-art approximate distinct-count algorithm. It allows distributed counting of up to the order of $2^{64}\approx1.8\cdot10^{19}$ distinct elements with a relative standard error of $1.04/\sqrt{\symNumReg}$ using only $6\symNumReg$ bits \cite{Heule2013}. Therefore, it is nowadays offered by a big number of data stores as part of their query language (see e.g. documentation of Timescale, Redis, Oracle Database, Snowflake, Microsoft SQL Server, Google BigQuery, Vertica, Elasticsearch, Aerospike, Amazon Redshift, KeyDB, DuckDB, or Dynatrace Grail). Further applications of \ac{HLL} include query optimization \cite{Freitag2019, Pavlopoulou2022}, caching \cite{Wires2014}, graph analysis \cite{Boldi2011, Priest2018}, attack detection \cite{Chabchoub2014, Clemens2023}, network volume estimation \cite{Basat2018}, or metagenomics \cite{Baker2019, Marcais2019, Elworth2020, Breitwieser2018}.

\ac{HLL} is actually very simple as exemplified in \cref{alg:insertion_hll}. It typically consists of a densely packed array of 6-bit registers $\symRegister_0, \symRegister_1, \ldots, \symRegister_{\symNumReg-1}$ \cite{Heule2013} where the number of registers $\symNumReg$ is a power of 2, $\symNumReg = 2^\symPrecision$. The choice of the precision parameter $\symPrecision$ allows trading space for better estimation accuracy. Adding an element requires calculating a 64-bit hash value. $\symPrecision$ bits are used to choose a register for the update. The \ac{NLZ} of the remaining $64-\symPrecision$ bits are interpreted as a geometrically distributed update value with success probability $\frac{1}{2}$ and positive support, that is used to update the selected register, if its current value is smaller. Estimating the distinct count from the register values is more challenging, but can also be implemented using a few lines of code \cite{Flajolet2007, Ertl2017}.

The popularity of \ac{HLL} is based on following advantageous features that make it especially useful in distributed systems:
\begin{description}[style=unboxed,leftmargin=0cm]
  \item[Speed:] Element insertion is a fast and allocation-free operation with a constant time complexity independent of the sketch size. In particular, given the hash value of the element, the update requires only a few CPU instructions.
  \item [Idempotency:] Further insertions of the same element will never change the state. This is actually a natural property every count-distinct algorithm should support to prevent duplicates from changing the result.
  \item [Mergeability:] Partial results calculated over subsets can be easily merged to a final result. This is important when data is distributed or processed in parallel.
  \item [Reproducibility:] The result does not depend on the processing order, which often cannot be guaranteed in practice anyway. Reproducibility is achieved by a commutative insert operation and a commutative and associative merge operation.
  \item [Reducibility:] The state can be reduced to a smaller state corresponding to a smaller precision parameter. The reduced state is identical to that obtained by direct recording with lower precision. This property allows adjusting the precision without affecting the mergeability with older records.
  \item [Estimation:] A fast and robust estimation algorithm ensures nearly unbiased estimates with a relative standard error bounded by a constant over the full range of practical distinct counts.
  \item [Simplicity:] The implementation requires only a few lines of code. The entire state can be stored in a single byte array of fixed length which makes serialization very fast and convenient. Furthermore, add and in-place merge operations do not require any additional memory allocations.
\end{description}
To our knowledge \ac{HLL} is so far the most space-efficient practical data structure having all these desired properties. Space-efficiency can be measured in terms of the \ac{MVP} \cite{Pettie2021a}, which is the relative variance of the (unbiased) distinct-count estimate $\symCardinalityEstimator$ multiplied by the storage size in bits
\begin{equation}
  \label{equ:def_mvp}
  \symMVP := \symVariance(\symCardinalityEstimator/\symCardinality)\times(\text{storage size in bits}),
\end{equation}
where $\symCardinality$ is the true distinct count.
If the \ac{MVP} is asymptotically (for sufficiently large distinct counts) a constant specific to the data structure, it can be used for comparison as it eliminates the general inverse dependence of the relative estimation error on the root of the storage size. Most \ac{HLL} implementations use 6-bit registers \cite{Heule2013} to support distinct counts beyond the billion range, resulting in a \ac{MVP} of 6.48 \cite{Pettie2021a}. A recent theoretical work conjectured a general lower bound of 1.98 for the \ac{MVP} of sketches supporting mergeability and reproducibility \cite{Pettie2021}, which shows the potential for improvement. Many different approaches have been proposed to beat the space efficiency of \ac{HLL}, but they all sacrificed at least one of the properties listed above.

\subsection{Related Work}

\myAlg{
  \caption{Inserts an element with 64-bit hash value $\langle\symHashBit_{63} \symHashBit_{62} \ldots \symHashBit_{0}\rangle_2$ into a \acl*{HLL} consisting of $\symNumReg = 2^\symPrecision$ ($\symPrecision\geq 2$) 6-bit registers $\symRegister_0, \symRegister_1, \ldots, \symRegister_{\symNumReg-1}$ with initial values $\symRegister_\symRegAddr = 0$.}
  \label{alg:insertion_hll}
  $\symRegAddr\gets \langle\symHashBit_{63} \symHashBit_{62}\ldots\symHashBit_{64-\symPrecision}\rangle_2$\Comment*[r]{extract register index}
  $\symMasked\gets \langle\,\underbracket[0.5pt][1pt]{0\ldots 0}_{\scriptscriptstyle\symPrecision}\!\symHashBit_{63-\symPrecision}\symHashBit_{62-\symPrecision} \ldots\symHashBit_{0}\rangle_2$\Comment*[r]{mask register index bits}
  $\symUpdateVal \gets \FuncNLZ{$\symMasked$} -\symPrecision + 1$\Comment*[r]{update value $\symUpdateVal \in [1, 65-\symPrecision]$}
  \Comment*[r]{function $\FuncNLZ$ returns the \acl*{NLZ}}
  $\symRegister_\symRegAddr \gets \max(\symRegister_\symRegAddr, \symUpdateVal)$\Comment*[r]{update register}
}

Lossless compression of \acf{HLL} can significantly reduce the storage size \cite{Durand2004,Scheuermann2007,Lang2017}. Since the compressed state prevents random access to registers as required for insertion, bulking is needed to realize at least amortized constant update times. The required buffer partially cancels out the memory savings. Recent techniques avoid buffering of insertions. The Apache Data Sketches library \cite{ApacheDataSketches} provides an implementation using 4 bits per register to store the most frequent values relative to a global offset. Out of range values are kept separately in an associative array. Overall, this leads to a smaller \ac{MVP} but also to a more expensive insert operation. Its runtime is proportional to the memory size in the worst case, because all registers must be updated whenever the global offset is increased. HyperLogLogLog \cite{Karppa2022} takes this strategy to the extreme with 3-bit registers and achieves a space saving of around 40\% at the expense of an insert operation which, except for very large numbers, has been reported to be on average more than an order of magnitude slower compared to \ac{HLL} \cite{Karppa2022}.

Interestingly, lossless compression of \ac{PCSA} \cite{Flajolet1985}, a less-space efficient (when uncompressed) predecessor of \ac{HLL} also known as FM-sketch, yields a smaller \ac{MVP} than compression of \ac{HLL} \cite{Scheuermann2007, Lang2017}. The \ac{CPC} sketch as part of the Apache Data Sketches library \cite{ApacheDataSketches} uses this finding. The serialized representation of the \ac{CPC} sketch achieves a \ac{MVP} of around 2.31 \cite{SketchesFeatureMatrix} that is already quite close to the conjectured lower bound of 1.98 \cite{Pettie2021}. However, the need of bulked updates to achieve amortized constant-time insertions more than doubles the memory footprint which also makes serialization significantly slower than for the original \ac{HLL}. Similar to all compressed variants of \ac{HLL}, the insert operation of the \ac{CPC} sketch takes time proportional to the sketch size in the worst case.

\begin{table}[t]
  \caption{Notation}
  \label{tab:notation}
  \scriptsize
  \begin{tabular*}{\linewidth}{@{\extracolsep{\fill}} ll}
    \toprule
    Symbol
    &
    Comment
    \\
    \midrule
    $\symCardinality$
    &
    distinct count
    \\
    $\symCardinalityEstimator$
    &
    distinct count estimate
    \\
    $\symBase$
    &
    base, $\symBase > 1$, defines distribution of update values, compare \eqref{equ:geometric}
    \\
    $\symPrecision$
    &
    precision parameter
    \\
    $\symNumReg$
    &
    number of registers, $\symNumReg = 2^\symPrecision$
    \\
    $\symMaxUpdateValMax$
    &
    maximum possible update value, cf. \cref{sec:stat_model}
    \\
    $\symBitsForMax$
    &
    number of bits used for storing the maximum occurred update values, $2^\symBitsForMax > \symMaxUpdateValMax$
    \\
    $\symNumExtraBits$
    &
    number of additional register bits to indicate updates with smaller values
    \\
    $\symRegister_\symRegAddr$
    &
    value of $\symRegAddr$-th register , $0\leq \symRegAddr < \symNumReg$, $0\leq \symRegister_\symRegAddr < (\symMaxUpdateValMax+1) 2^{\symNumExtraBits}$
    \\
    $\symMaxUpdateVal_\symRegAddr$
    &
    maximum occurred update value for $\symRegAddr$-th register, $\symMaxUpdateVal_\symRegAddr = \lfloor\symRegister_\symRegAddr/2^\symNumExtraBits\rfloor$, $0 \leq \symMaxUpdateVal_\symRegAddr \leq \symMaxUpdateValMax$
    \\
    $\symCount_\symIndexJ$
    &
    number of registers with value $\symIndexJ$, $\symCount_{\symIndexJ} := |\lbrace \symRegAddr\vert \symRegister_\symRegAddr = \symIndexJ \rbrace|$
    \\
    $\symDensityRegister$
    &
    \acf*{PMF} of register values, see \eqref{equ:register_pmf}
    \\
    $\symDensityRegisterSimple$
    &
    approximated \acs*{PMF} of register values, see \eqref{equ:simple_register_pmf}
    \\
    $\symLikelihood$
    &
    likelihood function, $\symLikelihood = \symLikelihood(\symCardinality\vert\symRegister_0\ldots \symRegister_{\symNumReg-1})$
    \\
    $\symFisher$
    &
    Fisher information, $\symFisher = \symExpectation(-\partial^2 /\partial \symCardinality^2\ln\symLikelihood)$, see \eqref{equ:fisher}
    \\
    $\symShannon$
    &
    Shannon entropy, $\symShannon=\symExpectation(-\log_2 \symLikelihood)$, see \eqref{equ:shannon_entropy}
    \\
    $\symGRA$
    &
    free parameter of \acf*{GRA} estimators
    \\
    $\symZ_\symUpdateVal$
    &
    short notation for $\exp(-\symCardinality(\symBase-1)/(\symNumReg\symBase^{\symUpdateVal}))$, compare \eqref{equ:prob_event_a}
    \\
    $\symGammaFunc$
    &
    gamma function, $\symGammaFunc(\symX):=\int_0^\infty \symY^{\symX-1}e^{-\symY}d\symY$
    \\
    $\symZetaFunc$
    &
    Hurvitz zeta function, $\symZetaFunc(\symX, \symY) := \sum_{\symMaxUpdateVal=0}^\infty (\symMaxUpdateVal + \symY)^{-\symX}= \frac{1}{\symGammaFunc(\symX)}\int_0^\infty \frac{\symZ^{\symX-1}e^{-\symY\symZ}}{1-e^{-\symZ}}d\symZ$
    \\
    $\symStateChangeProbability$
    &
    state change probability
    \\
    $\langle\ldots\rangle_2$
    &
    binary representation, e.g. $\langle 110\rangle_2 = 6$
    \\
    $\lfloor\ldots\rfloor$
    &
    floor function, e.g. $\lfloor 3.7\rfloor = 3$
    \\
    \bottomrule
  \end{tabular*}
\end{table}

Lossy compression of \ac{HLL} has also been proposed \cite{Xiao2020, Xiao2017}. However, like other approaches such as HyperBitBit/HyperBit \cite{Sedgewick2022}, they trade idempotency for less space and are therefore risky to use \cite{Pettie2021}. In contrast, sacrificing mergeability for less space is of greater practical interest \cite{Chen2011, Helmi2012,Lu2023}. When the data is not distributed and a merge operation is not actually needed, the \ac{MVP} of \ac{HLL} can be reduced by 36\% down to $4.16$ by martingale also known as \ac{HIP} estimation \cite{Ting2014, Cohen2015}. The theoretical limit of \ac{MVP} for non-mergeable sketches is at most 1.63 \cite{Pettie2021a}, and is also nearly reached by the serialized representation of the \ac{CPC} sketch \cite{SketchesFeatureMatrix}. Data structures with in-memory representations that can be updated in constant time in the worst case are \ac{HLL} with vectorized counters \cite{Bruschi2021} and the martingale curtain sketch which achieves a \ac{MVP} of 2.31 \cite{Pettie2021a}. As the insertions are not commutative, all these non-mergeable approaches also do not support reproducibility.

The only mergeable data structure we know of that is more space-efficient than \ac{HLL} while having essentially the same properties, in particular constant-time worst-case updates, is \ac{EHLL} \cite{Ohayon2021}. It extends the \ac{HLL} registers from 6 to 7 bits to store not only the maximum update value, but also whether there was an update with a value smaller by one. This additional information can be used to obtain more accurate estimates. In particular, the \ac{MVP} is reduced by 16\% to 5.43. The only thing missing to make it really practical is an estimator for small distinct counts. As with the original \ac{HLL} \cite{Flajolet2007}, it was only proposed to switch to the linear probabilistic count estimator \cite{Whang1990}. However, this is problematic because the estimation error in the transition region can be large \cite{Heule2013, Ertl2017}. Nevertheless, \ac{EHLL} motivated us to generalize its basic idea by extending \ac{HLL} registers even further.

\subsection{Summary of Contributions}
We first describe a data structure that generalizes the known data structures \acf{HLL} \cite{Flajolet2007}, \acf{EHLL} \cite{Ohayon2021}, and \acf{PCSA} \cite{Flajolet1985}. We derive analytic expressions for the Fisher information and the Shannon entropy as functions of the data structure parameters. Although these expressions reveal even more space-efficient configurations that could be the subject of future research, the focus of this work is on a setting that leads to a very practical data sketch called \acf{ULL} with a \ac{MVP} of 4.63 which is 28\% below that of \ac{HLL}. Since the Shannon entropy is also 24\% smaller for the same estimation error, \ac{ULL} is also more compact when using lossless compression. Moreover, our experimental results indicate that standard compression algorithms additionally benefit from the 8-bit register size of \ac{ULL}.

To extract all the information contained in the \ac{ULL} sketch, we applied the \ac{ML} method that achieves an estimation error as theoretically predicted by the Cram\'er-Rao bound \cite{Casella2002}. Alternatively, we present a faster approach based on a further generalization of the recently proposed \acf{GRA} estimator \cite{Wang2023}. Even though our theoretical analysis shows a smaller estimation efficiency than the \ac{ML} estimator, the \ac{MVP} with a value of 4.94 still corresponds to a 24\% space reduction compared to \ac{HLL}. As the \ac{GRA} estimator, the basic version of our new estimator works only for distinct counts that are neither too small nor too large. Therefore, we developed two additional estimators specific to those ranges that are also easy to evaluate. Using a novel approach, they are seamlessly combined with the basic estimator to cover the full range of distinct count values. We also analyzed martingale estimation, which can be used for non-distributed data, and found that in this case \ac{ULL} reduces the \ac{MVP} by 17\% to 3.47 compared to \ac{HLL}.

All theoretically derived estimators were verified by intensive simulations, which all show perfect agreement with the theoretically predicted estimation errors. In particular, we use a technique that allows verification for distinct counts on the order of $2^{64}$, which is not possible using traditional simulations. Finally, we also present the results of speed benchmarks, which show that \ac{ULL} is similarly fast as \ac{HLL}.

An implementation of \ac{ULL} is publicly available as part of the Hash4j open-source Java library at \url{https://github.com/dynatrace-oss/hash4j}.
Detailed instructions together with the necessary source code to reproduce all presented results and figures can be found at \url{https://github.com/dynatrace-research/ultraloglog-paper}. A version of this paper extended by an appendix with mathematical derivations and proofs is also available \cite{Ertl2023}.

\section{Generalized Data Structure}
\label{sec:data_structure}

We start by introducing a data structure for approximate distinct counting that generalizes \acf{HLL} \cite{Flajolet2007}, \acf{EHLL} \cite{Ohayon2021}, and \acf{PCSA} \cite{Flajolet1985}. As those, it consists of $\symNumReg$ registers which are initially set to zero. For every added element a uniformly distributed hash value is computed. This hash value is used to extract a uniform random register index $\symRegAddr\in[0, \symNumReg)$ and some geometrically distributed integer value with \ac{PMF}
\begin{equation}
  \label{equ:geometric}
  \symDensityUpdate(\symUpdateVal) = (\symBase - 1)\symBase^{-\symUpdateVal} \qquad \symUpdateVal \geq 1, \symBase > 1,
\end{equation}
parameterized by the base parameter $\symBase$, that is used to update the $\symRegAddr$-th register.
Each register consists of $\symBitsForMax + \symNumExtraBits$ bits. $\symBitsForMax$ bits are used to store the maximum update value $\symMaxUpdateVal_\symRegAddr$ seen so far. The remaining $\symNumExtraBits$ bits indicate whether there have been any updates with values $\symMaxUpdateVal_\symRegAddr-1,\ldots,\symMaxUpdateVal_\symRegAddr-\symNumExtraBits$, respectively. Obviously, this update procedure is idempotent meaning that further occurrences of the same elements will never change any register state. As a consequence, the final state of this data structure can be used to estimate the number of inserted distinct elements $\symCardinality$.

Every register state can be described by an integer value $\symRegister_\symRegAddr$ with $0\leq \symRegister_\symRegAddr <2^{\symBitsForMax + \symNumExtraBits}$. We assume that the most significant $\symBitsForMax$ bits of $\symRegister_\symRegAddr$ are used to store $\symMaxUpdateVal_\symRegAddr$, which can therefore be simply obtained by $\symMaxUpdateVal_\symRegAddr = \lfloor\symRegister_\symRegAddr/2^\symNumExtraBits\rfloor$. As example for $\symBitsForMax=6$ and $\symNumExtraBits=2$, $\symRegister_\symRegAddr=\langle00011010\rangle_2$ would mean that the largest update value was $\symMaxUpdateVal_\symRegAddr = \langle000110\rangle_2 = 6$. The right-most $\symNumExtraBits$ bits indicate that the register was also already updated with a value of $5$ but not yet with $4$.
If the register gets further updated with value 8, the state would become $\symRegister_\symRegAddr=\langle00100001\rangle_2$ where the first $\symBitsForMax=6$ bits encode $\symMaxUpdateVal_\symRegAddr = \langle001000\rangle_2 = 8$ and the left-most $\symNumExtraBits=2$ bits indicate that there was no update with value 7 but one with 6. Information about smaller update values is lost.
If $\symMaxUpdateVal_\symRegAddr\leq \symNumExtraBits$, there are only $\symMaxUpdateVal_\symRegAddr-1$ smaller update values and therefore only $\symMaxUpdateVal_\symRegAddr-1$ of the $\symNumExtraBits$ extra bits are relevant and some values like $\symRegister_\symRegAddr=\langle00001001\rangle_2$ for $\symBitsForMax=6$ and $\symNumExtraBits=2$ cannot be attained. Enumerating just possible states would lead to a slightly more compact encoding. However, for the sake of simplicity and also to avoid special cases, we refrain from this small improvement.

In practice, the number of registers $\symNumReg$ is usually some power of 2, $\symNumReg = 2^\symPrecision$ with $\symPrecision$ being the precision parameter. In this way, a uniform random register index can be chosen by just taking $\symPrecision$ bits from the hash value. Furthermore, the parameter $\symBase$ is often 2 such that the update value $\symUpdateVal$ can be easily obtained from the \acf{NLZ} of the remaining hash bits and therefore usually requires just a single CPU instruction. Obviously, the cases $\symBase=2$, $\symNumExtraBits=0$ and $\symBase=2$, $\symNumExtraBits=1$ correspond to \ac{HLL} and \ac{EHLL}, respectively. Furthermore, since \ac{PCSA} effectively keeps track of any update values, the stored information corresponds to our generalized data structure with $\symBase=2$, $\symNumExtraBits\rightarrow\infty$. (\ac{PCSA} typically uses just 64 bits for each register which is sufficient as update values greater than 64 are unlikely for real-world distinct counts.) Another advantage of choosing $\symNumReg$ as a power of 2 and $\symBase=2$ is that the data structure can be implemented in such a way that it can later be reduced to a smaller precision parameter \cite{Ertl2017}. The result will then be identical as if the smaller precision parameter was chosen from the beginning. This is important for migration scenarios where precision needs to be changed in a way that is compatible and mergeable with historical data.

\subsection{Statistical Model}
\label{sec:stat_model}

For simplification, we use the common Poisson approximation \cite{Flajolet2007, Ertl2017, Wang2023} that the number of inserted distinct elements is not fixed, but follows a Poisson distribution with mean $\symCardinality$. As a consequence, since updates are evenly distributed over all registers,
the number of updates with value $\symUpdateVal$ per register is again Poisson distributed with mean $\symCardinality \symDensityUpdate(\symUpdateVal)/\symNumReg=\symCardinality(\symBase-1)/(\symNumReg\symBase^{\symUpdateVal})$.
The probability that a register was updated with value $\symUpdateVal$ at least once, denoted by event $\symUpdateEvent_\symUpdateVal$, is therefore
\begin{equation}
  \label{equ:prob_event_a}
  \textstyle
  \symProbability(\symUpdateEvent_\symUpdateVal) =
  1 - \symZ_\symUpdateVal\quad \text{with $\symZ_\symUpdateVal := \exp(-\symCardinality(\symBase-1)/(\symNumReg\symBase^{\symUpdateVal}))$}.
\end{equation}
The probability that $\symMaxUpdateVal$ was the largest update value, which implies that there were no updates with values greater than $\symMaxUpdateVal$, is given by $\symProbability( \symUpdateEvent_\symMaxUpdateVal\wedge \bigwedge_{\symUpdateVal=\symMaxUpdateVal+1}^\infty \overline\symUpdateEvent_{\symUpdateVal}) = (1-\symZ_\symMaxUpdateVal) \prod_{\symUpdateVal=\symMaxUpdateVal+1}^\infty\symZ_{\symUpdateVal} = \symZ_{\symMaxUpdateVal}^{\frac{1}{\symBase-1}}
  (1-\symZ_{\symMaxUpdateVal})
$.
The Poisson approximation results in registers that are independent and identically distributed. For the generalized data structure the corresponding \acf{PMF} can be written as
\begin{align}
  \label{equ:register_pmf}
   & \symDensityRegister(\symRegister\vert\symCardinality) = \symProbability(\symRegister_\symRegAddr = \symRegister) =
  \\
   & \begin{cases}
       \scriptstyle\symZ_0^{\frac{1}{\symBase-1}}
        &
       \scriptstyle\symRegister = 0,
       \\
       \scriptstyle\symZ_{\symMaxUpdateVal}^{\frac{1}{\symBase-1}}
       (1-\symZ_{\symMaxUpdateVal})
       \prod_{\symIndexJ = 1}^{\symMaxUpdateVal-1}
       \symZ_{\symMaxUpdateVal-\symIndexJ}^{1-\symIndexBit_\symIndexJ}
       (1- \symZ_{\symMaxUpdateVal-\symIndexJ})^{ \symIndexBit_\symIndexJ}
        &
       \scriptstyle\symRegister = \symMaxUpdateVal 2^\symNumExtraBits + \langle\symIndexBit_1\ldots\symIndexBit_{\symMaxUpdateVal-1}\rangle_2 2^{\symNumExtraBits + 1- \symMaxUpdateVal},\ 1 \leq \symMaxUpdateVal\leq \symNumExtraBits,
       \\
       \scriptstyle\symZ_{\symMaxUpdateVal}^{\frac{1}{\symBase-1}}
       (1-\symZ_{\symMaxUpdateVal})
       \prod_{\symIndexJ = 1}^\symNumExtraBits
       \symZ_{\symMaxUpdateVal-\symIndexJ}^{1-\symIndexBit_\symIndexJ}
       (1- \symZ_{\symMaxUpdateVal-\symIndexJ})^{ \symIndexBit_\symIndexJ}
        &
       \scriptstyle\symRegister = \symMaxUpdateVal 2^\symNumExtraBits + \langle\symIndexBit_1\ldots\symIndexBit_{\symNumExtraBits}\rangle_2,\ \symNumExtraBits+1 \leq \symMaxUpdateVal< \symMaxUpdateValMax,
       \\
       \scriptstyle(1-\symZ^{\frac{1}{\symBase-1}}_{\symMaxUpdateValMax-1})
       \prod_{\symIndexJ = 1}^\symNumExtraBits
       \symZ_{\symMaxUpdateValMax-\symIndexJ}^{1-\symIndexBit_\symIndexJ}
       (1- \symZ_{\symMaxUpdateValMax-\symIndexJ})^{ \symIndexBit_\symIndexJ}
        &
       \scriptstyle\symRegister = \symMaxUpdateValMax 2^\symNumExtraBits + \langle\symIndexBit_1\ldots\symIndexBit_{\symNumExtraBits}\rangle_2,
       \\
       \scriptstyle 0
        &
       \scriptstyle \text{else.}
     \end{cases}\nonumber
\end{align}
This formula takes into account that update values are limited to the range $[1, \symMaxUpdateValMax]$.
The upper limit is on the one hand a consequence of the number of register bits reserved for storing the maximum update value. If $\symBitsForMax$ bits are used, the update values must be truncated at $2^\symBitsForMax - 1$ because higher update values cannot be stored. On the other hand, the way the geometrically distributed integer values are typically determined also leads to an upper limit.
For example, the update values in \cref{alg:insertion_hll} do not exceed $65 - \symPrecision$, which results from extracting the update value and the register index from a single 64-bit hash value.

For our theoretical analysis, we consider a simplified model. We assume that registers are initially set to $-\infty$ and that there are no restrictions on the update values. In particular, there are also updates with non-positive values $\symUpdateVal\leq 0$ with corresponding (virtual) events $\symUpdateEvent_\symUpdateVal$ occurring with probabilities according to \eqref{equ:prob_event_a}. Then the \ac{PMF} for a register simplifies to
\begin{multline}
  \label{equ:simple_register_pmf}
  \symDensityRegisterSimple(\symRegister\vert\symCardinality) =
  \symProbability(\symRegister_\symRegAddr = \symRegister)
  =
  \symZ_{\symMaxUpdateVal}^{\frac{1}{\symBase-1}}
  (1-\symZ_{\symMaxUpdateVal})
  \prod_{\symIndexJ = 1}^\symNumExtraBits
  \symZ_{\symMaxUpdateVal-\symIndexJ}^{1-\symIndexBit_\symIndexJ}
  (1- \symZ_{\symMaxUpdateVal-\symIndexJ})^{ \symIndexBit_\symIndexJ}\\
  \text{with}\ \symRegister = \symMaxUpdateVal2^\symNumExtraBits + \langle\symIndexBit_1\ldots\symIndexBit_{\symNumExtraBits}\rangle_2.
\end{multline}
If all register values are in the range $[(\symNumExtraBits+1) 2^\symNumExtraBits, \symMaxUpdateValMax 2^\symNumExtraBits)$, which is usually the case for not too small and not too large distinct counts, \eqref{equ:register_pmf} and \eqref{equ:simple_register_pmf} are equivalent. As a consequence, the following theoretical results will also hold for \eqref{equ:register_pmf} when the distinct count is in the intermediate range.

\subsection{Theoretical Analysis}

Given the register states $\symRegister_0,\ldots,\symRegister_{\symNumReg-1}$, the log-likelihood function can be expressed as
\begin{equation}
  \label{equ:log_likelihood}
  \textstyle
  \ln\symLikelihood =
  \ln\symLikelihood(\symCardinality\vert \symRegister_0,\ldots,\symRegister_{\symNumReg-1})
  =
  \sum_{\symRegAddr=0}^{\symNumReg-1} \ln \symDensityRegister(\symRegister_\symRegAddr\vert\symCardinality).
\end{equation}
When assuming the simplified \ac{PMF} \eqref{equ:simple_register_pmf},
the Fisher information can be written as (see \ifextended \cref{lem:fisher} in \cref{app:proofs}\else extended paper \cite{Ertl2023}\fi)
\begin{equation}
  \label{equ:fisher}
  \symFisher = \symExpectation\!\left(-\frac{\partial^2 }{\partial \symCardinality^2}\ln\symLikelihood\right)
  \approx
  \frac{\symNumReg}{\symCardinality^2}
  \frac{1}{\ln \symBase}\symZetaFunc\!\left(2, 1 + \frac{\symBase^{-\symNumExtraBits}}{\symBase-1}\right)
\end{equation}
using the Hurvitz zeta function $\symZetaFunc$ as defined in \cref{tab:notation}.
The approximation is based on the fact that the Fisher information is a periodic function of $\log_\symBase \symCardinality$ with period 1 and tiny relative amplitude that can be ignored in practice \ifextended (see \cref{lem:approximation} in \cref{app:proofs})\else as shown in the extended paper \cite{Ertl2023}\fi.
Formula \eqref{equ:fisher} matches the results reported for the special cases of generalized \ac{HLL} ($\symNumExtraBits = 0$) and generalized \ac{PCSA} ($\symNumExtraBits \rightarrow\infty$) \cite{Pettie2021}.

According to Cram\'er-Rao \cite{Casella2002} the reciprocal of the Fisher information $\symFisher$ is a lower bound for the variance of any unbiased estimator. Our experiments will show that for a sufficiently large number of registers $\symNumReg$, this lower bound can be actually reached using \acf{ML} estimation. The corresponding asymptotic \ac{MVP} is given by
\begin{equation}
  \label{equ:mvp_uncompressed}
  \symMVP = \symNumReg(\symBitsForMax + \symNumExtraBits) \symVariance\!\left(\frac{\symCardinalityEstimator}{\symCardinality}\right)
  =
  \frac{ \symNumReg(\symBitsForMax + \symNumExtraBits)}{\symFisher \symCardinality^2}
  \approx
  \frac{(\symBitsForMax + \symNumExtraBits)\ln \symBase}{\symZetaFunc(2, 1 + \frac{\symBase^{-\symNumExtraBits}}{\symBase-1})}.
\end{equation}
$(\symBitsForMax + \symNumExtraBits)$ is the number of bits used for a single register. $\symBitsForMax$ bits are used for storing the maximum update value and $\symNumExtraBits$ is the number of extra bits as already described before.
\begin{figure}[t]
  \centering
  \includegraphics[width=\linewidth]{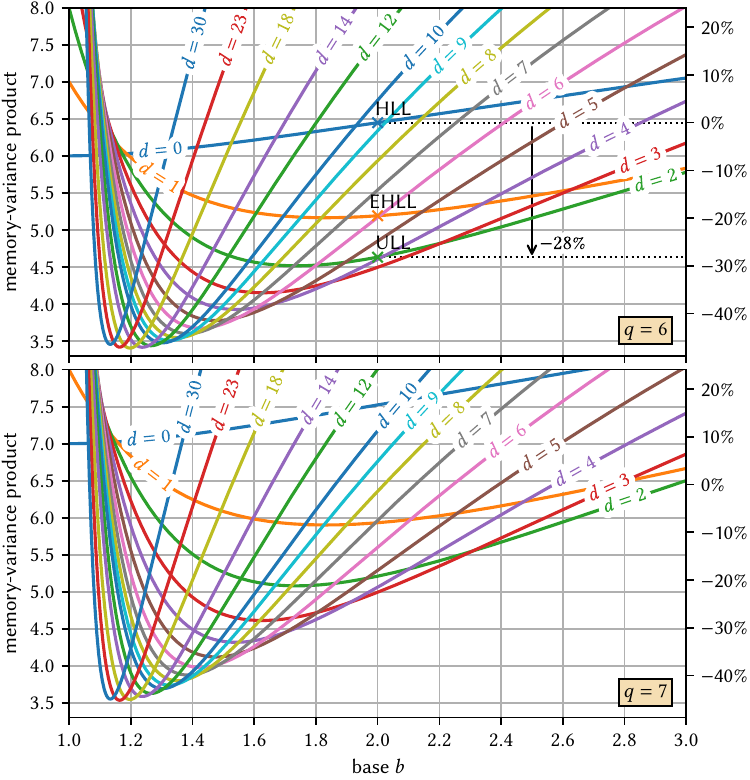}
  \caption{\boldmath The theoretical asymptotic \acf*{MVP} \eqref{equ:mvp_uncompressed} over the base $\symBase$ for $\symBitsForMax=6$ and $\symBitsForMax=7$ and various values of $\symNumExtraBits$ when assuming a memory footprint of $\symNumReg(\symBitsForMax+\symNumExtraBits)$ bits. The top chart shows the 28\% improvement of \acf*{ULL} over \acf*{HLL}.}
  \label{fig:mvp_lower_bound}
\end{figure}
The \ac{MVP} \eqref{equ:mvp_uncompressed} is plotted in \cref{fig:mvp_lower_bound} over the base $\symBase$ for $\symBitsForMax=6$ and $\symBitsForMax=7$ and various values of $\symNumExtraBits$ and allows comparing the memory-efficiencies of different configurations, when assuming that the state takes overall $\symNumReg(\symBitsForMax+\symNumExtraBits)$ bits and is not further compressed.

It is also interesting to study the case where the state is ideally compressed, which means that the number of bits needed to store the state is given by its Shannon entropy.
Considering the resulting \ac{MVP}, also called \ac{FISH} number \cite{Pettie2021}, again allows a better comparison.
Using the simplified \ac{PMF} \eqref{equ:simple_register_pmf}, the Shannon entropy can be approximated by (see \ifextended \cref{lem:shannon} in \cref{app:proofs}\else extended paper \cite{Ertl2023}\fi)
\begin{align}
  \label{equ:shannon_entropy}
  \symShannon
   & =
  \symExpectation(-\log_2\symLikelihood)
  \nonumber
  \\
   & \approx
  \textstyle
  \frac{\symNumReg}{(\ln 2) (\ln \symBase)}
  \left(
  \left(1+\frac{\symBase^{-\symNumExtraBits}}{\symBase-1}\right)^{\!-1}+
  \int_{0}^1 \symZ^{\frac{\symBase^{-\symNumExtraBits}}{\symBase-1}}
  \frac{
    (1-\symZ)
    \ln(1-\symZ)
  }{\symZ\ln \symZ}
  d\symZ
  \right).
\end{align}
Combining \eqref{equ:fisher} and \eqref{equ:shannon_entropy} gives for the \ac{MVP} under the assumption of an efficient estimator and optimal compression
\begin{equation}
  \label{equ:mvp_compressed}
  \symMVP = \symShannon\symVariance\!\left(\frac{\symCardinalityEstimator}{\symCardinality}\right) = \frac{\symShannon}{\symFisher \symCardinality^2}
  \approx
  \textstyle
  \frac{
    (1+\frac{\symBase^{-\symNumExtraBits}}{\symBase-1})^{-1}+
    \int_{0}^1 \symZ^{\frac{\symBase^{-\symNumExtraBits}}{\symBase-1}}
    \frac{
      (1-\symZ)
      \ln(1-\symZ)
    }{\symZ\ln \symZ}
    d\symZ
  }{
    \symZetaFunc(2, 1 + \frac{\symBase^{-\symNumExtraBits}}{\symBase-1})\ln 2 }
  .
\end{equation}
This function is plotted in \cref{fig:mvp_compressed} over the base $\symBase$ for various values of $\symNumExtraBits$. The results agree with the conjecture that the \ac{MVP} of sketches supporting mergeability and reproducibility is fundamentally bounded by \cite{Pettie2021}
\begin{equation*}
  \lim_{\frac{\symBase^{-\symNumExtraBits}}{\symBase-1}\rightarrow 0} \frac{\symShannon}{\symFisher \symCardinality^2}
  \approx
  \frac{1+\int_{0}^1
    \frac{
      (1-\symZ)
      \ln(1-\symZ)
    }{\symZ\ln \symZ}
    d\symZ}{\symZetaFunc(2, 1)\ln 2} \approx 1.98.
\end{equation*}

\subsection{Choice of Parameters}
\label{sec:parameter-choice}

The data structure proposed in \cref{sec:data_structure} has four parameters, the base $\symBase$, the number of registers $\symNumReg=2^\symPrecision$, the number of register bits $\symBitsForMax$ reserved for storing the maximum update value, and the number of additional bits $\symNumExtraBits$ to indicate the occurrences of the $\symNumExtraBits$ next smaller update values relative to the maximum. The previous theoretical results allow us to find parameters that lead to a small \ac{MVP}. Here we can essentially leave $\symNumReg$ aside since it can be used to define the accuracy/space tradeoff, but has essentially no effect on the \ac{MVP} according to \eqref{equ:mvp_uncompressed} and \eqref{equ:mvp_compressed}.

The parameters $\symBase$ and $\symBitsForMax$ define the operating range of the data structure and must be chosen such that it is very unlikely that all registers get saturated. In other words, the fraction of registers with
$\symRegister_\symRegAddr \geq \symMaxUpdateValMax 2^\symNumExtraBits$ must be small. This requires according to \eqref{equ:register_pmf} that $\symProbability(\symRegister_\symRegAddr < \symMaxUpdateValMax 2^\symNumExtraBits) = \symZ^{\frac{1}{\symBase-1}}_{\symMaxUpdateValMax-1} = \exp(-\frac{\symCardinality}{\symNumReg\symBase^{\symMaxUpdateValMax-1}})\approx 1$ or $\symCardinality \ll \symNumReg \symBase^{\symMaxUpdateValMax-1}$.
Hence, the maximum supported distinct count $\symCardinalityMax$ can be roughly estimated by $\symCardinalityMax \approx \symNumReg \symBase^{\symMaxUpdateValMax-1}$.
If $\symMaxUpdateValMax$ is limited by the register size ($\symMaxUpdateValMax = 2^\symBitsForMax - 1$, cf. \cref{sec:stat_model}), we get $\symCardinalityMax \approx \symNumReg \symBase^{2^\symBitsForMax - 2}$. For \ac{HLL} with $\symNumReg=256$ registers, $\symBase=2$, and a registers size of $\symBitsForMax=5$ bits, as originally proposed, we have $\symCardinalityMax \approx 275\ \text{billions}$ which might not be sufficient in all situations. Therefore, most \ac{HLL} implementations use nowadays $\symBitsForMax=6$ \cite{Heule2013}, which is definitely sufficient for any realistic counts. Some implementations even use a whole byte per register ($\symBitsForMax=8$), but mainly to have more convenient register access in memory \cite{Snowflake}.

Bearing in mind the operating range defined by $\symNumReg$, $\symBase$, and $\symBitsForMax$, we explore the theoretical \ac{MVP} \eqref{equ:mvp_uncompressed} for different configurations as shown in \cref{fig:mvp_lower_bound}. First we consider the case $\symBase=2$ which also covers \ac{HLL} ($\symBase=2$, $\symNumExtraBits=0$) and \ac{EHLL} ($\symBase=2$, $\symNumExtraBits=1$) and is of particular practical interest due to the very fast mapping of hash values to update values as discussed in \cref{sec:data_structure}. For $\symBase=2$, we need to choose $\symBitsForMax=6$, if we want to make sure that any realistic distinct counts can be handled. According to \cref{fig:mvp_lower_bound} the optimal choice for $\symBase=2$ would be $\symNumExtraBits=3$ resulting in a theoretical \ac{MVP} of $4.4940$.
Despite its slightly larger theoretical \ac{MVP} of $4.6313$, the case $\symNumExtraBits=2$ is very attractive from a practical point of view since a register takes $\symBitsForMax + \symNumExtraBits=8$ bits and fits perfectly into a single byte.
This enables fast updates as registers can be directly accessed and modified when stored in a byte array. This is why we picked $\symBase=2$, $\symNumExtraBits=2$, $\symBitsForMax=6$ as configuration for our new data structure called \acf{ULL}. Its theoretical \ac{MVP} is 28\% smaller than that of \ac{HLL} with a register size of $\symBitsForMax=6$ and a \ac{MVP} of $6.4485$. This even corresponds to a 46\% improvement, when compared to \ac{HLL} implementations that use $\symBitsForMax=8$ bits per register \cite{Snowflake} leading to $\symMVP\approx 8.5981$.

\cref{fig:mvp_lower_bound} shows that the \ac{MVP} could be further reduced by choosing bases $\symBase$ other than 2. For $\symBitsForMax=6$, a minimum \ac{MVP} of $3.4030$ is achieved for $\symBase\approx 1.1976$ and $\symNumExtraBits=18$. However, a small $\symBase$ leads to a small operating range. For example, $\symBase=1.1976$ yields $\symCardinalityMax\approx 18\,\text{millions}$ for $\symNumReg=256$ which is too small for many applications. However, the working range can be extended again by increasing $\symBitsForMax$. For $\symBitsForMax=7$ the minimum $\symMVP \approx 3.5338$ is obtained for $\symNumExtraBits=23$ and $\symBase\approx1.1642$ leading to $\symCardinalityMax\approx 53\,\text{billions}$ for $\symNumReg=256$.
From a practical point of view the configuration $\symBitsForMax=7$, $\symNumExtraBits=9$, $\symBase=\sqrt{2}$ with $\symMVP\approx3.9025$ could be an interesting choice for future research.
The operating range would be similar to \ac{ULL} or \ac{HLL}, registers would take exactly two bytes ($\symBitsForMax + \symNumExtraBits = 16$), and the mapping of hash values to update values for $\symBase=\sqrt{2}$ according to \eqref{equ:geometric} could be accomplished without expensive logarithm evaluations or table lookups. The update value can be obtained by doubling the \acf{NLZ} and adding 0 or 1 depending on the remaining hash bits, whose value range is divided into two parts in the ratio $1:\sqrt{2}$. However, a single 64-bit hash value might not be sufficient in this case.

\begin{figure}[t]
  \centering
  \includegraphics[width=\linewidth]{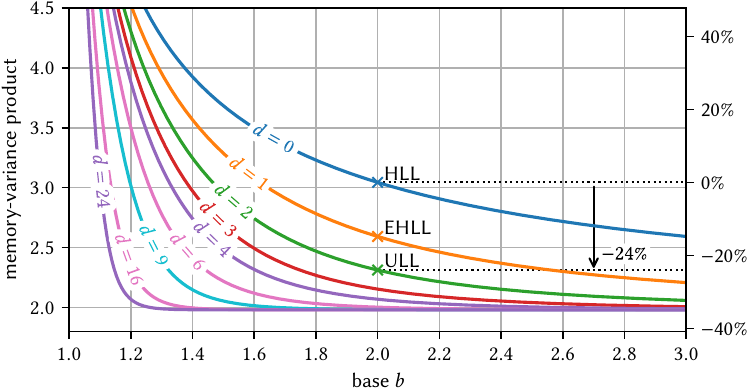}
  \caption{\boldmath The theoretical asymptotic \acf*{MVP} \eqref{equ:mvp_compressed} over the base $\symBase$ for various values of $\symNumExtraBits$ under the assumption of optimal lossless compression. The \acs*{MVP} of \acf*{ULL} is 24\% smaller than that of \acf*{HLL}.}
  \label{fig:mvp_compressed}
\end{figure}

\cref{fig:mvp_compressed} shows the \ac{MVP} under the assumption of optimal lossless compression as given by \eqref{equ:mvp_compressed}. The \ac{ULL} configuration with $\symNumExtraBits=2$ leads to $\symMVP\approx2.3122$ which is 24\% less than for \ac{HLL} with $\symNumExtraBits=0$ and $\symMVP\approx3.0437$.
Therefore, \ac{ULL} is potentially superior to any variant of \ac{HLL} that uses lossless compression techniques \cite{ApacheDataSketches, Karppa2022}, if similar techniques are also applied to \ac{ULL}.

\section{Distinct Count Estimation}

The theoretical results show clearly that \acf{ULL} as a special case of the presented generalized data structure with $\symBase=2$, $\symNumExtraBits=2$ and $\symBitsForMax=6$ encodes distinct count information, uncompressed and compressed, more efficiently than \ac{HLL} with $\symBase=2$, $\symNumExtraBits=0$ and $\symBitsForMax=6$.
However, an open question is whether this information can also be readily accessed using a simple and robust estimation procedure, which will be the focus of the following sections.

\subsection{Maximum-Likelihood Estimator}
Since we know the \acf{PMF} \eqref{equ:register_pmf} when using the Poisson approximation, we can simply use the \acf{ML} method. For $\symBase=2$, the log-likelihood function \eqref{equ:log_likelihood} is shaped like
\begin{equation*}
  \textstyle
  \ln\symLikelihood
  =-\frac{\symCardinality}{\symNumReg} \symLikelihoodFuncExponentOne+
  \sum_{\symMaxUpdateVal=1}^{\symMaxUpdateValMax-1}\symLikelihoodFuncExponentTwo_\symMaxUpdateVal \ln(1 - e^{-\frac{\symCardinality}{\symNumReg 2^\symMaxUpdateVal}} ),
\end{equation*}
where for \ac{ULL} with $\symNumExtraBits=2$ the coefficients $\symLikelihoodFuncExponentOne$ and $\symLikelihoodFuncExponentTwo_\symMaxUpdateVal$ are given by
\begin{align*}
   &
  \scriptstyle\symLikelihoodFuncExponentOne
  \ =\ \symCount_0 + \frac{\symCount_4}{2}+\frac{3 \symCount_8 + \symCount_{10}}{4} + \left(\sum_{\symMaxUpdateVal=3}^{\symMaxUpdateValMax-1} \frac{7\symCount_{4\symMaxUpdateVal}+3\symCount_{4\symMaxUpdateVal + 1}+5\symCount_{4\symMaxUpdateVal + 2}+\symCount_{4\symMaxUpdateVal + 3}}{2^\symMaxUpdateVal}\right)
  +\frac{3\symCount_{4\symMaxUpdateValMax}+\symCount_{4\symMaxUpdateValMax+1}+2\symCount_{4\symMaxUpdateValMax+2}}{2^{\symMaxUpdateValMax-1}},
  \\
   &
  \scriptstyle\symLikelihoodFuncExponentTwo_1
  \ =\ \symCount_4+\symCount_{10}+\symCount_{13}+\symCount_{15},
  \quad
  \symLikelihoodFuncExponentTwo_2
  \ =\ \symCount_8+\symCount_{10}+\symCount_{14}+\symCount_{15}+\symCount_{17}+\symCount_{19},
  \\
   &
  \scriptstyle \symLikelihoodFuncExponentTwo_\symMaxUpdateVal
  \ =\ \symCount_{4\symMaxUpdateVal}+\symCount_{4\symMaxUpdateVal+1}+\symCount_{4\symMaxUpdateVal+2}+\symCount_{4\symMaxUpdateVal+3}+\symCount_{4\symMaxUpdateVal+6}+\symCount_{4\symMaxUpdateVal+7}+\symCount_{4\symMaxUpdateVal+9}+\symCount_{4\symMaxUpdateVal+11}
  \quad \text{for $3\leq\symMaxUpdateVal\leq\symMaxUpdateValMax-2$},
  \\
   &
  \scriptstyle\symLikelihoodFuncExponentTwo_{\symMaxUpdateValMax-1}
  \ =\ \symCount_{4\symMaxUpdateValMax-4}+\symCount_{4\symMaxUpdateValMax-3}+\symCount_{4\symMaxUpdateValMax-2}+\symCount_{4\symMaxUpdateValMax-1}+\symCount_{4\symMaxUpdateValMax}+\symCount_{4\symMaxUpdateValMax+1}+2\symCount_{4\symMaxUpdateValMax+2}+2\symCount_{4\symMaxUpdateValMax+3}.
\end{align*}
$\symCount_{\symIndexJ} := |\lbrace \symRegAddr\vert \symRegister_\symRegAddr = \symIndexJ \rbrace|$ is the number of registers having value $\symIndexJ$.

As the corresponding \ac{ML} equation has the same shape as that for \ac{HLL}, we can reuse the numerically robust solver we developed based on the secant method \cite{Ertl2017}. This algorithm avoids the evaluation of expensive mathematical functions, but is still somewhat costly as the solution needs to be found iteratively.
The \ac{ML} estimate $\symCardinalityEstimatorML$ can be further improved by correcting for the first-order bias \cite{Cox1968}. Applying the correction factor as derived in \ifextended \cref{lem:first_order_ml_bias} in \cref{app:proofs} \else the extended paper \cite{Ertl2023} \fi under the assumption of the simplified \ac{PMF} \eqref{equ:simple_register_pmf} gives for \ac{ULL} with $\symBase=2$ and $\symNumExtraBits=2$
\begin{equation}
  \label{equ:ml_bias_correction}
  \textstyle\symCardinalityEstimator = \symCardinalityEstimatorML \left(1 + \frac{1}{\symNumReg} \frac{3(\ln 2) \symZetaFunc(3,\frac{5}{4})}{2 (\symZetaFunc(2,\frac{5}{4}))^2}\right)^{\!-1}
  \approx \symCardinalityEstimatorML \left(1 + \frac{0.48147}{\symNumReg}\right)^{\!-1}
\end{equation}
where $\symZetaFunc$ denotes again the Hurvitz zeta function.

The \ac{ML} method is known to be asymptotically efficient as $\symNumReg\rightarrow\infty$. The experimental results presented later in \cref{sec:estimation_error} show that the \ac{ML} estimate really matches the theoretically predicted \ac{MVP} \eqref{equ:mvp_uncompressed}.
For \ac{HLL}, estimators have been found that are easier to compute than the \ac{ML} estimator while giving almost equal estimates over the whole value range \cite{Ertl2017}. This was our motivation to look for simpler estimators for \ac{ULL}.

\subsection{\acs*{GRA} Estimator}
Recently, the \acf{GRA} estimator was proposed, which can be easily computed and is more efficient than existing estimators for \acf{PCSA} and \acf{HLL} \cite{Wang2023}. Therefore, we investigated if this estimation approach is also suitable for our generalized data structure and in particular for \ac{ULL}. The basic idea is to sum up $\symBase^{-\symGRA\symUpdateVal}$ with some constant $\symGRA>0$ for all update values $\symUpdateVal$ that we know with certainty, based on the current register value $\symRegister = \symMaxUpdateVal 2^\symNumExtraBits + \langle\symIndexBit_1\ldots\symIndexBit_{\symNumExtraBits}\rangle_2$, could not have occurred previously. The corresponding statistic for our generalized data structure can be expressed as $\sum_{\symUpdateVal=\symMaxUpdateVal-\symNumExtraBits}^{\symMaxUpdateVal-1} (1-\symIndexBit_{\symMaxUpdateVal-\symUpdateVal}) \symBase^{-\symGRA\symUpdateVal} + \sum_{\symUpdateVal=\symMaxUpdateVal+1}^\infty \symBase^{-\symGRA\symUpdateVal} = \symBase^{-\symGRA\symMaxUpdateVal}(\frac{1}{\symBase^\symGRA-1}
  +
  \sum_{\symIndexS=1}^{\symNumExtraBits}
  (1-\symIndexBit_{\symIndexS})
  \symBase^{\symGRA\symIndexS})$.
The analysis of the first two moments of this statistic \ifextended (cf. \cref{lem:expectation_gra_reg_contrib,lem:moment_two_gra_reg_contrib} in \cref{app:proofs}) \fi under the assumption of the simplified \ac{PMF} \eqref{equ:simple_register_pmf}, together with the delta method (see \ifextended \cref{lem:cardinality_estimator} in \cref{app:proofs}\else extended paper \cite{Ertl2023}\fi) yields the distinct count estimator
\begin{equation}
  \label{equ:gra_estimator}
  \textstyle
  \symCardinalityEstimator
  =
  \symNumReg^{1+\frac{1}{\symGRA}}
  \cdot
  \left(\sum_{\symRegAddr=0}^{\symNumReg-1} \symRegContrib(\symRegister_\symRegAddr)
  \right)^{\!-\frac{1}{\symGRA}}
  \cdot
  \left(1 + \frac{1+\symGRA}{2}\frac{\symVarianceFactor}{\symNumReg}\right)^{\!-1}
\end{equation}
where the individual register contributions are defined as
\begin{equation}
  \label{equ:gra_reg_contribution}
  \textstyle
  \symRegContrib(\symRegister)
  :=
  \frac{(\symBase-1+\symBase^{-\symGRA})^{\symGRA}\ln \symBase}{\symGammaFunc(\symGRA)}
  \symBase^{-\symGRA\symMaxUpdateVal}
  \left(\frac{1}{\symBase^\symGRA - 1} + \sum_{\symIndexS=1}^\symNumExtraBits (1 - \symIndexBit_\symIndexS) \symBase^{\symIndexS \symGRA}\right)
\end{equation}
with $\symRegister = \symMaxUpdateVal 2^\symNumExtraBits + \langle\symIndexBit_1\ldots\symIndexBit_{\symNumExtraBits}\rangle_2$ and
\begin{equation*}
  \textstyle
  \symVarianceFactor
  =
  {\textstyle
  \frac{1}{\symGRA^2}}\left({\textstyle\frac{\symGammaFunc(2\symGRA)\ln\symBase}{(\symGammaFunc(\symGRA))^2}}\left(
  1 + {\textstyle\frac{2\symBase^{-\symGRA\symNumExtraBits} }{\symBase^{\symGRA}-1 }}
  +
  \sum_{\symIndexS = 1}^\symNumExtraBits
  \frac{2\symBase^{-\symGRA\symIndexS} }{\textstyle \left(1+\frac{(\symBase-1)\symBase^{-\symIndexS}}{\symBase-1+\symBase^{-\symNumExtraBits}}\right)^{\!2\symGRA}}
  \right)
  -1\right).
\end{equation*}
The last factor of the estimator \eqref{equ:gra_estimator} comes from applying the second-order delta method \cite{Casella2002} and corrects some bias similar to the last factor in \eqref{equ:ml_bias_correction}.

\begin{figure}[t]
  \centering
  \includegraphics[width=\linewidth]{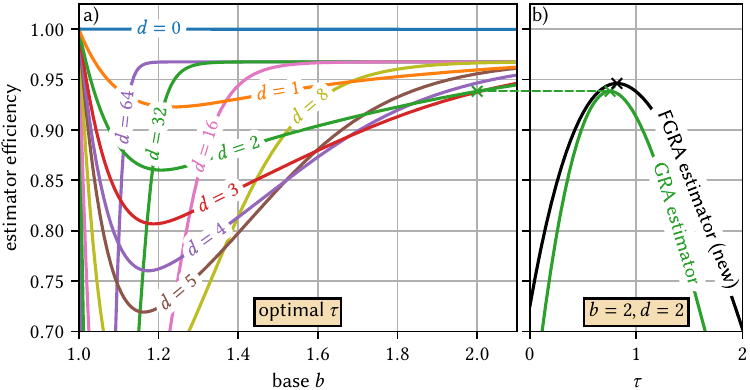}
  \caption{\boldmath a) The asymptotic \acs*{GRA} estimator efficiency over the base $\symBase$ for $\symBitsForMax=6$ and various values of $\symNumExtraBits$. b) The efficiencies of the \acs*{GRA} estimator and our proposed \acs*{FGRA} estimator as a function of $\symGRA$ for $\symBase=2$ and $\symNumExtraBits=2$. Crosses indicate optimal choices of $\symGRA$.}
  \label{fig:gra_efficiency}
\end{figure}

As $\symGRA$ is a free parameter, it is ideally chosen to minimize the variance approximated by (see \ifextended \cref{lem:cardinality_estimator} in \cref{app:proofs}\else extended paper \cite{Ertl2023}\fi)
\begin{equation}
  \label{equ:variance_gra}
  \symVariance(\symCardinalityEstimator/\symCardinality)
  \approx
  \symVarianceFactor/\symNumReg
  +
  \symComplexity(\symNumReg^{-2}).
\end{equation}
For \ac{ULL} with $\symBase=2$ and $\symNumExtraBits=2$, numerical minimization gives $\symVarianceFactor \approx 0.616990$ for $\symGRA\approx0.755097$.
In this case \eqref{equ:gra_reg_contribution} can be written as
\begin{equation}
  \label{equ:reg_contrib_constants}
  \symRegContrib(\symRegister) := 2^{-\symGRA\lfloor\symRegister/4\rfloor} \symContributionCoefficient_{\symRegister \bmod 4}
\end{equation}
with coefficients $\symContributionCoefficient_0 \approx 4.841356$, $\symContributionCoefficient_1 \approx 2.539198$, $\symContributionCoefficient_2 \approx 3.477312$, $\symContributionCoefficient_3 \approx 1.175153$.
The corresponding \ac{MVP} is $8\symVarianceFactor = 4.935917$ which is slightly greater than the theoretical \ac{MVP} \eqref{equ:mvp_uncompressed} with a value of $4.631289$. Hence, the efficiency of the \ac{GRA} estimator is $93.8\%$.

\cref{fig:gra_efficiency}a shows the asymptotic \ac{GRA} estimator efficiency as $\symNumReg\rightarrow\infty$ relative to \eqref{equ:mvp_uncompressed} over the base $\symBase$ for $\symBitsForMax=6$ and for various values of $\symNumExtraBits$. The \ac{GRA} estimator is very efficient for $\symNumExtraBits=0$ and therefore for \ac{HLL}. It can be clearly seen that its efficiency is lower for other configurations with $\symNumExtraBits\geq 1$. Therefore, we investigated if we can find a more efficient estimator for \ac{ULL}, that is as simple and cheap to evaluate as the \ac{GRA} estimator.

\subsection{New \acs*{FGRA} Estimator}

Our idea is to further generalize the \ac{GRA} estimator, by choosing not only $\symGRA$ but also the coefficients $\symContributionCoefficient_0$, $\symContributionCoefficient_1$, $\symContributionCoefficient_2$, and $\symContributionCoefficient_3$ in \eqref{equ:reg_contrib_constants} such that the variance is minimized. This is why we call the new estimator \ac{FGRA} estimator.
The optimal values for $\symContributionCoefficient_0$, $\symContributionCoefficient_1$, $\symContributionCoefficient_2$, and $\symContributionCoefficient_3$ for fixed $\symGRA$ are given by \ifextended (see \cref{lem:expectation_gamma_new,lem:fgra_second_moment,lem:fgra_minimum,lem:cardinality_estimator} in \cref{app:proofs}) \else \cite{Ertl2023} \fi
\begin{equation}
  \label{equ:our_estimator_contrib_coefficients}
  \textstyle
  \symContributionCoefficient_{\symIndexJ} = \frac{\ln \symBase}{\symGammaFunc(\symGRA)}\frac{\symNewConstant_\symIndexJ(\symGRA) }{\symNewConstant_\symIndexJ(2\symGRA) }\left(
  \frac{\symNewConstant_0^2(\symGRA) }{\symNewConstant_0(2\symGRA)}
  +
  \frac{\symNewConstant_1^2(\symGRA) }{\symNewConstant_1(2\symGRA)}
  +
  \frac{\symNewConstant_2^2(\symGRA) }{\symNewConstant_2(2\symGRA)}
  +
  \frac{\symNewConstant_3^2(\symGRA) }{\symNewConstant_3(2\symGRA)}
  \right)^{\!-1}
\end{equation}
where the functions $\symNewConstant_0$, $\symNewConstant_1$, $\symNewConstant_2$, and $\symNewConstant_3$ are defined as
\begin{align*}
  \scriptstyle\symNewConstant_0(\symGRA)
   &
  \scriptstyle\ :=\ \frac{1}{(\symBase^3-\symBase+1)^{\symGRA}}-\frac{1}{\symBase^{3\symGRA}},
  \\
  \scriptstyle\symNewConstant_1(\symGRA)
   &
  \scriptstyle\ :=\ \frac{1}{(\symBase^2-\symBase+1)^{\symGRA}}-\frac{1}{\symBase^{2\symGRA}}-\frac{1}{(\symBase^3-\symBase+1)^{\symGRA}}+\frac{1}{\symBase^{3\symGRA}}, \\
  \scriptstyle\symNewConstant_2(\symGRA)
   &
  \scriptstyle\ :=\ \frac{1}{(\symBase^3-\symBase^2+1)^{\symGRA}}-\frac{1}{(\symBase^3-\symBase^2+\symBase)^{\symGRA}}
  -\frac{1}{(\symBase^3-\symBase+1)^{\symGRA}}+\frac{1}{\symBase^{3\symGRA}},
  \\
  \scriptstyle\symNewConstant_3(\symGRA)
   &
  \scriptstyle\ :=\
  \frac{1}{(\symBase^3-\symBase+1)^{\symGRA}}
  -\frac{1}{(\symBase^3-\symBase^2+1)^{\symGRA}}+\frac{1}{(\symBase^3-\symBase^2+\symBase)^{\symGRA}} -\frac{1}{\symBase^{3\symGRA}}
  - \frac{1}{(\symBase^2-\symBase+1)^{\symGRA}}+\frac{1}{\symBase^{2\symGRA}}
  +1-\frac{1}{\symBase^{\symGRA}}.
\end{align*}
The minimum variance is then given by \eqref{equ:variance_gra} with
\begin{equation*}
  \textstyle
  \symVarianceFactor
  =
  \frac{1}{\symGRA^2}\left(\frac{\symGammaFunc(2\symGRA) \ln\symBase}{(\symGammaFunc(\symGRA))^2}\left(
    \frac{\symNewConstant_0^2(\symGRA) }{\symNewConstant_0(2\symGRA)}
    +
    \frac{\symNewConstant_1^2(\symGRA) }{\symNewConstant_1(2\symGRA)}
    +
    \frac{\symNewConstant_2^2(\symGRA) }{\symNewConstant_2(2\symGRA)}
    +
    \frac{\symNewConstant_3^2(\symGRA) }{\symNewConstant_3(2\symGRA)}
    \right)^{\!-1} - 1\right).
\end{equation*}
Numerical minimization of this expression with regard to $\symGRA$ for the \ac{ULL} case with $\symBase=2$ and $\symNumExtraBits=2$ yields $\symGRA\approx 0.819491$ for which $\symVarianceFactor\approx 0.611893$ and the corresponding coefficients following \eqref{equ:our_estimator_contrib_coefficients} are
\begin{equation*}
  \symContributionCoefficient_0\approx 4.663135,
  \symContributionCoefficient_1\approx 2.137850,
  \symContributionCoefficient_2\approx 2.781145,
  \symContributionCoefficient_3\approx 0.982408.
\end{equation*}
The resulting $\symMVP = 8\symVarianceFactor \approx 4.895145$ corresponds to an efficiency of $94.6\%$, which is a small improvement over the \ac{GRA} estimator as shown in \cref{fig:gra_efficiency}b.

An interesting but not optimal choice would be $\symGRA=1$, since this would make exponentiations in \eqref{equ:gra_estimator} and \eqref{equ:reg_contrib_constants} particularly cheap. The corresponding coefficients would be
$\symContributionCoefficient_0\approx 6.037409$,
$\symContributionCoefficient_1\approx 2.415940$,
$\symContributionCoefficient_2\approx 3.364340$,
$\symContributionCoefficient_3\approx 0.934924$,
which result in $\symVarianceFactor \approx 0.617163$ and a \ac{MVP} of $8\symVarianceFactor \approx 4.937304$ corresponding to an efficiency of $93.8\%$. In contrast, the efficiency of the \ac{GRA} estimator would already drop to $91.3\%$ if $\symGRA=1$ is chosen.

\subsection{Corrections for \acs*{GRA}/\acs*{FGRA} Estimators}
\label{sec:small_large_corr}
Both the \ac{GRA} as well as our newly proposed \ac{FGRA} estimator assume that register values are distributed according to \eqref{equ:simple_register_pmf} rather than \eqref{equ:register_pmf}.
However, as mentioned in \cref{sec:stat_model}, this assumption is not valid, if a significant number of register values lies outside of $[(\symNumExtraBits+1) 2^\symNumExtraBits, \symMaxUpdateValMax 2^\symNumExtraBits)$. For small (large) distinct counts, there are register values below (above) this range, and the presented plain \ac{GRA}/\ac{FGRA} estimators would no longer work without the corrections introduced below.

Previous work has often not properly addressed this problem.
For example, the estimation error of the original \ac{HLL} estimator \cite{Flajolet2007} exceeded the predicted asymptotic error significantly for certain distinct counts. To be useful in practice, the distinct count estimate should be within the same error bounds over the entire range. While the estimator for \ac{HLL} was patched in the meantime \cite{Heule2013, Qin2016, Zhao2016, Ertl2017, Ting2019}, recent algorithms such as \ac{EHLL} \cite{Ohayon2021} or HyperLogLogLog \cite{Karppa2022} still suggest, like the original \ac{HLL} algorithm \cite{Flajolet2007}, to switch between different estimators. However, the corresponding transitions are not seamless which results in discontinuous estimation errors \cite{Ertl2017}.
In the following we present a technique, which generalizes ideas presented in our earlier works \cite{Ertl2017,Ertl2021}, to compose estimators of kind \eqref{equ:gra_estimator} with specific estimators for small and large distinct counts such that the resulting estimator works seamlessly over the whole value range.

A register state $\symRegister= \symMaxUpdateVal2^\symNumExtraBits + \langle\symIndexBit_1\ldots\symIndexBit_{\symNumExtraBits}\rangle_2$ allows to draw conclusions about the occurrence of certain events $\symUpdateEvent_\symUpdateVal$ as introduced in \cref{sec:stat_model}. For example, if $\symMaxUpdateVal < \symMaxUpdateValMax$, we know for sure that $\symUpdateEvent_\symMaxUpdateVal$ has occurred while events $\symUpdateEvent_\symUpdateVal$ with $\symUpdateVal > \symMaxUpdateVal$ have not occurred. Bit $\symIndexBit_\symIndexS$ indicates the occurrence of $\symUpdateEvent_{\symMaxUpdateVal-\symIndexS}$ as long as $\symMaxUpdateVal>\symIndexS$. However, we know nothing about events $\symUpdateEvent_\symUpdateVal$ with $\symUpdateVal < \symMaxUpdateVal - \symNumExtraBits$ or $\symUpdateVal\leq 0$.

We can derive the \ac{PMF} $\symDensityRegisterSimple(\symRegisterSimple\vert \symRegister,\symCardinality)$, conditioned on the secured information about event occurrences we know from $\symRegister$, for register values $\symRegisterSimple$ following the simple model \eqref{equ:simple_register_pmf}. For the case $\symNumExtraBits+1 \leq \symMaxUpdateVal < \symMaxUpdateValMax$ corresponding to $(\symNumExtraBits+1) 2^\symNumExtraBits \leq \symRegister < \symMaxUpdateValMax 2^\symNumExtraBits$, $\symRegisterSimple$ and $\symRegister$ will always be equal, which yields $\symDensityRegisterSimple(\symRegisterSimple\vert \symRegister,\symCardinality) = [\symRegisterSimple = \symRegister]$ when using the Iverson bracket notation. A complete derivation of $\symDensityRegisterSimple(\symRegisterSimple\vert \symRegister,\symCardinality)$ for \ac{ULL} with $\symNumExtraBits=2$ and $\symBase=2$ is given in \ifextended \cref{app:proofs} under \cref{lem:conditional_pmf}\else the extended paper \cite{Ertl2023}\fi.

Assuming that the distinct count is equal to the estimate $\symCardinalityEstimatorAlt$ obtained by an alternative distinct count estimator, we define the corrected register contribution $\symRegContribCorr(\symRegister)$, to be used in \eqref{equ:gra_estimator} as replacement for $\symRegContrib(\symRegister)$, as the expected register contribution of $\symRegContrib(\symRegisterSimple)$ where $\symRegisterSimple$ follows the conditional \ac{PMF} $\symDensityRegisterSimple(\symRegisterSimple\vert \symRegister, \symCardinality = \symCardinalityEstimatorAlt)$
\begin{equation}
  \label{equ:corrected_contributions_general}
  \textstyle
  \symRegContribCorr(\symRegister) = \sum_{\symRegisterSimple=-\infty}^\infty \symDensityRegisterSimple(\symRegisterSimple\vert \symRegister,\symCardinality = \symCardinalityEstimatorAlt)\symRegContrib(\symRegisterSimple).
\end{equation}
This approach leads to $\symRegContribCorr(\symRegister) = \symRegContrib(\symRegister)$ for the case $(\symNumExtraBits+1) 2^\symNumExtraBits \leq \symRegister < \symMaxUpdateValMax 2^\symNumExtraBits$. Hence, if all registers are in this intermediate range, the estimator \eqref{equ:gra_estimator} remains unchanged. This is expected as the \acp{PMF} \eqref{equ:register_pmf} and \eqref{equ:simple_register_pmf} are equivalent in this case. However, the corrected contributions $\symRegContribCorr(\symRegister)$ of very small and very large register values will differ significantly from $\symRegContrib(\symRegister)$. For \ac{ULL} with $\symNumExtraBits=2$ and $\symBase=2$ where $\symRegContrib(\symRegister)$ follows \eqref{equ:reg_contrib_constants}, the corrected register contributions can be written as
\begin{equation}
  \label{equ:corrected_contributions}
  \scriptstyle\symRegContribCorr(\symRegister) =
  \begin{cases}
    \scriptstyle\symSmallRangeCorrFunc(\symZEstimate_0)
     &
    \scriptstyle\symRegister = 0,
    \\
    \scriptstyle2^{-\symGRA } \symCubicContribFunc(\symZEstimate_0)
     &
    \scriptstyle\symRegister = 4,
    \\
    \scriptstyle4^{-\symGRA} (\symZEstimate_0 (\symContributionCoefficient_0 -\symContributionCoefficient_1) + \symContributionCoefficient_1)
     &
    \scriptstyle\symRegister = 8,
    \\
    \scriptstyle4^{-\symGRA} (\symZEstimate_0 (\symContributionCoefficient_2 -\symContributionCoefficient_3) + \symContributionCoefficient_3)
     &
    \scriptstyle\symRegister = 10,
    \\
    \scriptstyle2^{-\symGRA\lfloor\symRegister/4\rfloor} \symContributionCoefficient_{\symRegister \bmod 4}
     &
    \scriptstyle \symRegister \in [12, 4\symMaxUpdateValMax),
    \\
    \scriptstyle\frac{\symZEstimate_\symMaxUpdateValMax(1+\sqrt{\symZEstimate_\symMaxUpdateValMax}) \symContributionCoefficient_{0}
      + 2^{-\symGRA}\sqrt{\symZEstimate_{\symMaxUpdateValMax}}
      (\symZEstimate_{\symMaxUpdateValMax} (\symContributionCoefficient_0 -\symContributionCoefficient_2)
      + \symContributionCoefficient_2)
      + \symLargeRangeCorrFunc(\sqrt{\symZEstimate_{\symMaxUpdateValMax}})
    }{2^{\symGRA\symMaxUpdateValMax}(1+\sqrt{\symZEstimate_{\symMaxUpdateValMax}})(1+\symZEstimate_{\symMaxUpdateValMax})}
     &
    \scriptstyle\symRegister = 4\symMaxUpdateValMax,
    \\
    \scriptstyle\frac{\symZEstimate_\symMaxUpdateValMax(1+\sqrt{\symZEstimate_\symMaxUpdateValMax}) \symContributionCoefficient_{1}
      + 2^{-\symGRA}\sqrt{\symZEstimate_{\symMaxUpdateValMax}}
      (\symZEstimate_{\symMaxUpdateValMax} (\symContributionCoefficient_0 -\symContributionCoefficient_2)
      + \symContributionCoefficient_2)
      + \symLargeRangeCorrFunc(\sqrt{\symZEstimate_{\symMaxUpdateValMax}})
    }{2^{\symGRA\symMaxUpdateValMax}(1+\sqrt{\symZEstimate_{\symMaxUpdateValMax}})(1+\symZEstimate_{\symMaxUpdateValMax})}
     &
    \scriptstyle\symRegister = 4\symMaxUpdateValMax+1,
    \\
    \scriptstyle\frac{\symZEstimate_\symMaxUpdateValMax(1+\sqrt{\symZEstimate_\symMaxUpdateValMax}) \symContributionCoefficient_{2}
      + 2^{-\symGRA}\sqrt{\symZEstimate_{\symMaxUpdateValMax}}
      (\symZEstimate_{\symMaxUpdateValMax} (\symContributionCoefficient_1-\symContributionCoefficient_3)
      + \symContributionCoefficient_3)
      + \symLargeRangeCorrFunc(\sqrt{\symZEstimate_{\symMaxUpdateValMax}})
    }{2^{\symGRA\symMaxUpdateValMax}(1+\sqrt{\symZEstimate_{\symMaxUpdateValMax}})(1+\symZEstimate_{\symMaxUpdateValMax})}
     &
    \scriptstyle\symRegister = 4\symMaxUpdateValMax+2,
    \\
    \scriptstyle\frac{\symZEstimate_\symMaxUpdateValMax(1+\sqrt{\symZEstimate_\symMaxUpdateValMax}) \symContributionCoefficient_{3}
      + 2^{-\symGRA}\sqrt{\symZEstimate_{\symMaxUpdateValMax}}
      (\symZEstimate_{\symMaxUpdateValMax} (\symContributionCoefficient_1-\symContributionCoefficient_3)
      + \symContributionCoefficient_3)
      + \symLargeRangeCorrFunc(\sqrt{\symZEstimate_{\symMaxUpdateValMax}})
    }{2^{\symGRA\symMaxUpdateValMax}(1+\sqrt{\symZEstimate_{\symMaxUpdateValMax}})(1+\symZEstimate_{\symMaxUpdateValMax})}
     &
    \scriptstyle\symRegister = 4\symMaxUpdateValMax+3
  \end{cases}
\end{equation}
(see \ifextended \cref{lem:corr_reg_contrib} in \cref{app:proofs}\else extended paper \cite{Ertl2023}\fi).
The functions $\symCubicContribFunc$, $\symSmallRangeCorrFunc$, and $\symLargeRangeCorrFunc$ are given by
\begin{align}
  \label{equ:psi_func}
  \scriptstyle
  \symCubicContribFunc(\symZ)
                                & \scriptstyle:=
  \symZ(\symZ(\symZ(\symContributionCoefficient_0 -\symContributionCoefficient_1-\symContributionCoefficient_2+\symContributionCoefficient_3)
  +
  (\symContributionCoefficient_2-\symContributionCoefficient_3))
  +
  (\symContributionCoefficient_1-\symContributionCoefficient_3))
  +
  \symContributionCoefficient_3,
  \\
  \label{equ:sigma_func}
  \scriptstyle
  \symSmallRangeCorrFunc(\symZ) & \scriptstyle := \frac{1}{\symZ} \sum_{\symMaxUpdateVal=0}^\infty 2^{\symGRA \symMaxUpdateVal} (\symZ^{2^{\symMaxUpdateVal}}-\symZ^{2^{\symMaxUpdateVal+1}})
  \symCubicContribFunc(\symZ^{2^{\symMaxUpdateVal+1}}),
  \\
  \label{equ:phi_func}
  \scriptstyle
  \symLargeRangeCorrFunc(\symZ) & \scriptstyle := \frac{4^{-\symGRA}}{1-\symZ}\sum_{\symMaxUpdateVal=0}^\infty 2^{-\symGRA \symMaxUpdateVal}
  (\symZ^{2^{-\symMaxUpdateVal-1}}-\symZ^{2^{-\symMaxUpdateVal}})
  \symCubicContribFunc( \symZ^{2^{-\symMaxUpdateVal}})
  \\
                                &
  \scriptstyle
  =
  \frac{4^{-\symGRA}}{2-2^{-\symGRA}}
  \left(
  \frac{2\symCubicContribFunc(\symZ)\sqrt{\symZ}}{1+\sqrt{\symZ}}
  +
  \sum_{\symMaxUpdateVal=1}^\infty
  \frac{\symZ^{2^{-\symMaxUpdateVal-1}}\left(2\symCubicContribFunc(\symZ^{2^{-\symMaxUpdateVal}})-
      (\symZ^{2^{-\symMaxUpdateVal-1}} + \symZ^{2^{-\symMaxUpdateVal}})
      \symCubicContribFunc(\symZ^{2^{-\symMaxUpdateVal+1}})
      \right)}{2^{\symGRA \symMaxUpdateVal}\prod_{\symIndexJ=1}^{\symMaxUpdateVal+1} (1 + \symZ^{2^{-\symIndexJ}})}
  \right).\nonumber
\end{align}
Furthermore, $\symZEstimate_0$ and $\symZEstimate_\symMaxUpdateValMax$ are defined as $\symZEstimate_0:=e^{-\frac{\symCardinalityEstimatorAlt}{\symNumReg}}$ and $\symZEstimate_\symMaxUpdateValMax:=e^{-\frac{\symCardinalityEstimatorAlt}{\symNumReg 2^\symMaxUpdateValMax}}$ and are therefore estimates of $\symZ_0 = e^{-\frac{\symCardinality}{\symNumReg}}$ and $\symZ_\symMaxUpdateValMax = e^{-\frac{\symCardinality}{\symNumReg 2^\symMaxUpdateValMax}}$, respectively. According to \eqref{equ:corrected_contributions}, $\symZEstimate_0$ is only needed for register values from $\lbrace 0, 4, 8, 10\rbrace$. Therefore, the estimator $\symZEstimate_0$ must work well only for small distinct counts. Similarly, the estimator $\symZEstimate_\symMaxUpdateValMax$ must work only for large distinct counts, when there are registers greater than or equal to $4\symMaxUpdateValMax$. Corresponding estimators for both cases are presented in the next sections.

The numerical evaluation of $\symSmallRangeCorrFunc$ is very cheap, because only basic mathematical operations are involved and the infinite series converges quickly. The convergence is slowest for arguments $\symZEstimate_0$ close to 1. The largest arguments smaller than 1 occur for distinct counts equal to 1, which implies $\symZEstimate_0\approx e^{-\frac{1}{\symNumReg}}$. However, even in this extreme case, empirical analysis showed numerical convergence after the first $\symPrecision+7$ terms for any precision $\symPrecision\in[3, 26]$ when using double-precision floating-point arithmetic. For not too small values of $\symPrecision$, the estimation costs are dominated by the iteration over all $\symNumReg = 2^{\symPrecision}$ registers to sum up the individual register contributions.

When using 64-bit hash values, the maximum update value is given by $\symMaxUpdateValMax=65-\symPrecision$ (see \cref{sec:stat_model}). In this case, registers with values $\geq 4\symMaxUpdateValMax$ are very unlikely in practice and the evaluation of $\symLargeRangeCorrFunc$ is rarely needed. Nevertheless, it also can be computed efficiently with the second (more complex appearing) expression for $\symLargeRangeCorrFunc$, which converges numerically after at most 22 terms for any $\symPrecision\in[3, 26]$.

\subsection{Estimator for Small Distinct Counts}
\label{sec:small_estimator}
The original \ac{HLL} algorithm \cite{Flajolet2007} switches, in case of small distinct counts with many registers equal to zero, to the estimator
\begin{equation}
  \label{equ:linear_conting_estimator}
  \symCardinalityEstimatorLow = \symNumReg \ln(\symNumReg/\symCount_0),
\end{equation}
known from probabilistic linear counting \cite{Whang1990}. $\symCount_0 := |\lbrace \symRegAddr\vert \symRegister_\symRegAddr = 0 \rbrace|$ denotes the number of registers that have never been updated. This estimator corresponds to the \ac{ML} estimator, when considering just the number of registers with $\symRegister_\symRegAddr = 0$ and using $\symProbability(\symRegister_\symRegAddr = 0) = \symZ_0 = e^{-\frac{\symCardinality}{\symNumReg}}$ which follows from \eqref{equ:register_pmf} for $\symBase=2$. For \ac{HLL}, the combination of this estimator with \eqref{equ:corrected_contributions_general} leads to the same correction terms as we have previously derived in a different way \cite{Ertl2017}. The resulting so-called \ac{CR} estimator was shown to be nearly as efficient as the \ac{ML} estimator.

This finding provides confidence that the same approach also works for \ac{ULL} to correct the \ac{GRA}/\ac{FGRA} estimator. Even though \eqref{equ:linear_conting_estimator} could be used again to estimate small distinct counts, we found a simple estimator that is able to exploit more information. We consider the four smallest possible register states $\symRegister_\symRegAddr \in \lbrace 0, 4, 8, 10 \rbrace$ with corresponding probabilities following \eqref{equ:register_pmf} and apply the \ac{ML} method to estimate $\symZ_0 = e^{-\frac{\symCardinality}{\symNumReg}}$. As shown in \ifextended \cref{app:proofs} under \cref{lem:small_range_estimator} \else the extended paper \cite{Ertl2023} \fi the resulting \ac{ML} estimator can be written as
\begin{equation}
  \label{equ:small_range_estimator}
  \symZEstimate_0 =
  \textstyle
  (
  (
  \sqrt{\symBeta^2 + 4\symAlpha \symGamma} - \symBeta
  )/(
  2 \symAlpha
  )
  )^{4}.
\end{equation}
Here $\symAlpha$, $\symBeta$, $\symGamma$ are defined as
\begin{equation*}
  \scriptstyle\symAlpha\ :=\ \symNumReg + 3(\symCount_0 + \symCount_4 + \symCount_8 + \symCount_{10}), \quad
  \symBeta \ :=\ \symNumReg - \symCount_0 - \symCount_4,
  \quad
  \symGamma \ :=\ 4\symCount_0 + 2\symCount_4 + 3\symCount_8 + \symCount_{10}
\end{equation*}
and $\symCount_\symIndexJ := |\lbrace \symRegAddr\vert \symRegister_\symRegAddr = \symIndexJ \rbrace|$ is again the number of registers with value $\symIndexJ$.

\subsection{Estimator for Large Distinct Counts}
\label{sec:large_estimator}
An estimator for large distinct counts can be found in a similar way through \ac{ML} estimation when considering the largest 4 register states $\symRegister_\symRegAddr \in \lbrace 4\symMaxUpdateValMax, 4\symMaxUpdateValMax+1, 4\symMaxUpdateValMax+2, 4\symMaxUpdateValMax+3\rbrace$.
The \ac{ML} estimator for $\symZ_\symMaxUpdateValMax = e^{-\frac{\symCardinality}{\symNumReg 2^\symMaxUpdateValMax}}$ can be written as
\ifextended (see \cref{app:proofs} under \cref{lem:large_range_estimator}) \else \cite{Ertl2023} \fi
\begin{equation}
  \label{equ:large_range_estimator}
  \symZEstimate_\symMaxUpdateValMax = \textstyle
  \sqrt{
    (
    \sqrt{\symBeta^2 + 4\symAlpha \symGamma} - \symBeta
    )/(
    2 \symAlpha)
  }.
\end{equation}
Here $\symAlpha$, $\symBeta$, $\symGamma$ are defined as
\begin{align*}
  \scriptstyle\symAlpha
   &
  \scriptstyle\ :=\ \symNumReg + 3(\symCount_{4\symMaxUpdateValMax} + \symCount_{4\symMaxUpdateValMax+1} + \symCount_{4\symMaxUpdateValMax+2} +\symCount_{4\symMaxUpdateValMax+3}), \quad
  \symBeta\ :=\ \symCount_{4\symMaxUpdateValMax} + \symCount_{4\symMaxUpdateValMax+1} + 2\symCount_{4\symMaxUpdateValMax+2} +2\symCount_{4\symMaxUpdateValMax+3},
  \\
  \scriptstyle\symGamma
   &
  \scriptstyle\ :=\ \symNumReg + 2\symCount_{4\symMaxUpdateValMax} + \symCount_{4\symMaxUpdateValMax+2} - \symCount_{4\symMaxUpdateValMax+3}.
\end{align*}

\subsection{Martingale Estimator}

\begin{figure}
  \centering
  \includegraphics[width=\linewidth]{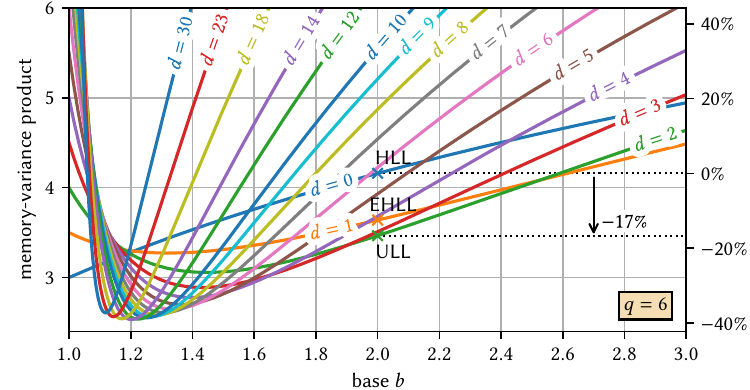}
  \caption{\boldmath The \acf*{MVP} \eqref{equ:mvp_martingale} as a function of the base $\symBase$ for $\symBitsForMax=6$ and various values of $\symNumExtraBits$ when using martingale estimation. The \acs*{MVP} of \acf*{ULL} is 17\% smaller compared to \acf*{HLL}.}
  \label{fig:var_martingale}
\end{figure}

If the data is not distributed and merging of sketches is not needed, the distinct count can be estimated in an online fashion by incremental updates, which is known as martingale or \acf{HIP} estimation \cite{Ting2014, Cohen2015} and which even leads to smaller estimation errors. This estimation approach can also be used for \ac{ULL}. It keeps track of the state change probability $\symStateChangeProbability$. Initially, $\symStateChangeProbability=1$ as the first update will certainly change the state. Whenever a register is changed, the probability of state changes for further elements decreases. The probability, that a new unseen element changes the \ac{ULL} state is given by
\begin{equation*}
  \symStateChangeProbability(\symRegister_0, \ldots, \symRegister_{\symNumReg-1}) = \sum_{\symRegAddr=0}^{\symNumReg-1} \symRegMartingale(\symRegister_\symRegAddr).
\end{equation*}
$\symRegMartingale(\symRegister_\symRegAddr)$ returns the probability that register $\symRegister_\symRegAddr$ is changed with the next new element. Obviously, we have $\symRegMartingale(0) = \frac{1}{\symNumReg}$ and $\symRegMartingale(4\symMaxUpdateValMax+3) = 0$ for the smallest and largest possible states, respectively. In the general case we have (see \ifextended \cref{lem:martingale_contrib} in \cref{app:proofs}\else extended paper \cite{Ertl2023}\fi)
\begin{align}
  \label{equ:martingale_contrib}
  \textstyle
  \symRegMartingale(0)
   &
  \textstyle
  = \frac{1}{\symNumReg},
  \quad
  \symRegMartingale(4) = \frac{1}{2\symNumReg},
  \quad
  \symRegMartingale(8) = \frac{3}{4\symNumReg},
  \quad
  \symRegMartingale(10) = \frac{1}{4\symNumReg},
  \nonumber
  \\
  \textstyle
  \symRegMartingale(\symRegister)
   &
  \textstyle
  =
  \frac{7 -2\symIndexBit_1 - 4\symIndexBit_2}{2^\symMaxUpdateVal\symNumReg}\quad \text{for $\symRegister = 4\symMaxUpdateVal + \langle\symIndexBit_1\symIndexBit_2\rangle_2$ and $3\leq \symMaxUpdateVal < \symMaxUpdateValMax$},
  \\
  \textstyle
  \symRegMartingale(\symRegister)
   &
  \textstyle
  = \frac{3 - \symIndexBit_1 -2\symIndexBit_2}{2^{\symMaxUpdateValMax-1}\symNumReg}\quad \text{for $\symRegister = 4\symMaxUpdateValMax + \langle\symIndexBit_1\symIndexBit_2\rangle_2$}.\nonumber
\end{align}

The martingale estimator is incremented with every state change by $\frac{1}{\symStateChangeProbability}$ prior the update as demonstrated by \cref{alg:martingale}. $\symStateChangeProbability$ itself can also be incrementally adjusted, such that the whole update takes constant time. The values of $\symRegMartingale$ can be stored in a lookup table. The martingale estimator is unbiased and optimal \cite{Pettie2021a}, but cannot be used if the data is distributed and mergeability is required. For our generalized data structure introduced in \cref{sec:data_structure}, the asymptotic \ac{MVP} can be derived as (see \ifextended \cref{lem:mvp_martingale} in \cref{app:proofs}\else extended paper \cite{Ertl2023}\fi)
\begin{equation}
  \label{equ:mvp_martingale}
  \textstyle\symMVP \approx (\symBitsForMax + \symNumExtraBits)\frac{1}{2}\ln(\symBase)(1+\frac{\symBase^{-\symNumExtraBits}}{\symBase-1}).
\end{equation}
This expression is consistent with previous specific results reported for \ac{HLL} \cite{Pettie2021a} and \ac{EHLL} \cite{Ohayon2021} and is plotted in \cref{fig:var_martingale} for $\symBitsForMax=6$ and various values of $\symBase$ and $\symNumExtraBits$. For $\symBase=2$, the \ac{ULL} configuration with $\symNumExtraBits=2$ is even optimal. \ac{ULL} achieves $\symMVP \approx 3.4657$ which is 17\% less than for \ac{HLL} with $\symMVP \approx 4.1589$.

\section{Practical Implementation}
As registers of \ac{ULL} have $\symBitsForMax + \symNumExtraBits = 8$ bits, they can be stored in a plain byte array. \cref{alg:insertion} shows the update procedure for inserting an element. The actual input is the hash value of the element from which the register index $\symRegAddr$ and an update value $\symUpdateVal$ are obtained. As for other probabilistic data structures it is essential to use a high-quality hash function such as WyHash \cite{Yi}, Komihash \cite{Vaneev}, or PolymurHash \cite{Peters} whose outputs are in practice indistinguishable from true random integers. If 64-bit hash values are used, the maximum update value is limited to $\symMaxUpdateValMax=65-\symPrecision$ leading to a maximum distinct count of roughly $\symCardinalityMax \approx \symNumReg 2^{\symMaxUpdateValMax-1} =2^{64}$ (cf. \cref{sec:parameter-choice}) which is more than sufficient for any realistic application.

\cref{alg:insertion} updates the $\symRegAddr$-th register by first unpacking its value to a 64-bit value $\symX = \langle\symX_{63}\symX_{62}\ldots\symX_1\symX_0\rangle_2$ (see \cref{alg:unpack}) where bit $\symX_{\symIndexJ+1}$ indicates the occurrence of update value $\symIndexJ$. After setting the corresponding bit for the new update value $\symUpdateVal$, $\symX$ is finally packed again into a register value (see \cref{alg:pack}), preserving just the information about the maximum update value $\symMaxUpdateVal$ and whether update values $\symMaxUpdateVal-1$ and $\symMaxUpdateVal-2$ have already occurred.

Although the update algorithm for \ac{HLL} (cf. \cref{alg:insertion_hll}) looks simpler, in practice the implementation complexity is similar if one also includes the bit twiddling required to pack 6-bit registers densely into a byte array.
\ac{ULL} insertions can be implemented entirely branch-free, as exemplified by the implementation in our Hash4j Java library (see \url{https://github.com/dynatrace-oss/hash4j}), which however uses the transformation $\symRegister \rightarrow \symRegister + 4\symPrecision - 8$ for non-zero registers, so that the largest possible register state is always $(4\symMaxUpdateValMax + 3) + 4\symPrecision - 8 = 255$, the maximum value of a byte.

\cref{alg:estimation} summarizes \ac{FGRA} estimation including corrections for small and large distinct counts as described in \cref{sec:small_large_corr}. For intermediate counts, if registers are all in the range $[12, 4\symMaxUpdateValMax)$, the algorithm just sums up the individual register contributions $\symRegContrib(\symRegister_\symRegAddr)$ given by \eqref{equ:reg_contrib_constants}. The values of $\symRegContrib$ can be pre-computed and stored in a lookup table for all possible register values. Following \eqref{equ:gra_estimator}, the final estimate is obtained by exponentiating the sum of all register contributions by $-\frac{1}{\symGRA}$ and multiplying by the precomputed factor $\symEstimationConstant_\symPrecision$. For small distinct counts, when there are any registers smaller than $12$, thus $\symCount_0+ \symCount_4+ \symCount_8+ \symCount_{10}>0$, the small range correction branch is executed, which requires computing the estimator described in \cref{sec:small_estimator}. Similarly, for large distinct counts with registers greater than or equal to $4\symMaxUpdateValMax$, equivalent to $\symCount_{4\symMaxUpdateValMax}+ \symCount_{4\symMaxUpdateValMax+1}+ \symCount_{4\symMaxUpdateValMax+2}+ \symCount_{4\symMaxUpdateValMax+3}>0$, the large range correction is applied based on the estimator presented in \cref{sec:large_estimator}.

\subsection{Merging and Downsizing}
\ac{ULL} sketches can be merged which is important if the data is distributed over space and/or time and partial results must be combined. The final result will only depend on the set of added elements and not on the insertion or merge order.
It is even possible to merge \ac{ULL} sketches with different precisions. For that, the sketch with higher precision must be reduced first to the lower precision of the other sketch.

Downsizing from precision $\symPrecision'$ to $\symPrecision<\symPrecision'$ is realized by merging the information of batches of $2^{\symPrecision' - \symPrecision}$ consecutive registers. The resulting register value depends not only on the individual register values, but also on the least significant $\symPrecision' - \symPrecision$ bits of the register indices. Only the value of the first register of a batch, whose index has $\symPrecision' - \symPrecision$ trailing zeros, is relevant. All other registers need to be treated differently, as the $\symPrecision' - \symPrecision$ trailing index bits would have affected the calculation of the \acf{NLZ} computation in \cref{alg:insertion} when recorded with lower precision $\symPrecision$.

\cref{alg:merge} shows in detail how an \ac{ULL} sketch with precision $\symPrecision'$ can be merged in-place into another one with precision $\symPrecision \leq \symPrecision'$. The algorithm simplifies a lot when the precisions are equal, because $\symPrecision'=\symPrecision$ implies $2^{\symPrecision' - \symPrecision} = 1$ and consequently the inner loop can be skipped entirely. In this case, and also due to the byte-sized registers, the \ac{ULL} merge operation is very well suited for \ac{SIMD} processing.

\cref{alg:merge} can also be applied to just downsize a \ac{ULL} sketch from precision $\symPrecision'$ to $\symPrecision$, if the sketch is simply added to an empty \ac{ULL} sketch with precision $\symPrecision$ and all registers equal to zero, $\symRegister_\symRegAddr = 0$. In this case the first unpack operation can be skipped as it would always return zero.

\myAlg{
  \caption{Incrementally updates the martingale estimate $\symCardinalityEstimatorMartingale$ and the state change probability $\symStateChangeProbability$ whenever a register is altered from $\symRegister$ to $\symRegister'$ ($\symRegister<\symRegister'$). Initially, $\symCardinalityEstimatorMartingale = 0$ and $\symStateChangeProbability=1$.}
  \label{alg:martingale}
  $\symCardinalityEstimatorMartingale\gets \symCardinalityEstimatorMartingale + \frac{1}{\symStateChangeProbability}$\Comment*[r]{update estimate}
  $\symStateChangeProbability\gets \symStateChangeProbability - \symRegMartingale(\symRegister) + \symRegMartingale(\symRegister')$ \Comment*[r]{update state change probability, see \eqref{equ:martingale_contrib}}
}

\myAlg{
  \caption{Inserts an element with 64-bit hash value $\langle\symHashBit_{63} \symHashBit_{62} \ldots \symHashBit_{0}\rangle_2$ into an UltraLogLog with registers $\symRegister_0, \symRegister_1, \ldots, \symRegister_{\symNumReg-1}$, initial values $\symRegister_\symRegAddr = 0$, and $\symNumReg = 2^\symPrecision$ ($\symPrecision \geq 3$).}
  \label{alg:insertion}
  $\symRegAddr\gets \langle\symHashBit_{63} \symHashBit_{62}\ldots\symHashBit_{64-\symPrecision}\rangle_2$\Comment*[r]{extract register index}
  $\symMasked\gets \langle\,\underbracket[0.5pt][1pt]{0\ldots 0}_{\scriptscriptstyle\symPrecision}\!\symHashBit_{63-\symPrecision}\symHashBit_{62-\symPrecision} \ldots\symHashBit_{0}\rangle_2$\Comment*[r]{mask register index bits}
  $\symUpdateVal \gets \FuncNLZ{$\symMasked$} -\symPrecision + 1$\Comment*[r]{update value $\symUpdateVal \in [1, 65-\symPrecision]$}
  $\symX \gets $\FuncUncompress{$\symRegister_\symRegAddr$}\Comment*[r]{see \cref{alg:unpack} for $\FuncUncompress$}
  $\symX \gets \symX$ \KwOr $2^{\symUpdateVal+1}$\Comment*[r]{bitwise or-operation}
  $\symRegister_\symRegAddr \gets $\FuncCompress{$\symX$}\Comment*[r]{see \cref{alg:pack} for $\FuncCompress$}
}

\myAlg{
  \caption{Packs a 64-bit value $\symX\geq 4$ into an 8-bit register. $\symX=\langle\symX_{63}\symX_{62}\ldots\symX_1\symX_0\rangle_2$ is interpreted as bitset where bit $\symX_{\symUpdateVal+1}$ indicates the occurrence of update value $\symUpdateVal$.}
  \label{alg:pack}
  \Fn{\FuncCompress{$\symX$}}{
    $\symMaxUpdateVal \gets 62 - \FuncNLZ{$\symX$}$\Comment*[r]{extract maximum update value}
    \KwRet $4\symMaxUpdateVal + \langle\symX_{\symMaxUpdateVal }\symX_{\symMaxUpdateVal - 1}\rangle_2$ \Comment*[r]{return 8-bit register value, is always $\geq 4$}
  }
  \
}

\myAlg{
  \caption{Unpacks an 8-bit register $\symRegister = \langle \symRegisterBit_7\symRegisterBit_6\ldots \symRegisterBit_0 \rangle_2$ into a 64-bit integer indicating the occurrence of update values if interpreted as bitset.}
  \label{alg:unpack}
  \Fn{\FuncUncompress{$\symRegister$}}{
    \lIf(\Comment*[f]{special case for initial state when $\symRegister=0$}){$\symRegister < 4$}{\KwRet 0}
    $\symMaxUpdateVal\gets \langle \symRegisterBit_7\symRegisterBit_6\ldots\symRegisterBit_2\rangle_2$\Comment*[r]{get maximum update value, same as $\lfloor \symRegister/4\rfloor$}
    \KwRet $\langle\,\underbracket[0.5pt][1pt]{0\ldots 0}_{\scriptscriptstyle 62 - \symMaxUpdateVal}\!1\symRegisterBit_1\symRegisterBit_0\!\underbracket[0.5pt][1pt]{0\ldots 0}_{\scriptscriptstyle \symMaxUpdateVal - 1}\,\rangle_2$ \Comment*[r]{return 64-bit value, is always $\geq 4$}
  }
}

\myAlg{
  \caption{Estimates the number of distinct elements from a given UltraLogLog with registers $\symRegister_0, \symRegister_1, \ldots, \symRegister_{\symNumReg-1}$ and $\symNumReg = 2^\symPrecision$ ($\symPrecision \geq 3$) using the \acs*{FGRA} estimator. The constants are defined as $\symMaxUpdateValMax := 65 - \symPrecision$, $\symGRA := 0.8194911375910897$, $\symVarianceFactor := 0.6118931496978437$, $\symContributionCoefficient_0 := 4.663135422063788$,
    $\symContributionCoefficient_1 := 2.1378502137958524$,
    $\symContributionCoefficient_2 := 2.781144650979996$,
    $\symContributionCoefficient_3 := 0.9824082545153715$, and $\symEstimationConstant_\symPrecision:= \symNumReg^{1+\frac{1}{\symGRA}} / (1 + \frac{1+\symGRA}{2}\frac{\symVarianceFactor}{\symNumReg})$.
  }
  \label{alg:estimation}
  $\symSum \gets 0$\Comment*[r]{used to sum up $\sum_{\symRegAddr=0}^{\symNumReg-1}\symRegContribCorr(\symRegister_\symRegAddr)$, see \eqref{equ:corrected_contributions}}
  $(\symCount_0, \symCount_4, \symCount_8, \symCount_{10},\symCount_{4\symMaxUpdateValMax}, \symCount_{4\symMaxUpdateValMax+1}, \symCount_{4\symMaxUpdateValMax+2}, \symCount_{4\symMaxUpdateValMax+3}) \gets (0,0,0,0,0,0,0,0)$\;
  \For(\Comment*[f]{iterate over all $\symNumReg$ registers}){$\symRegAddr\gets 0$ \KwTo $\symNumReg-1$}{
    \uIf{$\symRegister_\symRegAddr< 12 $}{
      \lIf{$\symRegister_\symRegAddr=0 $}{$\symCount_0\gets\symCount_0 +1$}
      \lIf{$\symRegister_\symRegAddr=4 $}{$\symCount_4\gets\symCount_4 +1$}
      \lIf{$\symRegister_\symRegAddr=8 $}{$\symCount_8\gets\symCount_8 +1$}
      \lIf{$\symRegister_\symRegAddr=10 $}{$\symCount_{10}\gets\symCount_{10} +1$}
    }\uElseIf{$\symRegister_\symRegAddr< 4\symMaxUpdateValMax$}{
      $\symSum \gets \symSum + \symRegContrib(\symRegister_\symRegAddr)$\Comment*[r]{see \eqref{equ:reg_contrib_constants}, $\symRegContrib(\symRegister_\symRegAddr)$ can be precomputed}
    }\Else{
      \lIf{$\symRegister_\symRegAddr=4\symMaxUpdateValMax $}{$\symCount_{4\symMaxUpdateValMax}\gets\symCount_{4\symMaxUpdateValMax} +1$}
      \lIf{$\symRegister_\symRegAddr=4\symMaxUpdateValMax + 1 $}{$\symCount_{4\symMaxUpdateValMax+1}\gets\symCount_{4\symMaxUpdateValMax+1} +1$}
      \lIf{$\symRegister_\symRegAddr=4\symMaxUpdateValMax + 2 $}{$\symCount_{4\symMaxUpdateValMax+2}\gets\symCount_{4\symMaxUpdateValMax+2} +1$}
      \lIf{$\symRegister_\symRegAddr=4\symMaxUpdateValMax +3 $}{$\symCount_{4\symMaxUpdateValMax+3}\gets\symCount_{4\symMaxUpdateValMax+3} +1$}
    }
  }
  \If(\Comment*[f]{small range correction}){$\symCount_0+ \symCount_4+ \symCount_8+ \symCount_{10}>0$}{
    $\symAlpha\gets \symNumReg + 3(\symCount_0 + \symCount_4 + \symCount_8 + \symCount_{10})$\;
    $\symBeta\gets\symNumReg - \symCount_0 - \symCount_4$\;
    $\symGamma\gets 4\symCount_0 + 2\symCount_4 + 3\symCount_8 + \symCount_{10}$\;
    $\symZ\gets (
      (\sqrt{\symBeta^2 + 4\symAlpha\symGamma} - \symBeta)/(2 \symAlpha)
      )^{4}$\Comment*[r]{see \eqref{equ:small_range_estimator}}
    \lIf(\Comment*[f]{for $\symSmallRangeCorrFunc$ see \eqref{equ:sigma_func}}){$\symCount_0 > 0$}{$\symSum \gets \symSum + \symCount_0\cdot \symSmallRangeCorrFunc(\symZ)$}
    \lIf(\Comment*[f]{for $\symCubicContribFunc$ see \eqref{equ:psi_func}}){$\symCount_4 > 0$}{$\symSum \gets \symSum + \symCount_4 \cdot 2^{-\symGRA }\cdot \symCubicContribFunc(\symZ)$}
    \lIf{$\symCount_8 > 0$}{$\symSum \gets \symSum + \symCount_8 \cdot (4^{-\symGRA} \cdot (\symZ \cdot (\symContributionCoefficient_0 -\symContributionCoefficient_1) + \symContributionCoefficient_1))$}
    \lIf{$\symCount_{10} > 0$}{
      $\symSum \gets \symSum + \symCount_{10}\cdot (4^{-\symGRA} \cdot (\symZ \cdot (\symContributionCoefficient_2 -\symContributionCoefficient_3) + \symContributionCoefficient_3))$
    }
  }
  \If(\Comment*[f]{large range correction}){$\symCount_{4\symMaxUpdateValMax}+ \symCount_{4\symMaxUpdateValMax+1}+ \symCount_{4\symMaxUpdateValMax+2}+ \symCount_{4\symMaxUpdateValMax+3}>0$}{
    $\symAlpha \gets\symNumReg + 3(\symCount_{4\symMaxUpdateValMax} + \symCount_{4\symMaxUpdateValMax+1} + \symCount_{4\symMaxUpdateValMax+2} +\symCount_{4\symMaxUpdateValMax+3})$\;
    $\symBeta \gets \symCount_{4\symMaxUpdateValMax} + \symCount_{4\symMaxUpdateValMax+1} + 2\symCount_{4\symMaxUpdateValMax+2} +2\symCount_{4\symMaxUpdateValMax+3}$\;
    $\symGamma \gets \symNumReg + 2\symCount_{4\symMaxUpdateValMax} + \symCount_{4\symMaxUpdateValMax+2} - \symCount_{4\symMaxUpdateValMax+3}$\;
    $\symZ\gets \sqrt{
        (\sqrt{\symBeta^2 + 4\symAlpha \symGamma} - \symBeta) / (
        2 \symAlpha
        )}$\Comment*[r]{see \eqref{equ:large_range_estimator}}
    $\symZ'\gets \sqrt{\symZ}$\;
    $\symSum'\gets \symZ\cdot (1+ \symZ')\cdot (\symContributionCoefficient_0\cdot \symCount_{4\symMaxUpdateValMax}+ \symContributionCoefficient_1 \cdot\symCount_{4\symMaxUpdateValMax+1}+\symContributionCoefficient_2\cdot \symCount_{4\symMaxUpdateValMax+2}+ \symContributionCoefficient_3\cdot\symCount_{4\symMaxUpdateValMax+3})$\;
    $\symSum'\gets \symSum' + 2^{-\symGRA}\cdot\symZ'\cdot
      (\symZ\cdot (\symContributionCoefficient_0 -\symContributionCoefficient_2)
      + \symContributionCoefficient_2)\cdot(\symCount_{4\symMaxUpdateValMax} + \symCount_{4\symMaxUpdateValMax+1})$\;
    $\symSum'\gets \symSum' + 2^{-\symGRA}\cdot\symZ'\cdot
      (\symZ\cdot (\symContributionCoefficient_1 -\symContributionCoefficient_3)
      + \symContributionCoefficient_3)\cdot(\symCount_{4\symMaxUpdateValMax+2} + \symCount_{4\symMaxUpdateValMax+3})$\;
    $\symSum'\gets \symSum' + \symLargeRangeCorrFunc(\symZ')\cdot (\symCount_{4\symMaxUpdateValMax}+ \symCount_{4\symMaxUpdateValMax+ 1}+ \symCount_{4\symMaxUpdateValMax+2}+ \symCount_{4\symMaxUpdateValMax+3})$\Comment*{for $\symLargeRangeCorrFunc$ see \eqref{equ:phi_func}}
    $\symSum \gets \symSum + \symSum' / (2^{\symGRA\symMaxUpdateValMax}\cdot(1+\symZ')\cdot(1+\symZ))$\;
  }
  \KwRet $ \symEstimationConstant_\symPrecision\cdot \symSum^{-1/\symGRA}$\Comment*[r]{return distinct count estimate, see \eqref{equ:gra_estimator}}
}

\subsection{Compatibility to HyperLogLog}
When using the same hash function for elements, \ac{ULL} can be implemented in a way that is compatible with an existing \ac{HLL} implementation meaning that an \ac{ULL} sketch can be mapped to a \ac{HLL} sketch of same precision by just dropping the last 2 register bits corresponding to the transformation $\lfloor\symRegister_\symRegAddr^{(\text{\acs*{ULL}})}/4\rfloor \rightarrow \symRegister_\symRegAddr^{(\text{\acs*{HLL}})}$. This allows to migrate to \ac{ULL}, even if there is historical data that was recorded using \ac{HLL} and still needs to be combined with newer data. Our Hash4j library also contains a \ac{HLL} implementation that is compatible with its \ac{ULL} implementation.

\section{Experiments}
Various experiments have been conducted to confirm the theoretical results and to demonstrate the practicality of \ac{ULL}.
For better reproduction, instructions and source code for all experiments including figure generation have been published at \url{https://github.com/dynatrace-research/ultraloglog-paper}. The simulations used the \ac{ULL} and \ac{HLL} implementations available in our open-source Hash4j library (v0.17.0) and were executed on an Amazon EC2 c5.metal instance running Ubuntu Server 22.04 LTS.

\myAlg{
\caption{Merges an UltraLogLog with registers $\symRegister'_0, \symRegister'_1, \ldots, \symRegister'_{\symNumReg'-1}$ and $\symNumReg' = 2^{\symPrecision'}$ into another UltraLogLog with registers $\symRegister_0, \symRegister_1, \ldots, \symRegister_{\symNumReg-1}$ and $\symNumReg = 2^{\symPrecision}$ where $\symPrecision\leq\symPrecision'$. This algorithm also allows to downsize an existing UltraLogLog by merging it into an empty UltraLogLog $(\symRegister_\symRegAddr = 0)$ with smaller precision parameter $\symPrecision<\symPrecision'$.
}
\label{alg:merge}
$\symRegAddrOther\gets 0$\;
\For{$\symRegAddr\gets 0$ \KwTo $\symNumReg-1$}{
$\symX \gets \FuncUncompress{$\symRegister_\symRegAddr$}\ \KwOr\ (\FuncUncompress{$\symRegister'_{\symRegAddrOther}$} \cdot 2^{\symPrecision' - \symPrecision})$\;
\Comment*[r]{bitwise or-operation, see \cref{alg:unpack} for $\FuncUncompress$}
$\symRegAddrOther\gets \symRegAddrOther + 1$\;
\For{$\symUpdateValOther\gets 1$ \KwTo $2^{\symPrecision' - \symPrecision}-1$}{
\If{$\symRegister'_{\symRegAddrOther} \neq 0$}{
$\symUpdateVal\gets\FuncNLZ{$\symUpdateValOther$} + \symPrecision' - \symPrecision - 63$
\Comment*[r]{$\symUpdateValOther$ is assumed to have 64 bits}
$\symX \gets \symX$ \KwOr $2^{\symUpdateVal+1}$\Comment*[r]{bitwise or-operation}
}
$\symRegAddrOther\gets \symRegAddrOther + 1$\;
}
\lIf(\Comment*[f]{see \cref{alg:pack} for $\FuncCompress$}){$\symX \neq 0$}{$\symRegister_{\symRegAddr} \gets $\FuncCompress{$\symX$}}
}
}

Extensive empirical tests \cite{Urban} have shown that the output of modern hash functions such as WyHash \cite{Yi}, Komihash \cite{Vaneev}, or PolymurHash \cite{Peters} can be considered like uniform random values. Otherwise, field-tested probabilistic data structures like \ac{HLL} would not work. This fact allows us to simplify the experiments and perform them without real or artificially generated data. Insertion of a new element can be simulated by simply generating a 64-bit random value to be used directly as the hash value of the inserted element in \cref{alg:insertion}.

\subsection{Estimation Error}
\label{sec:estimation_error}

To simulate the estimation error for a predefined distinct count value, the estimate is computed after updating the \ac{ULL} sketch using \cref{alg:insertion} with a corresponding number of random values and finally compared against the true distinct count. By repeating this process with many different random sequences, in our experiments 100\,000, the bias and the \ac{RMSE} can be empirically determined. However, this approach becomes computationally infeasible for distinct counts beyond 1 million and we need to switch to a different strategy.

After the first million of insertions, for which a random value was generated each time, we just generate the waiting time (the number of distinct count increments) until a register is processed with a certain update value the next time.
For each insertion, the probability that a register is updated with any possible value $\symUpdateVal\in[1,65-\symPrecision]$ is given by $1/(\symNumReg 2^{\min(\symUpdateVal, 64-\symPrecision)})$. Therefore, the number of distinct count increments until a register is updated with a specific value $\symUpdateVal$ the next time is geometrically distributed with corresponding success probability.
In this way, we determine the next update time for each register and for each possible update value. Since the same update value can only modify a register once, we do not need to consider further updates which might occur with the same value for the same register.
Knowing these $\symNumReg \times (65-\symPrecision)$ distinct count values in advance, where the state may change, enables us to make large distinct count increments, resulting in a huge speedup. This eventually allowed us to simulate the estimation error for distinct counts up to values of $10^{21}$ and also to test the presented estimators over the entire operating range.

\cref{fig:estimation-error} shows the empirically determined relative bias and \ac{RMSE} as well as the theoretical \ac{RMSE} given by $\sqrt{\symMVP/(8\symNumReg)}$ for the \ac{FGRA}, \ac{ML}, and the martingale estimator for precisions $\symPrecision\in\lbrace 8, 12, 16\rbrace$.
For intermediate distinct counts, for which the assumptions of our theoretical analysis hold, perfect agreement with theory is observed.
For small distinct counts, the difference in efficiency between \ac{ML} and \ac{FGRA} is more significant, but not particularly relevant in practice, as the estimation errors in both cases are well below the theoretically predicted errors.
Interestingly, the estimation error also decreases slightly near the end of the operating range, which is as predicted on the order of $2^{64} \approx 1.8\cdot 10^{19}$.
The estimators are essentially unbiased. The tiny bias which appears for small precisions for the \ac{ML} and \ac{FGRA} estimators can be ignored in practice as it is much smaller than the \ac{RMSE}.

Previous experiments have already confirmed that the proposed estimators for \ac{HLL} \cite{Ertl2017,Pettie2021a} and \ac{EHLL} \cite{Ohayon2021} do not undercut and at best reach the corresponding theoretical \acp{MVP} given by \eqref{equ:mvp_uncompressed} and \eqref{equ:mvp_martingale}.
Therefore, the perfect agreement with the theory for \ac{ULL} observed in our experiments finally proves the claimed and theoretically predicted improvements in storage efficiencies over the state of the art as shown in \cref{fig:mvp_lower_bound,fig:var_martingale}.

\begin{figure}[t]
  \centering
  \includegraphics[width=\linewidth]{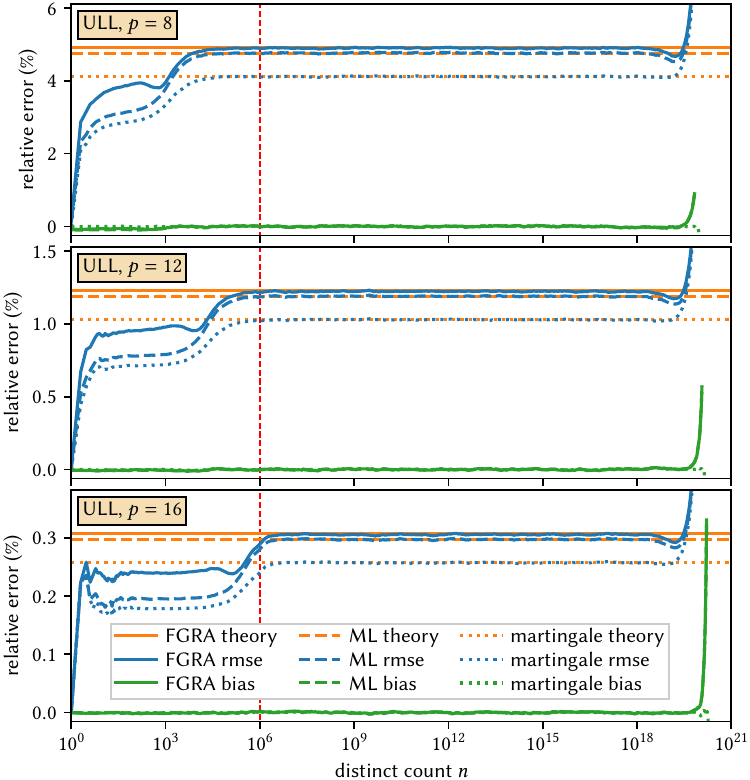}
  \caption{\boldmath The relative bias and the \acs*{RMSE} for the \acs*{FGRA}, \acs*{ML}, and the martingale estimator for precisions $\symPrecision\in\lbrace 8, 12, 16\rbrace$ obtained from 100\,000 simulation runs. The theoretically predicted errors perfectly match the experimental results. Individual insertions were simulated up to a distinct count of 10$^{6}$ before switching to the fast simulation strategy.}
  \label{fig:estimation-error}
\end{figure}

\subsection{Compression}

\begin{figure}[t]
  \centering
  \includegraphics[width=\linewidth]{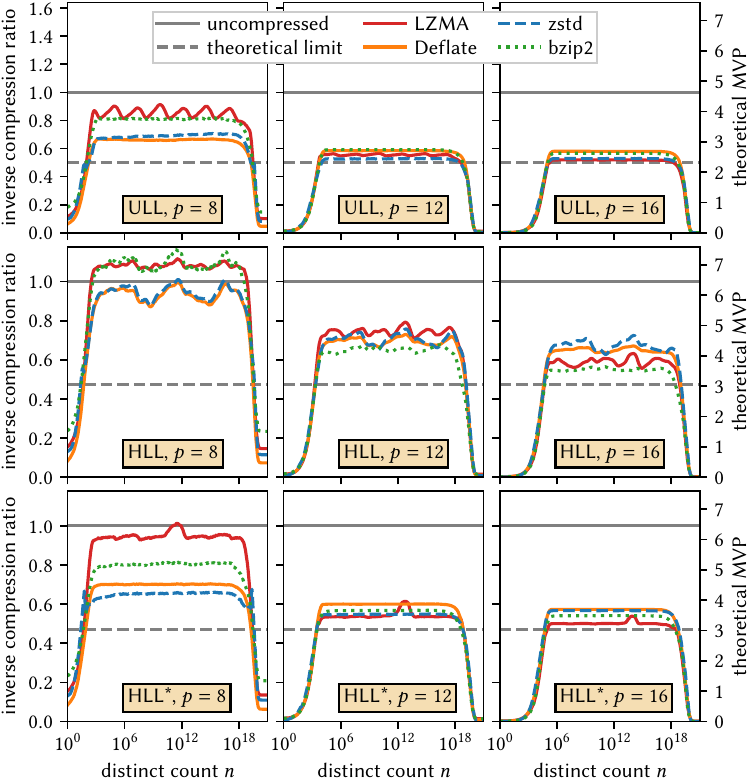}
  \caption{\boldmath The average inverse compression ratio for \ac{HLL} and \ac{ULL} over the distinct count for various standard compression algorithms. For \ac{HLL}*, the 6-bit registers are mapped to individual bytes before compression. The gray dashed line indicates the theoretical lower bound for intermediate distinct counts given by \eqref{equ:mvp_compressed}.}
  \label{fig:compression}
\end{figure}

To analyze the compressibility of the state, we generated 100 random sketches for predefined distinct counts and applied various standard compression algorithms (LZMA, Deflate, zstd, and bzip2) from the Apache Compress Java library \cite{ApacheCompress}. The corresponding average inverse compression ratios over the distinct count are shown in \cref{fig:compression} for \ac{HLL} and \ac{ULL}.
The results indicate that the algorithms generally work better for \ac{ULL} than \ac{HLL}, since the compressed size is closer to the theoretical limit given by the Shannon entropy \eqref{equ:shannon_entropy}. Interestingly, the compression for \ac{HLL} improves, if its 6-bit registers are first represented as individual bytes (cf. \ac{HLL}* in \cref{fig:compression}). Even though this leads to compression ratios that are sometimes better than for \ac{ULL}, \ac{ULL} is still more memory-efficient overall due to the significant lower \ac{MVP} when uncompressed.

The theoretical lower bound applies only to intermediate distinct counts under the validity of the simple model \eqref{equ:simple_register_pmf}. The observed high compressibility at small and large distinct counts is due to the large number of initial and saturated registers, respectively. Therefore, many \ac{HLL} implementations support a sparse mode that encodes only non-zero registers \cite{Heule2013}. We expect that this and other lossless compression techniques developed for \ac{HLL} \cite{ApacheDataSketches,Karppa2022} can also be applied to \ac{ULL}, but obviously at the cost of slower updates and memory reallocations.

\subsection{Performance}
\begin{figure}[t]
  \centering
  \includegraphics[width=\linewidth]{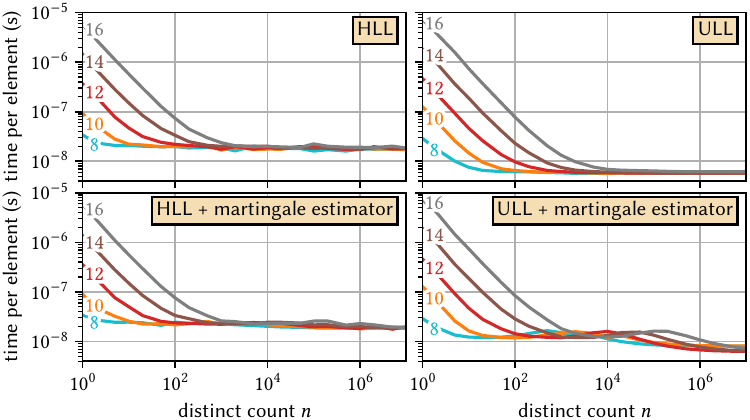}
  \caption{\boldmath Average insertion time per element when initializing a \acs*{HLL} or \acs*{ULL} sketch with $\symPrecision\in\lbrace8,10,12,14,16\rbrace$ and adding the given number of distinct elements.}
  \label{fig:add-performance}
\end{figure}

Since low processing costs are also critical for practical use, we measured the average time for inserting a given number of distinct elements into \ac{HLL} and \ac{ULL} sketches configured with precisions $\symPrecision\in\lbrace8,10,12,14,16\rbrace$. To reduce the impact of variable processor frequencies, Turbo Boost was disabled on the used Amazon EC2 c5.metal instance by setting the processor P-state to 1 \cite{AmazonPState}. The benchmarks were executed using OpenJDK 21.0.2.

The measurements shown in \cref{fig:add-performance} also include sketch initialization as well as the generation of random numbers which were used instead of hash values as described before.
The initialization costs dominate for small distinct counts, which are proportional to $\symNumReg = 2^\symPrecision$ due to the allocated register array. For large distinct counts the initialization costs can be neglected, and the average time per insertion converges to a value that is essentially independent of $\symPrecision$. There, it becomes apparent that \ac{HLL} insertions are significantly slower, mainly caused by the overhead of accessing the 6-bit registers packed in a byte array. If the insertions are accompanied by martingale estimator updates following \cref{alg:martingale}, the difference is less clear. Only for large distinct counts, where changes of register values become less frequent, \ac{ULL} is significantly faster again.

We also investigated the estimation costs for precisions $\symPrecision\in[8,16]$ and $\symCardinality \in \lbrace 1,2,5,10,20,50,\ldots, 10^7\rbrace$ as shown in \cref{fig:estimation-performance}. We considered the \ac{ML} estimator for both, the new \ac{FGRA} estimator for \ac{ULL}, and the \acf{CR} estimator for \ac{HLL} \cite{Ertl2017}. The latter corresponds to the \ac{GRA} estimator with $\symGRA=1$ which is for \ac{HLL} almost as efficient as the \ac{ML} and optimal \ac{GRA} estimators \cite{Wang2023}. \ac{ML} estimation is more expensive for \ac{ULL} than for \ac{HLL} for larger distinct counts. The \ac{CR} and the \ac{FGRA} estimator are most of the time significantly faster than their \ac{ML} counterparts.
The costs of the \ac{CR} estimator are roughly constant, while the \ac{FGRA} estimator peaks briefly before it falls back on the same level as the \ac{CR} estimator for equal $\symPrecision$. Our investigations showed that this peak is related to the more difficult branch prediction in \cref{alg:estimation}, when significant portions of registers have values from $\lbrace 0,4,8,10\rbrace$ and also values greater than or equal to $12$.

A fair comparison must take into account that an \ac{ULL} sketch with same precision generally leads to smaller errors. According to \eqref{equ:def_mvp} the theoretical relative errors for \ac{HLL} and \ac{ULL} are given by
$\sqrt{\symMVP^\text{(\acs*{HLL})}/(6\cdot 2^{\symPrecision^\text{(\acs*{HLL})}})}$ and $\sqrt{\symMVP^\text{(\acs*{ULL})}/(8\cdot 2^{\symPrecision^\text{(\acs*{ULL})}})}$. They are approximately equal if $\symPrecision^\text{(\acs*{HLL})} \approx \symPrecision^\text{(\acs*{ULL})} + 0.8$, which means that an \ac{ULL} with $\symPrecision=8$ is rather compared to a \ac{HLL} with $\symPrecision=9$.
\cref{fig:add-performance-over-error} shows the average estimation time for $\symCardinality=10^6$ and averaged over the cases $\symCardinality \in \lbrace1,2,5,10,20,50,\ldots, 10^7\rbrace$, which shows that \ac{FGRA} estimation from an \ac{ULL} is often faster (except for the peaks) than \ac{CR} estimation from an \ac{HLL} of equivalent precision.

\begin{figure}[t]
  \centering
  \includegraphics[width=\linewidth]{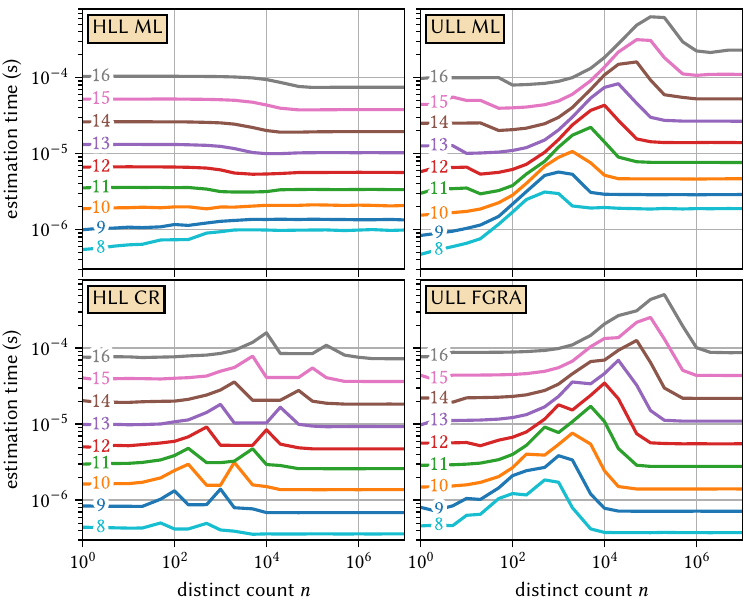}
  \caption{\boldmath Average estimation time over the true distinct count for \acs*{HLL} and \acs*{ULL} for various precisions $\symPrecision\in[8,16]$.}
  \label{fig:estimation-performance}
\end{figure}
\begin{figure}[t]
  \centering
  \includegraphics[width=\linewidth]{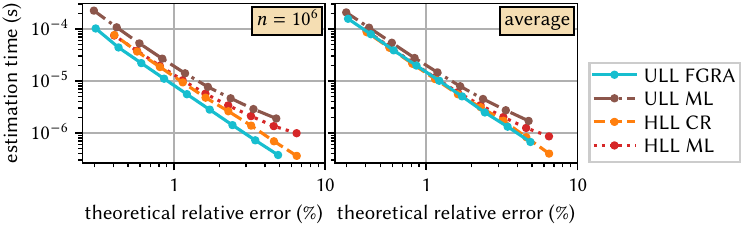}
  \caption{\boldmath
    Estimation time for $\symCardinality=10^6$ (left) and averaged over the cases $\symCardinality \in \lbrace1,2,5,10,20,50,\ldots, 10^7\rbrace$ (right) versus the theoretical relative error when varying $\symPrecision\in [8,16]$.}
  \label{fig:add-performance-over-error}
\end{figure}

\section{Future Work}
Although the \ac{ML} estimator has been shown to be efficient and achieves the Cram\'er-Rao bound, it is slower than the \ac{FGRA} estimator, which in turn has a worse efficiency of 94.6\%. Improving the \ac{ML} solver could be a way to achieve better estimation performance. Possibly, an alternative estimation method may also lead to a faster and even more efficient estimator.
Another interesting question is whether parameter choices with smaller \acp{MVP}, as discussed in \cref{sec:parameter-choice}, can be turned into practical data structures with efficient estimation algorithms. Finally, since \ac{HLL} has also been successfully applied for set similarity estimation \cite{Ertl2017, Nazi2018, Ertl2021}, what is the space efficiency of \ac{ULL} in this context, and can a fast and robust estimation algorithm be found for that as well?

\section{Conclusion}

We derived the \ac{ULL} sketch as a special case of a generalized data structure unifying \ac{HLL}, \ac{EHLL}, and \ac{PCSA}. The theoretically predicted space savings of 28\%, 24\%, and 17\% over \ac{HLL} when using the \ac{ML}, \ac{FGRA}, and martingale estimator, respectively, were perfectly matched by our experiments. The byte-sized registers lead to faster recording speed as well as better compressibility. Since \ac{ULL} also has the same properties as \ac{HLL} (constant-time insertions, idempotency, mergeability, reproducibility, reducibility), efficient and robust estimation with similar execution speed is possible, and even compatibility with \ac{HLL} can be achieved, we believe that \ac{ULL} has the potential to become the new standard algorithm for approximate distinct counting.

\bibliographystyle{ACM-Reference-Format}
\bibliography{bibliography}

\ifextended

  \clearpage

  \appendix

  \section{Proofs}
  \label{app:proofs}

  \begin{lemma}
    \label{lem:derivatives}
    If $\symZ_\symMaxUpdateVal := e^{-\frac{\symCardinality(\symBase-1)}{\symNumReg\symBase^{\symMaxUpdateVal}}}$ with $\symCardinality, \symNumReg > 0$ and $\symBase > 1$ the following identities hold:
    \begin{align}
      \symZ_\symMaxUpdateVal^\symBase
       & = \symZ_{\symMaxUpdateVal-1},\nonumber
      \\
      \ln\symZ_\symMaxUpdateVal
       & = -\frac{\symCardinality}{\symNumReg}\frac{\symBase-1}{\symBase^\symMaxUpdateVal}, \nonumber
      \\
      \frac{\partial}{\partial\symCardinality}\symZ_\symMaxUpdateVal
       & =\frac{1}{\symCardinality}\symZ_\symMaxUpdateVal\ln\symZ_\symMaxUpdateVal,
      \nonumber                                                                                       \\
      \frac{\partial}{\partial\symCardinality}\ln\symZ_\symMaxUpdateVal
       & =\frac{1}{\symCardinality}\ln\symZ_\symMaxUpdateVal,
      \nonumber                                                                                       \\
      \frac{\partial^2}{\partial\symCardinality^2}\ln\symZ_\symMaxUpdateVal
       & = 0,
      \nonumber                                                                                       \\
      \frac{\partial}{\partial\symCardinality}\ln(1-\symZ_\symMaxUpdateVal)
       & =
      -\frac{1}{\symCardinality}\frac{\symZ_\symMaxUpdateVal\ln\symZ_\symMaxUpdateVal}{1-\symZ_\symMaxUpdateVal},
      \nonumber                                                                                       \\
      \frac{\partial^2}{\partial\symCardinality^2}\ln(1-\symZ_\symMaxUpdateVal)
       & =
      -\frac{1}{\symCardinality^2}\frac{\symZ_\symMaxUpdateVal\ln^2 \symZ_\symMaxUpdateVal}{(1-\symZ_\symMaxUpdateVal)^2},
      \nonumber                                                                                       \\
      \frac{\partial^3}{\partial\symCardinality^3}\ln(1-\symZ_\symMaxUpdateVal)
       & =
      -\frac{1}{\symCardinality^3}\frac{\symZ_\symMaxUpdateVal(1+\symZ_\symMaxUpdateVal)\ln^3\symZ_\symMaxUpdateVal}{(1-\symZ_\symMaxUpdateVal)^3}.
      \nonumber
    \end{align}
  \end{lemma}
  \begin{proof}
    The proof is straightforward using basic calculus.
  \end{proof}

  \begin{lemma}
    \label{lem:approximation}
    If $\symZ_\symMaxUpdateVal := e^{-\frac{\symCardinality(\symBase-1)}{\symNumReg\symBase^{\symMaxUpdateVal}}}$ ($\symCardinality, \symNumReg > 0$ and $\symBase > 1$) and $\symSomeFunc$ is a smooth function on $(0,1)$ with $\symSomeFunc(0+)=0$ and $\symSomeFunc(1-)=0$,
    the series $\sum_{\symMaxUpdateVal=-\infty}^\infty \symSomeFunc(\symZ_{\symMaxUpdateVal})$ is a periodic function of $\log_\symBase(\symCardinality)$ with
    period 1. Furthermore, if the relative amplitude can be neglected, the following approximation can be used
    \begin{equation*}
      \sum_{\symMaxUpdateVal=-\infty}^\infty \symSomeFunc(\symZ_{\symMaxUpdateVal}) \approx
      \frac{1}{\ln\symBase}
      \int_{0}^\infty
      \frac{\symSomeFunc(e^{-\symY})}{\symY}d\symY
      =
      -\frac{1}{\ln\symBase}\int_{0}^1 \frac{\symSomeFunc(\symZ)}{\symZ\ln\symZ}d\symZ.
    \end{equation*}
  \end{lemma}
  \begin{proof}
    The periodicity of the series is shown by
    \begin{multline*}
      \sum_{\symMaxUpdateVal=-\infty}^\infty \symSomeFunc(\symZ_{\symMaxUpdateVal})
      =
      \sum_{\symMaxUpdateVal=-\infty}^\infty \symSomeFunc\!\left(e^{-\frac{\symCardinality(\symBase-1)}{\symNumReg\symBase^{\symMaxUpdateVal}}}\right)
      \\
      =
      \sum_{\symMaxUpdateVal=-\infty}^\infty \symSomeFunc\!\left(e^{-\frac{\symBase-1}{\symNumReg\symBase^{\symMaxUpdateVal-\log_\symBase(\symCardinality)}}}\right)
      =
      \sum_{\symMaxUpdateVal=-\infty}^\infty \symSomeFunc\!\left(e^{-\frac{\symBase-1}{\symNumReg\symBase^{\symMaxUpdateVal-(1+\log_\symBase(\symCardinality))}}}\right).
    \end{multline*}
    If its relative amplitude can be neglected, we can approximate
    \begin{multline*}
      \sum_{\symMaxUpdateVal=-\infty}^\infty \symSomeFunc(\symZ_{\symMaxUpdateVal})
      =
      \sum_{\symMaxUpdateVal=-\infty}^\infty \symSomeFunc\!\left(e^{-\frac{\symBase-1}{\symNumReg\symBase^{\symMaxUpdateVal-\log_\symBase(\symCardinality)}}}\right)
      \\
      \approx
      \int_{0}^1 \sum_{\symMaxUpdateVal=-\infty}^\infty \symSomeFunc\!\left(e^{-\frac{\symBase-1}{\symNumReg\symBase^{\symMaxUpdateVal-\symX}}}\right) d\symX
      =
      \int_{-\infty}^\infty\symSomeFunc\!\left(e^{-\frac{(\symBase-1)\symBase^{\symX}}{\symNumReg}}\right) d\symX
      \\
      =
      \frac{1}{\ln\symBase}
      \int_{0}^\infty
      \frac{\symSomeFunc(e^{-\symY})}{\symY}d\symY
      =
      -\frac{1}{\ln\symBase}\int_{0}^1 \frac{\symSomeFunc(\symZ)}{\symZ\ln\symZ}d\symZ.
    \end{multline*}
  \end{proof}

  \begin{lemma}
    \label{lem:identity_pmf}
    For the \ac{PMF} given in \eqref{equ:simple_register_pmf} the following identity holds:
    \begin{multline*}
      \sum_{\symIndexBit_1,\ldots,\symIndexBit_\symNumExtraBits\in\lbrace 0,1\rbrace}
      \symDensityRegisterSimple(\symMaxUpdateVal 2^\symNumExtraBits + \langle\symIndexBit_1\ldots\symIndexBit_{\symNumExtraBits}\rangle_2\vert\symCardinality)
      \prod_{\symIndexS=1}^\symNumExtraBits
      \symIndexBit_\symIndexS^{\symAlpha_\symIndexS}
      (1-\symIndexBit_\symIndexS)^{\symBeta_\symIndexS}
      \\ =
      \symZ_{\symMaxUpdateVal}^{\frac{1}{\symBase-1}}
      (1-\symZ_{\symMaxUpdateVal})
      \prod_{\symIndexS=1}^\symNumExtraBits
      (1-\symZ_{\symMaxUpdateVal-\symIndexS})^{\symAlpha_\symIndexS} \symZ_{\symMaxUpdateVal-\symIndexS}^{\symBeta_\symIndexS}
    \end{multline*}
    where $\symAlpha_\symIndexS, \symBeta_\symIndexS \in \lbrace 0, 1 \rbrace$, $\symAlpha_\symIndexS + \symBeta_\symIndexS \leq 1$, and $0^0$ is defined as 1.
    As special cases, we have
    \begin{align*}
       & \sum_{\symIndexBit_1,\ldots,\symIndexBit_\symNumExtraBits\in\lbrace 0,1\rbrace}
      \symDensityRegisterSimple(\symMaxUpdateVal 2^\symNumExtraBits + \langle\symIndexBit_1\ldots\symIndexBit_{\symNumExtraBits}\rangle_2\vert\symCardinality)
      \symIndexBit_\symIndexS
      =
      \symZ_{\symMaxUpdateVal}^{\frac{1}{\symBase-1}}
      (1-\symZ_{\symMaxUpdateVal})(1-\symZ_{\symMaxUpdateVal-\symIndexS}),
      \\
       & \sum_{\symIndexBit_1,\ldots,\symIndexBit_\symNumExtraBits\in\lbrace 0,1\rbrace}
      \symDensityRegisterSimple(\symMaxUpdateVal 2^\symNumExtraBits + \langle\symIndexBit_1\ldots\symIndexBit_{\symNumExtraBits}\rangle_2\vert\symCardinality)
      (1-\symIndexBit_\symIndexS)
      =
      \symZ_{\symMaxUpdateVal}^{\frac{1}{\symBase-1}}
      (1-\symZ_{\symMaxUpdateVal})\symZ_{\symMaxUpdateVal-\symIndexS},
      \\
       & \sum_{\symIndexBit_1,\ldots,\symIndexBit_\symNumExtraBits\in\lbrace 0,1\rbrace}
      \symDensityRegisterSimple(\symMaxUpdateVal 2^\symNumExtraBits + \langle\symIndexBit_1\ldots\symIndexBit_{\symNumExtraBits}\rangle_2\vert\symCardinality)
      =
      \symZ_{\symMaxUpdateVal}^{\frac{1}{\symBase-1}}
      (1-\symZ_{\symMaxUpdateVal}).
    \end{align*}
  \end{lemma}
  \begin{proof}
    \begin{align*}
       & \sum_{\symIndexBit_1,\ldots,\symIndexBit_\symNumExtraBits\in\lbrace 0,1\rbrace}
      \symDensityRegisterSimple(\symMaxUpdateVal 2^\symNumExtraBits + \langle\symIndexBit_1\ldots\symIndexBit_{\symNumExtraBits}\rangle_2\vert\symCardinality)
      \prod_{\symIndexS=1}^\symNumExtraBits
      \symIndexBit_\symIndexS^{\symAlpha_\symIndexS}
      (1-\symIndexBit_\symIndexS)^{\symBeta_\symIndexS}
      \\
       & =
      \sum_{\symIndexBit_1,\ldots,\symIndexBit_\symNumExtraBits\in\lbrace 0,1\rbrace}
      \symZ_{\symMaxUpdateVal}^{\frac{1}{\symBase-1}}
      (1-\symZ_{\symMaxUpdateVal})
      \prod_{\symIndexS = 1}^\symNumExtraBits
      \symZ_{\symMaxUpdateVal-\symIndexS}^{1-\symIndexBit_\symIndexS}
      (1- \symZ_{\symMaxUpdateVal-\symIndexS})^{ \symIndexBit_\symIndexS}
      \symIndexBit_\symIndexS^{\symAlpha_\symIndexS}
      (1-\symIndexBit_\symIndexS)^{\symBeta_\symIndexS}
      \\
       & =
      \symZ_{\symMaxUpdateVal}^{\frac{1}{\symBase-1}}
      (1-\symZ_{\symMaxUpdateVal})
      \prod_{\symIndexS = 1}^\symNumExtraBits
      \sum_{\symIndexBit_\symIndexS\in\lbrace 0,1\rbrace}
      \symZ_{\symMaxUpdateVal-\symIndexS}^{1-\symIndexBit_\symIndexS}
      (1- \symZ_{\symMaxUpdateVal-\symIndexS})^{ \symIndexBit_\symIndexS}
      \symIndexBit_\symIndexS^{\symAlpha_\symIndexS}
      (1-\symIndexBit_\symIndexS)^{\symBeta_\symIndexS}
      \\
       & =
      \symZ_{\symMaxUpdateVal}^{\frac{1}{\symBase-1}}
      (1-\symZ_{\symMaxUpdateVal})
      \prod_{\symIndexS = 1}^\symNumExtraBits
      \symZ_{\symMaxUpdateVal-\symIndexS}
      0^{\symAlpha_\symIndexS}
      +
      (1- \symZ_{\symMaxUpdateVal-\symIndexS})
      0^{\symBeta_\symIndexS}
      \\
       & =
      \symZ_{\symMaxUpdateVal}^{\frac{1}{\symBase-1}}
      (1-\symZ_{\symMaxUpdateVal})
      \prod_{\symIndexS=1}^\symNumExtraBits
      (1-\symZ_{\symMaxUpdateVal-\symIndexS})^{\symAlpha_\symIndexS} \symZ_{\symMaxUpdateVal-\symIndexS}^{\symBeta_\symIndexS}.
    \end{align*}
    The last equality can be verified by considering all possible cases $(\symAlpha_\symIndexS, \symBeta_\symIndexS)\in\lbrace(0, 0),(1, 0),(0, 1) \rbrace$.
  \end{proof}

  \begin{lemma}
    \label{lem:fisher}
    The Fisher information $\symFisher$ of random variables $\symRegister_0,\symRegister_1,\ldots,\symRegister_{\symNumReg-1}$ that are distributed according to \eqref{equ:simple_register_pmf}, can be approximated with \cref{lem:approximation} by
    \begin{equation*}
      \symFisher = \symExpectation\!\left(-\frac{\partial^2 \ln\symLikelihood}{\partial \symCardinality^2}\right)
      \approx
      \frac{\symNumReg}{\symCardinality^2}
      \frac{1}{\ln \symBase}\symZetaFunc\!\left(2, 1 + \frac{\symBase^{-\symNumExtraBits}}{\symBase-1}\right).
    \end{equation*}
    $\symZetaFunc$ denotes the Hurvitz zeta function $\symZetaFunc(\symX, \symY) := \sum_{\symMaxUpdateVal=0}^\infty (\symMaxUpdateVal + \symY)^{-\symX}= \frac{1}{\symGammaFunc(\symX)}\int_0^\infty \frac{\symZ^{\symX-1}e^{-\symY\symZ}}{1-e^{-\symZ}}d\symZ$.
  \end{lemma}
  The numerical evaluations presented in \cref{fig:fisher_information_max_relative_error} indicate that the relative error introduced by \cref{lem:approximation} is smaller than $10^{-4}$ for any $\symNumExtraBits$ and $\symBase\leq 2$ and vanishes as $\symBase\rightarrow 1$.
  \begin{proof}
    Since the registers are independent under the Poisson approximation, it is sufficient to consider a single register and multiply the result by $\symNumReg$
    \begin{align*}
      \symFisher & = \symExpectation\!\left(-\frac{\partial^2 \ln\symLikelihood}{\partial \symCardinality^2}\right)
      =
      -\symNumReg\sum_{\symRegister=-\infty}^\infty
      \symDensityRegisterSimple(\symRegister\vert\symCardinality)
      \frac{\partial^2 \ln\symDensityRegisterSimple(\symRegister\vert\symCardinality)}{\partial \symCardinality^2}
      \\
                 & =
      \frac{\symNumReg}{\symCardinality^2}
      \sum_{\symMaxUpdateVal=-\infty}^\infty
      \sum_{\symIndexBit_1,\ldots,\symIndexBit_{\symNumExtraBits}\in\lbrace 0,1\rbrace} \symDensityRegisterSimple( \symMaxUpdateVal 2^\symNumExtraBits + \langle\symIndexBit_1\ldots\symIndexBit_{\symNumExtraBits}
      \rangle_2\vert\symCardinality)
      \cdot
      \\
                 & \quad\cdot
      \left(\frac{\symZ_\symMaxUpdateVal \ln^2\symZ_\symMaxUpdateVal}{(1-\symZ_\symMaxUpdateVal)^2}
      +
      \sum_{\symIndexJ = 1}^\symNumExtraBits
      \symIndexBit_\symIndexJ
      \frac{\symZ_{\symMaxUpdateVal-\symIndexJ}\ln^2\symZ_{\symMaxUpdateVal-\symIndexJ}}{(1-\symZ_{\symMaxUpdateVal-\symIndexJ})^2}\right).
    \end{align*}
    Here we used \eqref{equ:simple_register_pmf} and \cref{lem:derivatives} for the second-order derivatives. With the help of \cref{lem:identity_pmf} we can write
    \begin{align*}
       & \symFisher=
      \frac{\symNumReg}{\symCardinality^2}
      \sum_{\symMaxUpdateVal=-\infty}^\infty
      \symZ_\symMaxUpdateVal^{\frac{1}{\symBase-1}}(1-\symZ_\symMaxUpdateVal)
      \left(\frac{\symZ_\symMaxUpdateVal \ln^2\symZ_\symMaxUpdateVal}{(1-\symZ_\symMaxUpdateVal)^2}
      +
      \sum_{\symIndexJ = 1}^\symNumExtraBits
      \frac{\symZ_{\symMaxUpdateVal-\symIndexJ}\ln^2\symZ_{\symMaxUpdateVal-\symIndexJ}}{1-\symZ_{\symMaxUpdateVal-\symIndexJ}}\right)
      \\
       & =
      \frac{\symNumReg}{\symCardinality^2}
      \sum_{\symMaxUpdateVal=-\infty}^\infty
      \left(\frac{\symZ_\symMaxUpdateVal^{\frac{\symBase}{\symBase-1}}\ln^2\symZ_\symMaxUpdateVal
      }{1-\symZ_\symMaxUpdateVal}
      +
      \sum_{\symIndexJ = 1}^\symNumExtraBits
      \frac{\symZ_\symMaxUpdateVal^{\frac{1}{\symBase-1}}(1-\symZ_\symMaxUpdateVal)\symZ_{\symMaxUpdateVal-\symIndexJ} \ln^2\symZ_{\symMaxUpdateVal-\symIndexJ}}{1-\symZ_{\symMaxUpdateVal-\symIndexJ}}\right)
      \\
       & =
      \frac{\symNumReg}{\symCardinality^2}
      \sum_{\symMaxUpdateVal=-\infty}^\infty
      \left(\frac{\symZ_\symMaxUpdateVal^{\frac{1}{\symBase-1}}}{1-\symZ_\symMaxUpdateVal}
      +
      \sum_{\symIndexJ = 1}^\symNumExtraBits
      \frac{\symZ_{\symMaxUpdateVal+\symIndexJ}^\frac{1}{\symBase-1}-\symZ_{\symMaxUpdateVal+\symIndexJ-1}^\frac{1}{\symBase-1}}{1-\symZ_{\symMaxUpdateVal}}\right)
      \symZ_\symMaxUpdateVal \ln^2\symZ_\symMaxUpdateVal
      \\
       & =
      \frac{\symNumReg}{\symCardinality^2}
      \sum_{\symMaxUpdateVal=-\infty}^\infty
      \symZ_{\symMaxUpdateVal+\symNumExtraBits}^\frac{1}{\symBase-1}
      \frac{\symZ_\symMaxUpdateVal \ln^2\symZ_\symMaxUpdateVal}{1-\symZ_\symMaxUpdateVal}
      =
      \frac{\symNumReg}{\symCardinality^2}
      \sum_{\symMaxUpdateVal=-\infty}^\infty
      \frac{\symZ_\symMaxUpdateVal^{1+ \frac{\symBase^{-\symNumExtraBits}}{\symBase-1}}\ln^2 \symZ_\symMaxUpdateVal}{1-\symZ_\symMaxUpdateVal}.
    \end{align*}
    Using \cref{lem:approximation} with $\symSomeFunc(\symX) = \frac{\symX^{1+ \frac{\symBase^{-\symNumExtraBits}}{\symBase-1}}\ln^2 \symX}{1-\symX}$ finally gives
    \begin{align*}
      \symFisher & \approx
      \frac{\symNumReg}{\symCardinality^2}
      \frac{1}{\ln \symBase}
      \int_{0}^{\infty}
      \frac{e^{-\symY(1+ \frac{\symBase^{-\symNumExtraBits}}{\symBase-1})} \symY}{1-e^{-\symY}} d\symY
      =
      \frac{\symNumReg}{\symCardinality^2}
      \frac{1}{\ln \symBase}\symZetaFunc\!\left(2, 1 + \frac{\symBase^{-\symNumExtraBits}}{\symBase-1}\right).
    \end{align*}
  \end{proof}

  \begin{figure}
    \centering
    \includegraphics[width=\linewidth]{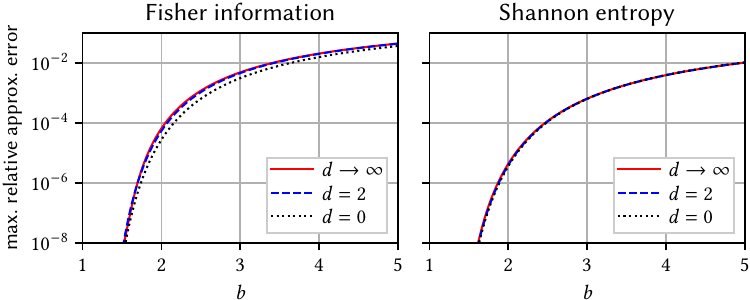}
    \caption{\boldmath Maximum relative approximation errors introduced by \cref{lem:approximation} for the Fisher information $\symFisher$ and the Shannon entropy $\symShannon$ over $\symBase$ for various values of $\symNumExtraBits$.}
    \label{fig:fisher_information_max_relative_error}
  \end{figure}

  \begin{lemma}
    \label{lem:shannon}
    The Shannon entropy $\symShannon$ of random variables $\symRegister_0,\symRegister_1,\ldots,\symRegister_{\symNumReg-1}$ that are distributed according to \eqref{equ:simple_register_pmf} can be approximated with \cref{lem:approximation} by
    \begin{equation*}
      \symShannon
      \approx
      \frac{\symNumReg}{(\ln 2)(\ln \symBase) }
      \left(
      \left(1+\frac{\symBase^{-\symNumExtraBits}}{\symBase-1}\right)^{\!-1}+
      \int_{0}^1 \symZ^{\frac{\symBase^{-\symNumExtraBits}}{\symBase-1}}
      \frac{
        (1-\symZ)
        \ln(1-\symZ)
      }{\symZ\ln(\symZ)}
      d\symZ
      \right).
    \end{equation*}
  \end{lemma}
  The numerical evaluations presented in \cref{fig:fisher_information_max_relative_error} indicate that the relative error introduced by \cref{lem:approximation} is smaller than $10^{-5}$ for any $\symNumExtraBits$ and $\symBase\leq 2$ and vanishes as $\symBase\rightarrow 1$.
  \begin{proof}
    Since the registers are independent under the Poisson approximation, we can write the expectation of $\ln \symLikelihood$ as
    \begin{align*}
       &
      \symExpectation(\ln \symLikelihood)
      =
      \symNumReg\sum_{\symRegister=-\infty}^\infty
      \symDensityRegisterSimple(\symRegister\vert\symCardinality)
      \ln\symDensityRegisterSimple(\symRegister\vert\symCardinality)
      \\
       & =
      \symNumReg
      \sum_{\symMaxUpdateVal=-\infty}^\infty
      \sum_{\symIndexBit_1,\ldots,\symIndexBit_{\symNumExtraBits}\in\lbrace 0,1\rbrace} \symDensityRegisterSimple( \symMaxUpdateVal 2^\symNumExtraBits + \langle\symIndexBit_1\ldots\symIndexBit_{\symNumExtraBits}
      \rangle_2\vert\symCardinality)
      \cdot
      \\
       & \quad\textstyle\cdot\ln\!\left(
      \symZ_\symMaxUpdateVal^{\frac{1}{\symBase-1}}(1-\symZ_\symMaxUpdateVal)
      \prod_{\symIndexJ = 1}^\symNumExtraBits
      \symZ_{\symMaxUpdateVal-\symIndexJ}^{1-\symIndexBit_\symIndexJ}
        (1-\symZ_{\symMaxUpdateVal-\symIndexJ})^{ \symIndexBit_\symIndexJ}
      \right)
      \\
       & =
      \symNumReg
      \sum_{\symMaxUpdateVal=-\infty}^\infty
      \sum_{\symIndexBit_1,\ldots,\symIndexBit_{\symNumExtraBits}\in\lbrace 0,1\rbrace} \symDensityRegisterSimple( \symMaxUpdateVal 2^\symNumExtraBits + \langle\symIndexBit_1\ldots\symIndexBit_{\symNumExtraBits}
      \rangle_2\vert\symCardinality)
      \cdot
      \\
       & \quad\cdot
      \left(\begin{array}{@{}l}
                \frac{\ln(\symZ_\symMaxUpdateVal)}{\symBase-1}
                +
                \ln(1-\symZ_\symMaxUpdateVal)
                \\
                +\sum_{\symIndexJ = 1}^\symNumExtraBits
                (1-\symIndexBit_\symIndexJ)\ln(\symZ_{\symMaxUpdateVal-\symIndexJ})
                +
                \symIndexBit_\symIndexJ\ln(1-\symZ_{\symMaxUpdateVal-\symIndexJ})
              \end{array}\right).
    \end{align*}
    Here we used \eqref{equ:simple_register_pmf}. \cref{lem:identity_pmf} allows to write
    \begin{align*}
       &
      \symExpectation(\ln \symLikelihood)
      \\
       & =
      \symNumReg\sum_{\symMaxUpdateVal=-\infty}^\infty
      \symZ_\symMaxUpdateVal^{\frac{1}{\symBase-1}}(1-\symZ_\symMaxUpdateVal)
      \left(
      \frac{\ln(\symZ_\symMaxUpdateVal)}{\symBase-1}
      +
      \ln(1-\symZ_\symMaxUpdateVal)
      \right)
      \\
       & \quad +
      \symNumReg\sum_{\symMaxUpdateVal=-\infty}^\infty\sum_{\symIndexJ = 1}^\symNumExtraBits
      \symZ_\symMaxUpdateVal^{\frac{1}{\symBase-1}}(1-\symZ_\symMaxUpdateVal)
      \left(\begin{array}{@{}l}
                \symZ_{\symMaxUpdateVal-\symIndexJ}\ln(\symZ_{\symMaxUpdateVal-\symIndexJ})
                \\
                +
                (1-\symZ_{\symMaxUpdateVal-\symIndexJ})\ln(1-\symZ_{\symMaxUpdateVal-\symIndexJ})
              \end{array}\right)
      \\
       & =
      \symNumReg\sum_{\symMaxUpdateVal=-\infty}^\infty
      \symZ_\symMaxUpdateVal^{\frac{1}{\symBase-1}}(1-\symZ_\symMaxUpdateVal)
      \left(
      \frac{\ln(\symZ_\symMaxUpdateVal)}{\symBase-1}
      +
      \ln(1-\symZ_\symMaxUpdateVal)
      \right)
      \\
       & \quad +
      \symNumReg\sum_{\symMaxUpdateVal=-\infty}^\infty\sum_{\symIndexJ = 1}^\symNumExtraBits
      (\symZ_{\symMaxUpdateVal+\symIndexJ}^{\frac{1}{\symBase-1}}-\symZ_{\symMaxUpdateVal+\symIndexJ-1}^{\frac{1}{\symBase-1}})
      \left(\begin{array}{@{}l}
                \symZ_{\symMaxUpdateVal}\ln(\symZ_{\symMaxUpdateVal})
                \\+
                (1-\symZ_{\symMaxUpdateVal})\ln(1-\symZ_{\symMaxUpdateVal})
              \end{array}\right)
      \\
       & =
      \symNumReg\sum_{\symMaxUpdateVal=-\infty}^\infty
      \symZ_\symMaxUpdateVal^{\frac{1}{\symBase-1}}(1-\symZ_\symMaxUpdateVal)
      \left(
      \frac{\ln(\symZ_\symMaxUpdateVal)}{\symBase-1}
      +
      \ln(1-\symZ_\symMaxUpdateVal)
      \right)
      \\
       & \quad +
      \symNumReg\sum_{\symMaxUpdateVal=-\infty}^\infty(\symZ_{\symMaxUpdateVal+\symNumExtraBits}^{\frac{1}{\symBase-1}}-\symZ_{\symMaxUpdateVal}^{\frac{1}{\symBase-1}})
      \left(
      \symZ_{\symMaxUpdateVal}\ln(\symZ_{\symMaxUpdateVal})
      +
      (1-\symZ_{\symMaxUpdateVal})\ln(1-\symZ_{\symMaxUpdateVal})
      \right)
      \\
       & =
      \symNumReg\sum_{\symMaxUpdateVal=-\infty}^\infty
      \symZ_\symMaxUpdateVal^{\frac{1}{\symBase-1}}(1-\symZ_\symMaxUpdateVal)
      \frac{\ln(\symZ_\symMaxUpdateVal)}{\symBase-1}
      -
      \symZ_{\symMaxUpdateVal}^{\frac{1}{\symBase-1}}
      \symZ_{\symMaxUpdateVal}\ln(\symZ_{\symMaxUpdateVal})
      \\
       & \quad + \symNumReg\sum_{\symMaxUpdateVal=-\infty}^\infty
      \symZ_{\symMaxUpdateVal+\symNumExtraBits}^{\frac{1}{\symBase-1}}
      \left(
      \symZ_{\symMaxUpdateVal}\ln(\symZ_{\symMaxUpdateVal})
      +
      (1-\symZ_{\symMaxUpdateVal})\ln(1-\symZ_{\symMaxUpdateVal})
      \right)
      \\
       & =
      \symNumReg\sum_{\symMaxUpdateVal=-\infty}^\infty
      \symZ_\symMaxUpdateVal^{\frac{1}{\symBase-1}}(1-\symZ_\symMaxUpdateVal)
      \frac{\ln(\symZ_\symMaxUpdateVal)}{\symBase-1}
      -
      \symZ_{\symMaxUpdateVal}^{\frac{1}{\symBase-1}}
      \symZ_{\symMaxUpdateVal}\ln(\symZ_{\symMaxUpdateVal})
      \\
       & \quad + \symNumReg\sum_{\symMaxUpdateVal=-\infty}^\infty
      \symZ_{\symMaxUpdateVal}^{\frac{\symBase^{-\symNumExtraBits}}{\symBase-1}}
      \left(
      \symZ_{\symMaxUpdateVal}\ln(\symZ_{\symMaxUpdateVal})
      +
      (1-\symZ_{\symMaxUpdateVal})\ln(1-\symZ_{\symMaxUpdateVal})
      \right)
      \\
       & =
      \symNumReg\sum_{\symMaxUpdateVal=-\infty}^\infty
      \frac{\symZ_\symMaxUpdateVal^{\frac{1}{\symBase-1}}\ln(\symZ_\symMaxUpdateVal)(1-\symZ_\symMaxUpdateVal\symBase)}{\symBase-1}
      \\
       & \quad +
      \symNumReg\sum_{\symMaxUpdateVal=-\infty}^\infty
      \symZ_{\symMaxUpdateVal}^{\frac{\symBase^{-\symNumExtraBits}}{\symBase-1}}
      \left(
      \symZ_{\symMaxUpdateVal}\ln(\symZ_{\symMaxUpdateVal})
      +
      (1-\symZ_{\symMaxUpdateVal})\ln(1-\symZ_{\symMaxUpdateVal})
      \right).
    \end{align*}
    Using the definition of the Shannon entropy $\symShannon = -\symExpectation(\log_2 \symLikelihood)$ and applying \cref{lem:approximation} with function $\symSomeFunc(\symX) =
      \frac{\symX^\frac{1}{\symBase-1}\ln(\symX)(1-\symX\symBase)}{\symBase-1}+
      \symX^\frac{\symBase^{-\symNumExtraBits}}{\symBase-1}\left(
      \symX\ln(\symX) + (1-\symX)\ln(1-\symX)
      \right)$ finally gives
    \begin{align*}
      \symShannon
       & = -\symExpectation(\log_2 \symLikelihood) = -\frac{1}{\ln 2} \symExpectation(\ln \symLikelihood) \\
       & \approx
      \textstyle
      \frac{\symNumReg}{(\ln\symBase)(\ln 2)}
      \left(
      \int_0^1
      \frac{\symZ^{\frac{1}{\symBase-1}-1}(1-\symZ\symBase)}{\symBase-1}
      +
      \symZ^{\frac{\symBase^{-\symNumExtraBits}}{\symBase-1}}
      \left(
      1 + \frac{(1-\symZ)\ln(1-\symZ)}{\symZ\ln\symZ}
      \right)
      d\symZ
      \right)
      \\
       & =
      \textstyle
      \frac{\symNumReg}{(\ln 2) (\ln \symBase)}
      \left(
      \left(1+\frac{\symBase^{-\symNumExtraBits}}{\symBase-1}\right)^{\!-1}+
      \int_{0}^1 \symZ^{\frac{\symBase^{-\symNumExtraBits}}{\symBase-1}}
      \frac{
        (1-\symZ)
        \ln(1-\symZ)
      }{\symZ\ln \symZ}
      d\symZ
      \right).
    \end{align*}
  \end{proof}

  \begin{lemma}
    \label{lem:first_order_ml_bias}
    The first-order bias of the \ac{ML} estimate $\symCardinalityEstimatorML$ for the distinct count $\symCardinality$ under the assumption of the simplified \ac{PMF} \eqref{equ:simple_register_pmf} is roughly given by
    \begin{equation*}
      \symBias(\symCardinalityEstimatorML) \approx \frac{\symCardinality}{\symNumReg} (\ln \symBase)
      \left(
      1 +\frac{2\symBase^{-\symNumExtraBits}}{\symBase-1}
      \right)
      \frac{
        \symZetaFunc\!\left(3, 1 + \frac{\symBase^{-\symNumExtraBits}}{\symBase-1}\right)}{\left(\symZetaFunc\!\left(2, 1 + \frac{\symBase^{-\symNumExtraBits}}{\symBase-1}\right)\right)^2}
    \end{equation*}
    where $\symZetaFunc$ denotes the Hurvitz zeta function as introduced in \cref{lem:fisher}.
    The bias-corrected \ac{ML} estimator is therefore given by
    \begin{equation*}
      \symCardinalityEstimatorML \left(1 + \frac{1}{\symNumReg} (\ln \symBase)
      \left(
      1 +\frac{2\symBase^{-\symNumExtraBits}}{\symBase-1}
      \right)
      \frac{
        \symZetaFunc\!\left(3, 1 + \frac{\symBase^{-\symNumExtraBits}}{\symBase-1}\right)}{\left(\symZetaFunc\!\left(2, 1 + \frac{\symBase^{-\symNumExtraBits}}{\symBase-1}\right)\right)^2}\right)^{\!-1}.
    \end{equation*}
  \end{lemma}
  \begin{proof}
    The first-order bias of the \ac{ML} estimator is given by \cite{Cox1968}
    \begin{equation*}
      \symBias(\symCardinalityEstimatorML)
      \approx
      \frac{1}{\symFisher^{2}}
      \symExpectation\!\left(
      \frac{1}{2}\frac{\partial^3 \ln\symLikelihood}{\partial \symCardinality^3}
      +\frac{\partial^2 \ln\symLikelihood}{\partial \symCardinality^2}\frac{\partial \ln\symLikelihood}{\partial \symCardinality}
      \right).
    \end{equation*}
    We have
    \begin{align*}
       & \symExpectation\!\left(\frac{\partial^3 \ln\symLikelihood}{\partial \symCardinality^3}\right)
      =
      \symNumReg\sum_{\symRegister=-\infty}^\infty
      \symDensityRegisterSimple(\symRegister\vert\symCardinality)
      \frac{\partial^3 \ln\symDensityRegisterSimple(\symRegister\vert\symCardinality)}{\partial \symCardinality^3}
      \\
       & =
      -\frac{\symNumReg}{\symCardinality^3}
      \sum_{\symMaxUpdateVal=-\infty}^\infty
      \sum_{\symIndexBit_1,\ldots,\symIndexBit_{\symNumExtraBits}\in\lbrace 0,1\rbrace} \symDensityRegisterSimple(\symMaxUpdateVal2^\symNumExtraBits + \langle\symIndexBit_1\ldots\symIndexBit_{\symNumExtraBits}
      \rangle_2\vert\symCardinality)
      \cdot
      \\
       & \quad \cdot\left(\frac{\symZ_\symMaxUpdateVal(1+\symZ_\symMaxUpdateVal)\ln^3 \symZ_\symMaxUpdateVal}{(1-\symZ_\symMaxUpdateVal)^3}
      +
      \sum_{\symIndexJ = 1}^\symNumExtraBits
      \symIndexBit_\symIndexJ\frac{\symZ_{\symMaxUpdateVal-\symIndexJ}(1+\symZ_{\symMaxUpdateVal-\symIndexJ})\ln^3\symZ_{\symMaxUpdateVal-\symIndexJ}}{(1-\symZ_{\symMaxUpdateVal-\symIndexJ})^3}\right).
    \end{align*}
    Here we used \eqref{equ:simple_register_pmf} and \cref{lem:derivatives} for the third-order derivatives. Using \cref{lem:identity_pmf} leads to
    \begin{align*}
       & \symExpectation\!\left(\frac{\partial^3 \ln\symLikelihood}{\partial \symCardinality^3}\right)
      =
      -\frac{\symNumReg}{\symCardinality^3}
      \sum_{\symMaxUpdateVal=-\infty}^\infty \symZ_\symMaxUpdateVal^{\frac{1}{\symBase-1}}(1-\symZ_\symMaxUpdateVal)
      \cdot
      \\
       & \quad
      \cdot
      \left(\frac{\symZ_\symMaxUpdateVal(1+\symZ_\symMaxUpdateVal)\ln^3 \symZ_\symMaxUpdateVal}{(1-\symZ_\symMaxUpdateVal)^3}
      +
      \sum_{\symIndexJ = 1}^\symNumExtraBits
      \frac{\symZ_{\symMaxUpdateVal-\symIndexJ}(1+\symZ_{\symMaxUpdateVal-\symIndexJ})\ln^3 \symZ_{\symMaxUpdateVal-\symIndexJ}}{(1-\symZ_{\symMaxUpdateVal-\symIndexJ})^2}\right)
      \\
       & =
      -\frac{\symNumReg}{\symCardinality^3}
      \sum_{\symMaxUpdateVal=-\infty}^\infty
      \left(
      \begin{array}{@{}l}
          \frac{\symZ_\symMaxUpdateVal^{\frac{\symBase}{\symBase-1}}(1+\symZ_\symMaxUpdateVal)\ln^3 \symZ_\symMaxUpdateVal}{(1-\symZ_\symMaxUpdateVal)^2}
          \\
          +
          \sum_{\symIndexJ = 1}^\symNumExtraBits
          (\symZ_{\symMaxUpdateVal+\symIndexJ}^{\frac{1}{\symBase-1}}-\symZ_{\symMaxUpdateVal+\symIndexJ-1}^{\frac{1}{\symBase-1}})
          \frac{\symZ_{\symMaxUpdateVal}(1+\symZ_{\symMaxUpdateVal})\ln^3 \symZ_\symMaxUpdateVal}{(1-\symZ_{\symMaxUpdateVal})^2}
        \end{array}
      \right)
      \\
       & =
      -\frac{\symNumReg}{\symCardinality^3}
      \sum_{\symMaxUpdateVal=-\infty}^\infty
      \left(
      \symZ_\symMaxUpdateVal^{\frac{\symBase}{\symBase-1}}
      +
      (\symZ_{\symMaxUpdateVal+\symNumExtraBits}^{\frac{1}{\symBase-1}}-\symZ_{\symMaxUpdateVal}^{\frac{1}{\symBase-1}})
      \symZ_{\symMaxUpdateVal}
      \right)
      \frac{(1+\symZ_\symMaxUpdateVal)\ln^3\symZ_\symMaxUpdateVal}{(1-\symZ_\symMaxUpdateVal)^2}
      \\
       & =
      -\frac{\symNumReg}{\symCardinality^3}
      \sum_{\symMaxUpdateVal=-\infty}^\infty
      \frac{\symZ_{\symMaxUpdateVal}^{1 + \frac{\symBase^{-\symNumExtraBits}}{\symBase-1}}(1+\symZ_\symMaxUpdateVal)\ln^3\symZ_\symMaxUpdateVal}{(1-\symZ_\symMaxUpdateVal)^2}.
    \end{align*}
    Furthermore, we have
    \begin{align*}
       & \symExpectation\!\left(\frac{\partial \ln\symLikelihood}{\partial \symCardinality} \frac{\partial^2 \ln\symLikelihood}{\partial \symCardinality^2}\right)
      \\
       & =
      \symNumReg\sum_{\symRegister=-\infty}^\infty
      \symDensityRegisterSimple(\symRegister\vert\symCardinality)
      \frac{\partial \ln\symDensityRegisterSimple(\symRegister\vert\symCardinality)}{\partial \symCardinality}
      \frac{\partial^2 \ln\symDensityRegisterSimple(\symRegister\vert\symCardinality)}{\partial \symCardinality^2}
      \\
       & =
      -\frac{\symNumReg}{\symCardinality^3}
      \sum_{\symMaxUpdateVal=-\infty}^\infty
      \sum_{\symIndexBit_1,\ldots,\symIndexBit_{\symNumExtraBits}\in\lbrace 0,1\rbrace} \symDensityRegisterSimple(\symMaxUpdateVal2^\symNumExtraBits + \langle\symIndexBit_1\ldots\symIndexBit_{\symNumExtraBits}
      \rangle_2\vert\symCardinality)
      \cdot
      \\
       & \quad
      \cdot\left(\frac{\ln\symZ_\symMaxUpdateVal}{\symBase-1}-\frac{\symZ_\symMaxUpdateVal\ln\symZ_\symMaxUpdateVal}{1-\symZ_\symMaxUpdateVal}
      +
      \sum_{\symIndexJ = 1}^\symNumExtraBits
      (1-\symIndexBit_\symIndexJ)
      \ln\symZ_{\symMaxUpdateVal - \symIndexJ}
      -
      \symIndexBit_\symIndexJ
      \frac{\symZ_{\symMaxUpdateVal - \symIndexJ}\ln\symZ_{\symMaxUpdateVal - \symIndexJ}}{1-\symZ_{\symMaxUpdateVal - \symIndexJ}}\right)
      \cdot
      \\
       & \quad\cdot
      \left(\frac{\symZ_\symMaxUpdateVal \ln^2\symZ_\symMaxUpdateVal}{(1-\symZ_\symMaxUpdateVal)^2}
      +
      \sum_{\symIndexJ = 1}^\symNumExtraBits
      \symIndexBit_\symIndexJ
      \frac{\symZ_{\symMaxUpdateVal-\symIndexJ}\ln^2\symZ_{\symMaxUpdateVal-\symIndexJ}}{(1-\symZ_{\symMaxUpdateVal-\symIndexJ})^2}\right).
    \end{align*}
    Here we used \eqref{equ:simple_register_pmf} and \cref{lem:derivatives} for the derivatives. Using again \cref{lem:identity_pmf} together with $\symIndexBit_\symIndexJ (1- \symIndexBit_\symIndexJ) = 0$ and $\symIndexBit_\symIndexJ^2 = \symIndexBit_\symIndexJ$ because $\symIndexBit_\symIndexJ\in\lbrace0,1\rbrace$ gives
    \begin{align*}
       & \symExpectation\!\left(\frac{\partial \ln\symLikelihood}{\partial \symCardinality} \frac{\partial^2 \ln\symLikelihood}{\partial \symCardinality^2}\right)
      \\
       & =
      -\frac{\symNumReg}{\symCardinality^3}
      \sum_{\symMaxUpdateVal=-\infty}^\infty
      \symZ_\symMaxUpdateVal^{\frac{1}{\symBase-1}}(1-\symZ_\symMaxUpdateVal)
      \left(\frac{\ln\symZ_\symMaxUpdateVal}{\symBase-1}-\frac{\symZ_\symMaxUpdateVal\ln\symZ_\symMaxUpdateVal}{1-\symZ_\symMaxUpdateVal}
      \right)
      \cdot
      \\
       & \quad\cdot
      \left(\frac{\symZ_\symMaxUpdateVal \ln^2 \symZ_\symMaxUpdateVal}{(1-\symZ_\symMaxUpdateVal)^2}
      +
      \sum_{\symIndexJ = 1}^\symNumExtraBits
      \frac{\symZ_{\symMaxUpdateVal-\symIndexJ}\ln^2 \symZ_{\symMaxUpdateVal-\symIndexJ}}{1-\symZ_{\symMaxUpdateVal-\symIndexJ}}\right)
      \\
       & \quad+
      \frac{\symNumReg}{\symCardinality^3}
      \sum_{\symMaxUpdateVal=-\infty}^\infty
      \symZ_\symMaxUpdateVal^{\frac{1}{\symBase-1}}(1-\symZ_\symMaxUpdateVal)
      \sum_{\symIndexJ = 1}^\symNumExtraBits
      \frac{\symZ^2_{\symMaxUpdateVal-\symIndexJ}\ln^3 \symZ_{\symMaxUpdateVal-\symIndexJ}}{(1-\symZ_{\symMaxUpdateVal-\symIndexJ})^2}
      \\
       & =
      \frac{\symNumReg}{\symCardinality^3}
      \sum_{\symMaxUpdateVal=-\infty}^\infty
      \symZ_\symMaxUpdateVal^{\frac{1}{\symBase-1}}(1-\symZ_\symMaxUpdateVal)
      \left(-\frac{\ln\symZ_\symMaxUpdateVal}{\symBase-1}+\frac{\symZ_\symMaxUpdateVal\ln\symZ_\symMaxUpdateVal}{1-\symZ_\symMaxUpdateVal}
      \right)
      \cdot
      \\
       & \quad \cdot
      \left(\frac{\symZ_\symMaxUpdateVal \ln^2 \symZ_\symMaxUpdateVal}{(1-\symZ_\symMaxUpdateVal)^2}
      +
      \sum_{\symIndexJ = 1}^\symNumExtraBits
      \frac{\symZ_{\symMaxUpdateVal-\symIndexJ}\ln^2 \symZ_{\symMaxUpdateVal-\symIndexJ}}{1-\symZ_{\symMaxUpdateVal-\symIndexJ}}\right)
      \\
       & \quad +
      \frac{\symNumReg}{\symCardinality^3}
      \sum_{\symMaxUpdateVal=-\infty}^\infty
      \sum_{\symIndexJ = 1}^\symNumExtraBits
      (\symZ_{\symMaxUpdateVal+\symIndexJ}^{\frac{1}{\symBase-1}}-\symZ_{\symMaxUpdateVal+\symIndexJ-1}^{\frac{1}{\symBase-1}})
      \frac{\symZ^2_{\symMaxUpdateVal}\ln^3 \symZ_{\symMaxUpdateVal}}{(1-\symZ_{\symMaxUpdateVal})^2}
      \\
       & =
      \frac{\symNumReg}{\symCardinality^3}
      \sum_{\symMaxUpdateVal=-\infty}^\infty
      \left(\begin{array}{@{}l}
                \symZ_\symMaxUpdateVal^{\frac{1}{\symBase-1}}
                \frac{\symZ_\symMaxUpdateVal^2 \ln^3 \symZ_\symMaxUpdateVal}{(1-\symZ_\symMaxUpdateVal)^2}
                +
                \symZ_\symMaxUpdateVal^{\frac{1}{\symBase-1}}
                \symZ_\symMaxUpdateVal(\ln\symZ_\symMaxUpdateVal)
                \left(\sum_{\symIndexJ = 1}^\symNumExtraBits
                \frac{\symZ_{\symMaxUpdateVal-\symIndexJ}\ln^2 \symZ_{\symMaxUpdateVal-\symIndexJ}}{1-\symZ_{\symMaxUpdateVal-\symIndexJ}}\right)
                \\
                -
                \frac{\symZ_\symMaxUpdateVal^{\frac{1}{\symBase-1}}}{\symBase-1}
                \frac{\symZ_\symMaxUpdateVal \ln^3 \symZ_\symMaxUpdateVal}{1-\symZ_\symMaxUpdateVal}
                \\
                -
                \symZ_\symMaxUpdateVal^{\frac{1}{\symBase-1}}(1-\symZ_\symMaxUpdateVal)
                \frac{\ln\symZ_\symMaxUpdateVal}{\symBase-1}
                \left(\sum_{\symIndexJ = 1}^\symNumExtraBits
                \frac{\symZ_{\symMaxUpdateVal-\symIndexJ}\ln^2 \symZ_{\symMaxUpdateVal-\symIndexJ}}{1-\symZ_{\symMaxUpdateVal-\symIndexJ}}\right)
                \\
                +
                \symZ_{\symMaxUpdateVal+\symNumExtraBits}^{\frac{1}{\symBase-1}}
                \frac{\symZ^2_{\symMaxUpdateVal}\ln^3 \symZ_{\symMaxUpdateVal}}{(1-\symZ_{\symMaxUpdateVal})^2}
                -
                \symZ_{\symMaxUpdateVal}^{\frac{1}{\symBase-1}}
                \frac{\symZ^2_{\symMaxUpdateVal}\ln^3 \symZ_{\symMaxUpdateVal}}{(1-\symZ_{\symMaxUpdateVal})^2}
              \end{array}\right)
      \\
       & =
      \frac{\symNumReg}{\symCardinality^3}
      \sum_{\symMaxUpdateVal=-\infty}^\infty
      \left(\begin{array}{@{}l}
                \frac{\symZ_\symMaxUpdateVal \ln^3 \symZ_\symMaxUpdateVal}{1-\symZ_\symMaxUpdateVal}
                \left(
                \frac{\symZ_{\symMaxUpdateVal+\symNumExtraBits}^{\frac{1}{\symBase-1}}\symZ_{\symMaxUpdateVal}}{1-\symZ_\symMaxUpdateVal}
                -
                \frac{\symZ_\symMaxUpdateVal^{\frac{1}{\symBase-1}}}{\symBase-1}
                \right)
                \\
                +
                \symZ_\symMaxUpdateVal^{\frac{1}{\symBase-1}}\ln\symZ_\symMaxUpdateVal
                \left(
                \symZ_\symMaxUpdateVal
                -
                \frac{1-\symZ_\symMaxUpdateVal}{\symBase-1}
                \right)
                \left(\sum_{\symIndexJ = 1}^\symNumExtraBits
                \frac{\symZ_{\symMaxUpdateVal-\symIndexJ}\ln^2 \symZ_{\symMaxUpdateVal-\symIndexJ}}{1-\symZ_{\symMaxUpdateVal-\symIndexJ}}\right)
              \end{array}\right)
      \\
       & =
      \frac{\symNumReg}{\symCardinality^3}
      \sum_{\symMaxUpdateVal=-\infty}^\infty
      \left(\begin{array}{@{}l}
                \frac{\symZ_\symMaxUpdateVal \ln^3 \symZ_\symMaxUpdateVal}{1-\symZ_\symMaxUpdateVal}
                \left(
                \frac{\symZ_{\symMaxUpdateVal}^{\frac{\symBase^{-\symNumExtraBits}}{\symBase-1}}\symZ_{\symMaxUpdateVal}}{1-\symZ_\symMaxUpdateVal}
                -
                \frac{\symZ_\symMaxUpdateVal^{\frac{1}{\symBase-1}}}{\symBase-1}
                \right)
                \\
                +
                \sum_{\symIndexJ = 1}^\symNumExtraBits
                \symZ_{\symMaxUpdateVal+\symIndexJ}^{\frac{1}{\symBase-1}}\ln\symZ_{\symMaxUpdateVal+\symIndexJ}
                \left(
                \symZ_{\symMaxUpdateVal+\symIndexJ}
                -
                \frac{1-\symZ_{\symMaxUpdateVal+\symIndexJ}}{\symBase-1}
                \right)
                \frac{\symZ_{\symMaxUpdateVal}\ln^2 \symZ_{\symMaxUpdateVal}}{1-\symZ_{\symMaxUpdateVal}}
              \end{array}\right)
      \\
       & =
      \frac{\symNumReg}{\symCardinality^3}
      \sum_{\symMaxUpdateVal=-\infty}^\infty
      \left(\begin{array}{@{}l}
                \frac{\symZ_\symMaxUpdateVal \ln^3 \symZ_\symMaxUpdateVal}{1-\symZ_\symMaxUpdateVal}
                \left(
                \frac{\symZ_{\symMaxUpdateVal}^{\frac{\symBase^{-\symNumExtraBits}}{\symBase-1}}\symZ_{\symMaxUpdateVal}}{1-\symZ_\symMaxUpdateVal}
                -
                \frac{\symZ_\symMaxUpdateVal^{\frac{1}{\symBase-1}}}{\symBase-1}
                \right)
                \\
                +
                \sum_{\symIndexJ = 1}^\symNumExtraBits
                \symZ_{\symMaxUpdateVal}^{\frac{\symBase^{-\symIndexJ}}{\symBase-1}}\ln(\symZ_{\symMaxUpdateVal}^{\symBase^{-\symIndexJ}})
                \left(
                \symZ_{\symMaxUpdateVal}^{\symBase^{-\symIndexJ}}
                -
                \frac{1-\symZ_{\symMaxUpdateVal}^{\symBase^{-\symIndexJ}}}{\symBase-1}
                \right)
                \frac{\symZ_{\symMaxUpdateVal}\ln^2 \symZ_{\symMaxUpdateVal}}{1-\symZ_{\symMaxUpdateVal}}
              \end{array}\right)
      \\
       & =
      \frac{\symNumReg}{\symCardinality^3}
      \sum_{\symMaxUpdateVal=-\infty}^\infty
      \frac{\symZ_\symMaxUpdateVal \ln^3 \symZ_\symMaxUpdateVal}{1-\symZ_\symMaxUpdateVal}
      \cdot
      \\ & \quad \cdot
      \left(
      \frac{\symZ_{\symMaxUpdateVal}^{\frac{\symBase^{-\symNumExtraBits}}{\symBase-1}}\symZ_{\symMaxUpdateVal}}{1-\symZ_\symMaxUpdateVal}
      -
      \frac{\symZ_\symMaxUpdateVal^{\frac{1}{\symBase-1}}}{\symBase-1}
      +
      \sum_{\symIndexJ = 1}^\symNumExtraBits
      \symZ_{\symMaxUpdateVal}^{\frac{\symBase^{-\symIndexJ}}{\symBase-1}}\symBase^{-\symIndexJ}
      \left(
      \symZ_{\symMaxUpdateVal}^{\symBase^{-\symIndexJ}}
      -
      \frac{1-\symZ_{\symMaxUpdateVal}^{\symBase^{-\symIndexJ}}}{\symBase-1}
      \right)
      \right)
      \\
       & =
      \frac{\symNumReg}{\symCardinality^3}
      \sum_{\symMaxUpdateVal=-\infty}^\infty
      \frac{\symZ_\symMaxUpdateVal \ln^3 \symZ_\symMaxUpdateVal}{1-\symZ_\symMaxUpdateVal}
      \cdot
      \\ & \quad \cdot
      \left(
      \frac{\symZ_{\symMaxUpdateVal}^{\frac{\symBase^{-\symNumExtraBits}}{\symBase-1}}\symZ_{\symMaxUpdateVal}}{1-\symZ_\symMaxUpdateVal}
      -
      \frac{\symZ_\symMaxUpdateVal^{\frac{1}{\symBase-1}}}{\symBase-1}
      +
      \sum_{\symIndexJ = 1}^\symNumExtraBits
      \frac{\symZ_{\symMaxUpdateVal}^{\frac{\symBase^{-\symIndexJ+1}}{\symBase-1}}\symBase^{-\symIndexJ+1}}{\symBase-1}
      -
      \frac{\symZ_{\symMaxUpdateVal}^{\frac{\symBase^{-\symIndexJ}}{\symBase-1}}\symBase^{-\symIndexJ}}{\symBase-1}
      \right)
      \\
       & =
      \frac{\symNumReg}{\symCardinality^3}
      \sum_{\symMaxUpdateVal=-\infty}^\infty
      \frac{\symZ_\symMaxUpdateVal \ln^3 \symZ_\symMaxUpdateVal}{1-\symZ_\symMaxUpdateVal}
      \left(
      \frac{\symZ_{\symMaxUpdateVal}^{\frac{\symBase^{-\symNumExtraBits}}{\symBase-1}}\symZ_{\symMaxUpdateVal}}{1-\symZ_\symMaxUpdateVal}
      -
      \frac{\symZ_{\symMaxUpdateVal}^{\frac{\symBase^{-\symNumExtraBits}}{\symBase-1}}\symBase^{-\symNumExtraBits}}{\symBase-1}
      \right)
      \\
       & =
      \frac{\symNumReg}{\symCardinality^3}
      \sum_{\symMaxUpdateVal=-\infty}^\infty
      \frac{\symZ_{\symMaxUpdateVal}^{\frac{\symBase^{-\symNumExtraBits}}{\symBase-1}}\symZ_\symMaxUpdateVal \ln^3 \symZ_\symMaxUpdateVal}{1-\symZ_\symMaxUpdateVal}
      \left(
      \frac{\symZ_{\symMaxUpdateVal}}{1-\symZ_\symMaxUpdateVal}
      -
      \frac{\symBase^{-\symNumExtraBits}}{\symBase-1}
      \right)
      \\
       & =
      \frac{\symNumReg}{\symCardinality^3}
      \sum_{\symMaxUpdateVal=-\infty}^\infty
      \frac{\symZ_{\symMaxUpdateVal}^{1+\frac{\symBase^{-\symNumExtraBits}}{\symBase-1}}\ln^3 \symZ_\symMaxUpdateVal}{(1-\symZ_\symMaxUpdateVal)^2}
      \left(
      \symZ_{\symMaxUpdateVal}\left(1+\frac{\symBase^{-\symNumExtraBits}}{\symBase-1}\right)
      -
      \frac{\symBase^{-\symNumExtraBits}}{\symBase-1}
      \right).
    \end{align*}
    Therefore, we get
    \begin{align*}
       & \symBias(\symCardinalityEstimatorML)
      \approx
      \frac{1}{\symFisher^{2}}\symExpectation\!\left(
      \frac{1}{2}\frac{\partial^3 \ln\symLikelihood}{\partial \symCardinality^3}
      +\frac{\partial^2 \ln\symLikelihood}{\partial \symCardinality^2}\frac{\partial \ln\symLikelihood}{\partial \symCardinality}
      \right)
      \\
       & =
      \frac{1}{\symFisher^{2}}
      \frac{\symNumReg}{\symCardinality^3}
      \left(\begin{array}{@{}l}
                \sum_{\symMaxUpdateVal=-\infty}^\infty
                -\frac{1}{2}
                \frac{\symZ_{\symMaxUpdateVal}^{1 + \frac{\symBase^{-\symNumExtraBits}}{\symBase-1}}(1+\symZ_\symMaxUpdateVal)\ln^3\symZ_\symMaxUpdateVal}{(1-\symZ_\symMaxUpdateVal)^2}
                \\
                +
                \sum_{\symMaxUpdateVal=-\infty}^\infty
                \frac{\symZ_{\symMaxUpdateVal}^{1+\frac{\symBase^{-\symNumExtraBits}}{\symBase-1}}\ln^3 \symZ_\symMaxUpdateVal}{(1-\symZ_\symMaxUpdateVal)^2}
                \left(
                \symZ_{\symMaxUpdateVal}\left(1+\frac{\symBase^{-\symNumExtraBits}}{\symBase-1}\right)
                -
                \frac{\symBase^{-\symNumExtraBits}}{\symBase-1}
                \right)
              \end{array}\right)
      \\
       & =
      -\frac{1}{\symFisher^{2}}
      \frac{\symNumReg}{\symCardinality^3}
      \sum_{\symMaxUpdateVal=-\infty}^\infty
      \frac{\symZ_{\symMaxUpdateVal}^{1+\frac{\symBase^{-\symNumExtraBits}}{\symBase-1}}\ln^3 \symZ_\symMaxUpdateVal}{(1-\symZ_\symMaxUpdateVal)^2}
      \left(
      \frac{1+\symZ_\symMaxUpdateVal}{2}
      -
      \symZ_{\symMaxUpdateVal}\left(1+\frac{\symBase^{-\symNumExtraBits}}{\symBase-1}\right)
      +
      \frac{\symBase^{-\symNumExtraBits}}{\symBase-1}
      \right)
      \\
       & =
      -\frac{1}{\symFisher^{2}}\frac{\symNumReg}{\symCardinality^3}\left(
      \frac{1}{2} +\frac{\symBase^{-\symNumExtraBits}}{\symBase-1}
      \right)
      \sum_{\symMaxUpdateVal=-\infty}^\infty
      \frac{\symZ_{\symMaxUpdateVal}^{1 + \frac{\symBase^{-\symNumExtraBits}}{\symBase-1}}\ln^3\symZ_\symMaxUpdateVal}{1-\symZ_\symMaxUpdateVal}.
    \end{align*}
    \cref{lem:approximation} with $\symSomeFunc(\symX) = \frac{\symX^{1+ \frac{\symBase^{-\symNumExtraBits}}{\symBase-1}}\ln^3 \symX}{1-\symX}$ yields
    \begin{align*}
      \symBias(\symCardinalityEstimatorML)
       & \approx
      \frac{1}{\symFisher^{2}}\frac{\symNumReg}{\symCardinality^3}\left(
      \frac{1}{2} +\frac{\symBase^{-\symNumExtraBits}}{\symBase-1}
      \right)
      \frac{1}{\ln \symBase}
      \int_{0}^{\infty}
      \frac{e^{-\symY(1+ \frac{\symBase^{-\symNumExtraBits}}{\symBase-1})} \symY^2}{1-e^{-\symY}} d\symY
      \\
       & =
      \frac{1}{\symFisher^{2}}\frac{\symNumReg}{\symCardinality^3}\left(
      1 +\frac{2\symBase^{-\symNumExtraBits}}{\symBase-1}
      \right)
      \frac{1}{\ln \symBase}\symZetaFunc\!\left(3, 1 + \frac{\symBase^{-\symNumExtraBits}}{\symBase-1}\right).
    \end{align*}
    Replacing $\symFisher$ by using \cref{lem:fisher} finally gives
    \begin{equation*}
      \symBias(\symCardinalityEstimatorML)
      \approx
      \frac{\symCardinality}{\symNumReg} (\ln \symBase)
      \left(
      1 +\frac{2\symBase^{-\symNumExtraBits}}{\symBase-1}
      \right)
      \frac{
        \symZetaFunc\!\left(3, 1 + \frac{\symBase^{-\symNumExtraBits}}{\symBase-1}\right)}{\left(\symZetaFunc\!\left(2, 1 + \frac{\symBase^{-\symNumExtraBits}}{\symBase-1}\right)\right)^2}.
    \end{equation*}
  \end{proof}

  \begin{lemma}
    \label{lem:expectation_gra_reg_contrib}
    For a random variable $\symRegister$ distributed according to \eqref{equ:simple_register_pmf}, the expectation of
    \begin{equation*}
      \symRegContrib(\symRegister)
      =
      \symSomeConstant \symBase^{-\symGRA\symMaxUpdateVal}\left(\frac{1}{\symBase^\symGRA-1}
      +
      \sum_{\symIndexS=1}^{\symNumExtraBits}
        (1-\symIndexBit_{\symIndexS})
      \symBase^{\symIndexS\symGRA}
      \right)
    \end{equation*}
    with $\symRegister = \symMaxUpdateVal2^\symNumExtraBits + \langle\symIndexBit_1\ldots\symIndexBit_{\symNumExtraBits}\rangle_2$
    and $\symSomeConstant = \frac{(\symBase-1+\symBase^{-\symNumExtraBits})^{\symGRA}\ln \symBase}{\symGammaFunc(\symGRA)}$ can be approximated using \cref{lem:approximation} by
    \begin{equation*}
      \symExpectation(\symRegContrib(\symRegister))
      \approx
      \frac{\symNumReg^\symGRA}{\symCardinality^\symGRA}.
    \end{equation*}
  \end{lemma}
  \begin{proof}
    Using \eqref{equ:simple_register_pmf} the expectation can be expressed as
    \begin{align*}
       & \symExpectation(\symRegContrib(\symRegister)) =
      \sum_{\symRegister=-\infty}^\infty
      \symDensityRegisterSimple(\symRegister\vert\symCardinality)
      \symRegContrib(\symRegister)
      \\
       & =
      \sum_{\symMaxUpdateVal=-\infty}^\infty
      \symSomeConstant
      \symBase^{-\symGRA\symMaxUpdateVal}
      \sum_{\symIndexBit_1,\ldots,\symIndexBit_\symNumExtraBits\in\lbrace0,1\rbrace}
      \symDensityRegisterSimple(\symMaxUpdateVal2^\symNumExtraBits + \langle\symIndexBit_1\ldots\symIndexBit_{\symNumExtraBits}
      \rangle_2\vert\symCardinality)
      \cdot
      \\
       & \quad \cdot
      \left(\frac{1}{\symBase^\symGRA-1}
      +
      \sum_{\symIndexS=1}^{\symNumExtraBits}
        (1-\symIndexBit_{\symIndexS})
      \symBase^{\symIndexS\symGRA}
      \right).
    \end{align*}
    With the help of \cref{lem:identity_pmf} we get
    \begin{align*}
       & \symExpectation(\symRegContrib(\symRegister))
      \\
       & =
      \sum_{\symMaxUpdateVal=-\infty}^\infty
      \symSomeConstant
      \symBase^{-\symGRA\symMaxUpdateVal}
      \symZ_{\symMaxUpdateVal}^{\frac{1}{\symBase-1}}
      (1-\symZ_{\symMaxUpdateVal})
      \left(\frac{1}{\symBase^\symGRA-1}+
      \sum_{\symIndexS=1}^\symNumExtraBits \symZ_{\symMaxUpdateVal-\symIndexS} \symBase^{\symGRA\symIndexS}
      \right)
      \\
       & =
      \symSomeConstant\sum_{\symMaxUpdateVal=-\infty}^\infty
      \frac{\symBase^{-\symGRA\symMaxUpdateVal}
      \symZ_{\symMaxUpdateVal}^{\frac{1}{\symBase-1}}
      (1-\symZ_{\symMaxUpdateVal})
      }{\symBase^\symGRA-1}
      +
      \sum_{\symIndexS=1}^\symNumExtraBits
      \left(\symZ_{\symMaxUpdateVal-\symIndexS}^{1+\frac{\symBase^{-\symIndexS}}{\symBase-1}}-\symZ_{\symMaxUpdateVal-\symIndexS}^{1+\frac{\symBase^{1-\symIndexS}}{\symBase-1}}\right)
      \symBase^{\symGRA(\symIndexS-\symMaxUpdateVal)}
      \\
       & =
      \symSomeConstant\sum_{\symMaxUpdateVal=-\infty}^\infty
      \frac{\symBase^{-\symGRA\symMaxUpdateVal}
      \symZ_{\symMaxUpdateVal}^{\frac{1}{\symBase-1}}
      (1-\symZ_{\symMaxUpdateVal})
      }{\symBase^\symGRA-1}
      +
      \sum_{\symIndexS=1}^\symNumExtraBits
      \left(\symZ_{\symMaxUpdateVal}^{1+\frac{\symBase^{-\symIndexS}}{\symBase-1}}-\symZ_{\symMaxUpdateVal}^{1+\frac{\symBase^{1-\symIndexS}}{\symBase-1}}\right)
      \symBase^{-\symGRA\symMaxUpdateVal}
      \\
       & =
      \symSomeConstant\sum_{\symMaxUpdateVal=-\infty}^\infty
      \symBase^{-\symGRA\symMaxUpdateVal}
      \left(
      \frac{
      \symZ_{\symMaxUpdateVal}^{\frac{1}{\symBase-1}}
      (1-\symZ_{\symMaxUpdateVal})
      }{\symBase^\symGRA-1}
      +
      \symZ_{\symMaxUpdateVal}^{1+\frac{\symBase^{-\symNumExtraBits}}{\symBase-1}}-\symZ_{\symMaxUpdateVal}^{1+\frac{1}{\symBase-1}}
      \right)
      \\
       & =
      \frac{\symSomeConstant\symNumReg^\symGRA}{\symCardinality^\symGRA (\symBase-1)^\symGRA}
      \sum_{\symMaxUpdateVal=-\infty}^\infty
      (-\ln \symZ_{\symMaxUpdateVal})^\symGRA
      \left(
      \frac{
      \symZ_{\symMaxUpdateVal}^{\frac{1}{\symBase-1}}
      (1-\symZ_{\symMaxUpdateVal})
      }{\symBase^\symGRA-1}
      +
      \symZ_{\symMaxUpdateVal}^{1+\frac{\symBase^{-\symNumExtraBits}}{\symBase-1}}-\symZ_{\symMaxUpdateVal}^{1+\frac{1}{\symBase-1}}
      \right)
      \\
       & =
      \frac{\symSomeConstant\symNumReg^\symGRA}{\symCardinality^\symGRA (\symBase-1)^\symGRA}
      \left(
      \sum_{\symMaxUpdateVal=-\infty}^\infty
      (-\ln \symZ_{\symMaxUpdateVal})^\symGRA
      \symZ_{\symMaxUpdateVal}^{1+\frac{\symBase^{-\symNumExtraBits}}{\symBase-1}}
      \right)
      \\
       & \quad +
      \frac{\symSomeConstant\symNumReg^\symGRA}{\symCardinality^\symGRA (\symBase-1)^\symGRA}
      \left(
      \sum_{\symMaxUpdateVal=-\infty}^\infty
      \frac{
      (-\ln \symZ_{\symMaxUpdateVal})^\symGRA
      \symZ_{\symMaxUpdateVal}^{\frac{1}{\symBase-1}}
      }{\symBase^\symGRA-1}
      -
      \frac{
        (-\ln \symZ_{\symMaxUpdateVal}^\symBase)^\symGRA
        {\symZ_{\symMaxUpdateVal}}^{\frac{\symBase}{\symBase-1}}
      }{\symBase^\symGRA-1}
      \right)
      \\
       & =
      \frac{\symSomeConstant\symNumReg^\symGRA}{\symCardinality^\symGRA (\symBase-1)^\symGRA}
      \left(
      \sum_{\symMaxUpdateVal=-\infty}^\infty
      (-\ln \symZ_{\symMaxUpdateVal})^\symGRA
      \symZ_{\symMaxUpdateVal}^{1+\frac{\symBase^{-\symNumExtraBits}}{\symBase-1}}
      \right)
      \\
       & \quad+
      \frac{\symSomeConstant\symNumReg^\symGRA}{\symCardinality^\symGRA (\symBase-1)^\symGRA}
      \left(
      \sum_{\symMaxUpdateVal=-\infty}^\infty
      \frac{
      (-\ln \symZ_{\symMaxUpdateVal})^\symGRA
      \symZ_{\symMaxUpdateVal}^{\frac{1}{\symBase-1}}
      }{\symBase^\symGRA-1}
      -
      \frac{
        (-\ln \symZ_{\symMaxUpdateVal-1})^\symGRA
        {\symZ_{\symMaxUpdateVal-1}}^{\frac{1}{\symBase-1}}
      }{\symBase^\symGRA-1}
      \right)
      \\
       & =
      \frac{\symSomeConstant\symNumReg^\symGRA}{\symCardinality^\symGRA (\symBase-1)^\symGRA}
      \sum_{\symMaxUpdateVal=-\infty}^\infty
      (-\ln \symZ_{\symMaxUpdateVal})^\symGRA
      \symZ_{\symMaxUpdateVal}^{1+\frac{\symBase^{-\symNumExtraBits}}{\symBase-1}}.
    \end{align*}
    The approximation given in \cref{lem:approximation} and the identity $\int_0^\infty \symY^{\symGRA-1} e^{-\symY\symX}d\symY = \frac{\symGammaFunc(\symGRA)}{\symX^\symGRA}$ finally leads to
    \begin{multline*}
      \symExpectation(\symRegContrib(\symRegister)) \approx
      \frac{\symSomeConstant\symNumReg^\symGRA}{\symCardinality^\symGRA (\symBase-1)^\symGRA \ln\symBase}
      \int_{0}^\infty
      \symY^{\symGRA-1}
      e^{-\symY(1+\frac{\symBase^{-\symNumExtraBits}}{\symBase-1})}
      d\symY
      \\
      =
      \frac{\symSomeConstant\symNumReg^\symGRA\symGammaFunc(\symGRA)}{\symCardinality^\symGRA (\symBase-1+\symBase^{-\symNumExtraBits})^\symGRA \ln\symBase}
      =
      \frac{\symNumReg^\symGRA}{\symCardinality^\symGRA}
      .
    \end{multline*}

  \end{proof}

  \begin{lemma}
    \label{lem:moment_two_gra_reg_contrib}
    For a random variable $\symRegister$ distributed according to \eqref{equ:simple_register_pmf}, the expectation of $(\symRegContrib(\symRegister))^2$ with $\symRegContrib(\symRegister)$ as defined in \cref{lem:expectation_gra_reg_contrib} can be approximated using \cref{lem:approximation} by
    \begin{multline*}
      \symExpectation((\symRegContrib(\symRegister))^2)
      \\
      \approx
      \frac{\symNumReg^{2\symGRA} \symGammaFunc(2\symGRA)\ln\symBase}{\symCardinality^{2\symGRA}(\symGammaFunc(\symGRA))^2}
      \left(
      1 + \frac{2\symBase^{-\symGRA\symNumExtraBits}}{\symBase^{\symGRA}-1}
      +
      \sum_{\symIndexS = 1}^\symNumExtraBits
      \frac{2\symBase^{-\symGRA\symIndexS}}{\left(1+\frac{(\symBase-1)\symBase^{-\symIndexS}}{\symBase-1+\symBase^{-\symNumExtraBits}}\right)^{\!2\symGRA}}\right).
    \end{multline*}
  \end{lemma}
  \begin{proof}
    Since $\symIndexBit_{\symIndexS} \in \lbrace 0, 1\rbrace$, we have $(1-\symIndexBit_{\symIndexS}) = (1-\symIndexBit_{\symIndexS})^2$ and we can write
    \begin{align*}
      (\symRegContrib(\symRegister))^2
       & =
      \symSomeConstant^2
      \symBase^{-2\symGRA\symMaxUpdateVal}
      \left(\frac{1}{\symBase^\symGRA-1}
      +
      \sum_{\symIndexS=1}^{\symNumExtraBits}
      (1-\symIndexBit_{\symIndexS})
      \symBase^{\symIndexS\symGRA}
      \right)^{\!2}
      \\
       & =
      \symSomeConstant^2
      \symBase^{-2\symGRA\symMaxUpdateVal}
      \left(\begin{array}{@{}l}
                \frac{1}{(\symBase^\symGRA-1)^2}
                +
                \sum_{\symIndexS=1}^{\symNumExtraBits}
                (1-\symIndexBit_{\symIndexS})
                \left(
                \frac{2\symBase^{\symGRA\symIndexS}}{\symBase^\symGRA-1}
                +
                \symBase^{2\symGRA\symIndexS}\right)
                \\
                +
                2\sum_{\symIndexS=1}^{\symNumExtraBits-1}
                (1-\symIndexBit_{\symIndexS})
                \sum_{\symIndexS'=1}^{\symNumExtraBits-\symIndexS}
                \symBase^{\symGRA(2\symIndexS+\symIndexS')}
                (1-\symIndexBit_{\symIndexS+\symIndexS'})
              \end{array}\right).
    \end{align*}
    Using \eqref{equ:simple_register_pmf} the expectation can be expressed as
    \begin{align*}
       & \symExpectation((\symRegContrib(\symRegister))^2) =
      \sum_{\symRegister=-\infty}^\infty
      \symDensityRegisterSimple(\symRegister\vert\symCardinality)
      (\symRegContrib(\symRegister))^2
      \\
       & =
      \sum_{\symMaxUpdateVal=-\infty}^\infty\sum_{\symIndexBit_1,\ldots,\symIndexBit_\symNumExtraBits\in\lbrace0,1\rbrace}
      \symSomeConstant^2
      \symBase^{-2\symGRA\symMaxUpdateVal}
      \symDensityRegisterSimple(\symMaxUpdateVal2^\symNumExtraBits + \langle\symIndexBit_1\ldots\symIndexBit_{\symNumExtraBits}\rangle_2\vert\symCardinality)
      \cdot
      \\
       & \quad\cdot
      \left(\begin{array}{@{}l}
                \frac{1}{(\symBase^\symGRA-1)^2}
                +
                \sum_{\symIndexS=1}^{\symNumExtraBits}
                (1-\symIndexBit_{\symIndexS})
                \left(
                \frac{2\symBase^{\symGRA\symIndexS}}{\symBase^\symGRA-1}
                +
                \symBase^{2\symGRA\symIndexS}\right)
                \\
                +
                2\sum_{\symIndexS=1}^{\symNumExtraBits-1}
                (1-\symIndexBit_{\symIndexS})
                \sum_{\symIndexS'=1}^{\symNumExtraBits-\symIndexS}
                \symBase^{\symGRA(2\symIndexS+\symIndexS')}
                (1-\symIndexBit_{\symIndexS+\symIndexS'})
              \end{array}\right).
    \end{align*}
    \cref{lem:identity_pmf} allows to write that as
    \begin{align*}
       & \symExpectation((\symRegContrib(\symRegister))^2) =
      \sum_{\symMaxUpdateVal=-\infty}^\infty
      \symSomeConstant^2
      \symBase^{-2\symGRA\symMaxUpdateVal}
      \symZ_{\symMaxUpdateVal}^{\frac{1}{\symBase-1}}
      (1-\symZ_{\symMaxUpdateVal})
      \cdot
      \\
       & \quad\cdot
      \left(\begin{array}{@{}l}
                \frac{1}{(\symBase^\symGRA-1)^2}
                +
                \sum_{\symIndexS=1}^{\symNumExtraBits}
                \symZ_{\symMaxUpdateVal-\symIndexS}
                \left(
                \frac{2\symBase^{\symGRA\symIndexS}}{\symBase^\symGRA-1}
                +
                \symBase^{2\symGRA\symIndexS}\right)
                \\
                +
                2\sum_{\symIndexS=1}^{\symNumExtraBits-1}
                \symZ_{\symMaxUpdateVal-\symIndexS}
                \sum_{\symIndexS'=1}^{\symNumExtraBits-\symIndexS}
                \symBase^{\symGRA(2\symIndexS+\symIndexS')}
                \symZ_{\symMaxUpdateVal-\symIndexS-\symIndexS'}
              \end{array}\right)
      \\
       & =
      \sum_{\symMaxUpdateVal=-\infty}^\infty
      \symSomeConstant^2
      \symBase^{-2\symGRA\symMaxUpdateVal}
      (\symZ_{\symMaxUpdateVal}^{\frac{1}{\symBase-1}}-\symZ_{\symMaxUpdateVal-1}^{\frac{1}{\symBase-1}})
      \cdot
      \\
       & \quad\cdot
      \left(\begin{array}{@{}l}
                \frac{1}{(\symBase^\symGRA-1)^2}
                +
                \sum_{\symIndexS=1}^{\symNumExtraBits}
                \symZ_{\symMaxUpdateVal-\symIndexS}
                \left(
                \frac{2\symBase^{\symGRA\symIndexS}}{\symBase^\symGRA-1}
                +
                \symBase^{2\symGRA\symIndexS}\right)
                \\
                +
                2\sum_{\symIndexS=1}^{\symNumExtraBits-1}
                \sum_{\symIndexS'=1}^{\symNumExtraBits-\symIndexS}
                \symBase^{\symGRA(2\symIndexS+\symIndexS')}
                \symZ_{\symMaxUpdateVal-\symIndexS}
                \symZ_{\symMaxUpdateVal-\symIndexS-\symIndexS'}
              \end{array}\right)
      \\
       & =
      \sum_{\symMaxUpdateVal=-\infty}^\infty
      \symSomeConstant^2
      \symBase^{-2\symGRA\symMaxUpdateVal}
      \symZ_{\symMaxUpdateVal}^{\frac{1}{\symBase-1}}
      \cdot
      \\
       & \quad\cdot
      \left(\begin{array}{@{}l}
                \frac{1-\symBase^{-2\symGRA}}{(\symBase^\symGRA-1)^2}
                +
                \sum_{\symIndexS=1}^{\symNumExtraBits}
                (\symZ_{\symMaxUpdateVal-\symIndexS}-\symZ_{\symMaxUpdateVal+1-\symIndexS}\symBase^{-2\symGRA})
                \left(
                \frac{2\symBase^{\symGRA\symIndexS}}{\symBase^\symGRA-1}
                +
                \symBase^{2\symGRA\symIndexS}\right)
                \\
                +
                2\sum_{\symIndexS=1}^{\symNumExtraBits-1}
                \sum_{\symIndexS'=1}^{\symNumExtraBits-\symIndexS}
                \left(\begin{array}{@{}l}
                    \symBase^{\symGRA(2\symIndexS+\symIndexS')}
                    \symZ_{\symMaxUpdateVal-\symIndexS}
                    \symZ_{\symMaxUpdateVal-\symIndexS-\symIndexS'}
                    \\
                    -
                    \symBase^{\symGRA(2\symIndexS+\symIndexS'-2)}
                    \symZ_{\symMaxUpdateVal+1-\symIndexS}
                    \symZ_{\symMaxUpdateVal+1-\symIndexS-\symIndexS'}
                  \end{array}\right)
              \end{array}\right)
      \\
       & =
      \sum_{\symMaxUpdateVal=-\infty}^\infty
      \symSomeConstant^2
      \symBase^{-2\symGRA\symMaxUpdateVal}
      \symZ_{\symMaxUpdateVal}^{\frac{1}{\symBase-1}}
      \cdot
      \\
       & \quad\cdot
      \left(\begin{array}{@{}l}
                \symBase^{-2\symGRA}\frac{\symBase^{\symGRA}+1}{\symBase^\symGRA-1}
                +
                \sum_{\symIndexS=1}^{\symNumExtraBits}
                (\symZ_{\symMaxUpdateVal-\symIndexS}-\symZ_{\symMaxUpdateVal+1-\symIndexS}\symBase^{-2\symGRA})
                \frac{2\symBase^{\symGRA\symIndexS}}{\symBase^\symGRA-1}
                \\
                +\sum_{\symIndexS=1}^{\symNumExtraBits}
                (\symZ_{\symMaxUpdateVal-\symIndexS}\symBase^{2\symGRA\symIndexS}-\symZ_{\symMaxUpdateVal-(\symIndexS-1)}\symBase^{2\symGRA(\symIndexS-1)})
                \\
                +
                2\sum_{\symIndexS'=1}^{\symNumExtraBits-1}
                \sum_{\symIndexS=1}^{\symNumExtraBits-\symIndexS'}
                \left(\begin{array}{@{}l}
                    \symBase^{\symGRA(2\symIndexS+\symIndexS')}
                    \symZ_{\symMaxUpdateVal-\symIndexS}
                    \symZ_{\symMaxUpdateVal-\symIndexS-\symIndexS'}
                    \\
                    -
                    \symBase^{\symGRA(2(\symIndexS-1)+\symIndexS')}
                    \symZ_{\symMaxUpdateVal-(\symIndexS-1)}
                    \symZ_{\symMaxUpdateVal-(\symIndexS-1)-\symIndexS'}
                  \end{array}\right)
              \end{array}\right)
      \\
       & =
      \sum_{\symMaxUpdateVal=-\infty}^\infty
      \symSomeConstant^2
      \symBase^{-2\symGRA\symMaxUpdateVal}
      \symZ_{\symMaxUpdateVal}^{\frac{1}{\symBase-1}}
      \left(\begin{array}{@{}l}
                \symBase^{-2\symGRA}\frac{\symBase^{\symGRA}+1}{\symBase^\symGRA-1}
                \\
                +\sum_{\symIndexS=1}^{\symNumExtraBits}
                (\symZ_{\symMaxUpdateVal-\symIndexS}-\symZ_{\symMaxUpdateVal+1-\symIndexS}\symBase^{-2\symGRA})
                \frac{2\symBase^{\symGRA\symIndexS}}{\symBase^\symGRA-1}
                \\
                +
                \symZ_{\symMaxUpdateVal-\symNumExtraBits}\symBase^{2\symGRA \symNumExtraBits}-\symZ_{\symMaxUpdateVal}
                \\
                +
                2\sum_{\symIndexS'=1}^{\symNumExtraBits-1}
                \left(\begin{array}{@{}l}
                    \symBase^{\symGRA(2\symNumExtraBits-\symIndexS')}
                    \symZ_{\symMaxUpdateVal-\symNumExtraBits+\symIndexS'}
                    \symZ_{\symMaxUpdateVal-\symNumExtraBits}
                    \\
                    -
                    \symBase^{\symGRA\symIndexS'}
                    \symZ_{\symMaxUpdateVal}
                    \symZ_{\symMaxUpdateVal-\symIndexS'}
                  \end{array}\right)
              \end{array}\right)
      \\
       & =
      \sum_{\symMaxUpdateVal=-\infty}^\infty
      \symSomeConstant^2
      \symBase^{-2\symGRA\symMaxUpdateVal}
      \left(\begin{array}{@{}l}
                \symZ_{\symMaxUpdateVal}^{\frac{1}{\symBase-1}}
                \symBase^{-2\symGRA}\frac{\symBase^{\symGRA}+1}{\symBase^\symGRA-1}
                +
                \symZ_{\symMaxUpdateVal}^{\frac{1}{\symBase-1}}
                \symZ_{\symMaxUpdateVal-\symNumExtraBits}\symBase^{2\symGRA\symNumExtraBits}- \symZ_{\symMaxUpdateVal}^{\frac{1}{\symBase-1}}
                \symZ_{\symMaxUpdateVal}
                \\
                + \symZ_{\symMaxUpdateVal}^{\frac{1}{\symBase-1}}
                \sum_{\symIndexS=1}^{\symNumExtraBits}
                (\symZ_{\symMaxUpdateVal-\symIndexS}-\symZ_{\symMaxUpdateVal+1-\symIndexS}\symBase^{-2\symGRA})
                \frac{2\symBase^{\symGRA\symIndexS}}{\symBase^\symGRA-1}
                \\
                +
                2 \symZ_{\symMaxUpdateVal}^{\frac{1}{\symBase-1}}
                \sum_{\symIndexS=1}^{\symNumExtraBits}
                \left(\begin{array}{@{}l}
                    \symBase^{\symGRA(2\symNumExtraBits-\symIndexS)}
                    \symZ_{\symMaxUpdateVal-\symNumExtraBits+\symIndexS}
                    \symZ_{\symMaxUpdateVal-\symNumExtraBits}
                    \\
                    -
                    \symBase^{\symGRA\symIndexS}
                    \symZ_{\symMaxUpdateVal}
                    \symZ_{\symMaxUpdateVal-\symIndexS}
                  \end{array}\right)
              \end{array}\right)
      \\
       & =
      \sum_{\symMaxUpdateVal=-\infty}^\infty
      \symSomeConstant^2
      \symBase^{-2\symGRA\symMaxUpdateVal}
      \cdot
      \\
       & \quad\cdot
      \left(\begin{array}{@{}l}
                \symZ_{\symMaxUpdateVal}^{\frac{\symBase}{\symBase-1}}
                \frac{\symBase^{\symGRA}+1}{\symBase^\symGRA-1}
                +
                \symZ_{\symMaxUpdateVal}^{1+\frac{\symBase^{-\symNumExtraBits}}{\symBase-1}}
                - \symZ_{\symMaxUpdateVal}^{\frac{\symBase}{\symBase-1}}
                \\
                +
                2\sum_{\symIndexS=1}^{\symNumExtraBits}
                (\symZ_{\symMaxUpdateVal}^{\frac{1}{\symBase-1}}\symZ_{\symMaxUpdateVal-\symIndexS}-\symZ_{\symMaxUpdateVal}^{\frac{1}{\symBase-1}}\symZ_{\symMaxUpdateVal+1-\symIndexS}\symBase^{-2\symGRA})
                \frac{\symBase^{\symGRA\symIndexS}}{\symBase^\symGRA-1}
                \\
                +
                2\sum_{\symIndexS=1}^{\symNumExtraBits}
                \left(
                \symZ_{\symMaxUpdateVal}^{\frac{1}{\symBase-1}}
                \symBase^{\symGRA(2\symNumExtraBits-\symIndexS)}
                \symZ_{\symMaxUpdateVal-\symNumExtraBits+\symIndexS}
                \symZ_{\symMaxUpdateVal-\symNumExtraBits}
                -
                \symZ_{\symMaxUpdateVal}^{\frac{\symBase}{\symBase-1}}
                \symBase^{\symGRA\symIndexS}
                \symZ_{\symMaxUpdateVal-\symIndexS}
                \right)
              \end{array}\right)
      \\
       & =
      \sum_{\symMaxUpdateVal=-\infty}^\infty
      \symSomeConstant^2
      \symBase^{-2\symGRA\symMaxUpdateVal}
      \cdot
      \\
       & \quad\cdot
      \left(\begin{array}{@{}l}
                \symZ_{\symMaxUpdateVal}^{\frac{\symBase}{\symBase-1}}
                \frac{2}{\symBase^\symGRA-1}
                +
                \symZ_{\symMaxUpdateVal}^{1+\frac{\symBase^{-\symNumExtraBits}}{\symBase-1}}
                \\
                +
                2\sum_{\symIndexS=1}^{\symNumExtraBits}
                \left(\begin{array}{@{}l}
                    (\symZ_{\symMaxUpdateVal}^{\frac{1}{\symBase-1}}\symZ_{\symMaxUpdateVal-\symIndexS}-\symZ_{\symMaxUpdateVal}^{\frac{1}{\symBase-1}}\symZ_{\symMaxUpdateVal+1-\symIndexS}\symBase^{-2\symGRA})
                    \frac{\symBase^{\symGRA\symIndexS}}{\symBase^\symGRA-1}
                    \\
                    -
                    \symZ_{\symMaxUpdateVal}^{\frac{\symBase}{\symBase-1}}
                    \symBase^{\symGRA\symIndexS}
                    \symZ_{\symMaxUpdateVal-\symIndexS}
                  \end{array}\right)
                \\
                +
                2\sum_{\symIndexS=1}^{\symNumExtraBits}
                \symZ_{\symMaxUpdateVal+\symNumExtraBits}^{\frac{1}{\symBase-1}}
                \symBase^{-\symGRA\symIndexS}
                \symZ_{\symMaxUpdateVal+\symIndexS}
                \symZ_{\symMaxUpdateVal}
              \end{array}\right)
      \\
       & =
      \sum_{\symMaxUpdateVal=-\infty}^\infty
      \symSomeConstant^2
      \symBase^{-2\symGRA\symMaxUpdateVal}
      \cdot
      \\
       & \quad\cdot
      \left(\begin{array}{@{}l}
                \symZ_{\symMaxUpdateVal}^{\frac{\symBase}{\symBase-1}}
                \frac{2}{\symBase^\symGRA-1}
                +
                \symZ_{\symMaxUpdateVal}^{1+\frac{\symBase^{-\symNumExtraBits}}{\symBase-1}}
                \\
                +
                \frac{2}{\symBase^\symGRA-1}\sum_{\symIndexS=1}^{\symNumExtraBits}
                \left(\begin{array}{@{}l}
                    \symZ_{\symMaxUpdateVal}^{\frac{1}{\symBase-1}}\symZ_{\symMaxUpdateVal-\symIndexS}\symBase^{\symGRA\symIndexS}
                    -\symZ_{\symMaxUpdateVal}^{\frac{1}{\symBase-1}}\symZ_{\symMaxUpdateVal+1-\symIndexS}\symBase^{\symGRA(\symIndexS-2)}
                    \\
                    +
                    \symZ_{\symMaxUpdateVal}^{\frac{\symBase}{\symBase-1}}
                    \symBase^{\symGRA\symIndexS}
                    \symZ_{\symMaxUpdateVal-\symIndexS}
                    -
                    \symZ_{\symMaxUpdateVal}^{\frac{\symBase}{\symBase-1}}
                    \symBase^{\symGRA(\symIndexS+1)}
                    \symZ_{\symMaxUpdateVal-\symIndexS}
                  \end{array}\right)
                \\
                +
                2\sum_{\symIndexS=1}^{\symNumExtraBits}
                \symZ_{\symMaxUpdateVal+\symNumExtraBits}^{\frac{1}{\symBase-1}}
                \symBase^{-\symGRA\symIndexS}
                \symZ_{\symMaxUpdateVal+\symIndexS}
                \symZ_{\symMaxUpdateVal}
              \end{array}\right)
      \\
       & =
      \sum_{\symMaxUpdateVal=-\infty}^\infty
      \symSomeConstant^2
      \symBase^{-2\symGRA\symMaxUpdateVal}
      \cdot
      \\
       & \quad\cdot
      \left(\begin{array}{@{}l}
                \symZ_{\symMaxUpdateVal}^{\frac{\symBase}{\symBase-1}}
                \frac{2}{\symBase^\symGRA-1}
                +
                \symZ_{\symMaxUpdateVal}^{1+\frac{\symBase^{-\symNumExtraBits}}{\symBase-1}}
                \\
                +
                \frac{2}{\symBase^\symGRA-1}\sum_{\symIndexS=1}^{\symNumExtraBits}
                \left(\begin{array}{@{}l}
                    \symZ_{\symMaxUpdateVal}^{\frac{1}{\symBase-1}}\symZ_{\symMaxUpdateVal-\symIndexS}\symBase^{\symGRA\symIndexS}
                    -
                    \symZ_{\symMaxUpdateVal}^{\frac{\symBase}{\symBase-1}}\symZ_{\symMaxUpdateVal-\symIndexS}\symBase^{\symGRA\symIndexS}
                    \\
                    +
                    \symZ_{\symMaxUpdateVal}^{\frac{\symBase}{\symBase-1}}
                    \symBase^{\symGRA\symIndexS}
                    \symZ_{\symMaxUpdateVal-\symIndexS}
                    -
                    \symZ_{\symMaxUpdateVal}^{\frac{\symBase}{\symBase-1}}
                    \symBase^{\symGRA(\symIndexS+1)}
                    \symZ_{\symMaxUpdateVal-\symIndexS}
                  \end{array}\right)
                \\
                +
                2\sum_{\symIndexS=1}^{\symNumExtraBits}
                \symZ_{\symMaxUpdateVal+\symNumExtraBits}^{\frac{1}{\symBase-1}}
                \symBase^{-\symGRA\symIndexS}
                \symZ_{\symMaxUpdateVal+\symIndexS}
                \symZ_{\symMaxUpdateVal}
              \end{array}\right)
      \\
       & =
      \sum_{\symMaxUpdateVal=-\infty}^\infty
      \symSomeConstant^2
      \symBase^{-2\symGRA\symMaxUpdateVal}
      \cdot
      \\
       & \quad\cdot
      \left(\begin{array}{@{}l}
                \symZ_{\symMaxUpdateVal}^{\frac{\symBase}{\symBase-1}}
                \frac{2}{\symBase^\symGRA-1}
                +
                \symZ_{\symMaxUpdateVal}^{1+\frac{\symBase^{-\symNumExtraBits}}{\symBase-1}}
                \\
                +
                \frac{2}{\symBase^\symGRA-1}\sum_{\symIndexS=1}^{\symNumExtraBits}
                \left(
                \symZ_{\symMaxUpdateVal}^{\frac{1}{\symBase-1}}\symZ_{\symMaxUpdateVal-\symIndexS}\symBase^{\symGRA\symIndexS}
                -
                \symZ_{\symMaxUpdateVal}^{\frac{1}{\symBase-1}}
                \symBase^{\symGRA(\symIndexS-1)}
                \symZ_{\symMaxUpdateVal-(\symIndexS-1)}
                \right)
                \\
                +
                2\sum_{\symIndexS=1}^{\symNumExtraBits}
                \symZ_{\symMaxUpdateVal+\symNumExtraBits}^{\frac{1}{\symBase-1}}
                \symBase^{-\symGRA\symIndexS}
                \symZ_{\symMaxUpdateVal+\symIndexS}
                \symZ_{\symMaxUpdateVal}
              \end{array}\right)
      \\
       & =
      \sum_{\symMaxUpdateVal=-\infty}^\infty
      \symSomeConstant^2
      \symBase^{-2\symGRA\symMaxUpdateVal}
      \left(\begin{array}{@{}l}
                \symZ_{\symMaxUpdateVal}^{\frac{\symBase}{\symBase-1}}
                \frac{2}{\symBase^\symGRA-1}
                +
                \symZ_{\symMaxUpdateVal}^{1+\frac{\symBase^{-\symNumExtraBits}}{\symBase-1}}
                \\
                +
                \frac{2}{\symBase^\symGRA-1}
                \left(
                \symZ_{\symMaxUpdateVal}^{\frac{1}{\symBase-1}}\symZ_{\symMaxUpdateVal-\symNumExtraBits}\symBase^{\symGRA\symNumExtraBits}
                -
                \symZ_{\symMaxUpdateVal}^{\frac{1}{\symBase-1}}
                \symZ_{\symMaxUpdateVal}
                \right)
                \\
                +
                2\sum_{\symIndexS=1}^{\symNumExtraBits}
                \symZ_{\symMaxUpdateVal+\symNumExtraBits}^{\frac{1}{\symBase-1}}
                \symBase^{-\symGRA\symIndexS}
                \symZ_{\symMaxUpdateVal+\symIndexS}
                \symZ_{\symMaxUpdateVal}
              \end{array}\right)
      \\
       & =
      \sum_{\symMaxUpdateVal=-\infty}^\infty
      \symSomeConstant^2
      \symBase^{-2\symGRA\symMaxUpdateVal}
      \left(\begin{array}{@{}l}
                \symZ_{\symMaxUpdateVal}^{1+\frac{\symBase^{-\symNumExtraBits}}{\symBase-1}}
                +
                \frac{2}{\symBase^\symGRA-1}
                \symZ_{\symMaxUpdateVal+\symNumExtraBits}^{\frac{1}{\symBase-1}}\symZ_{\symMaxUpdateVal}\symBase^{-\symGRA\symNumExtraBits}
                \\
                +
                2\sum_{\symIndexS=1}^{\symNumExtraBits}
                \symZ_{\symMaxUpdateVal+\symNumExtraBits}^{\frac{1}{\symBase-1}}
                \symBase^{-\symGRA\symIndexS}
                \symZ_{\symMaxUpdateVal+\symIndexS}
                \symZ_{\symMaxUpdateVal}
              \end{array}\right)
      \\
       & =
      \symSomeConstant^2
      \left(1 + \frac{2\symBase^{-\symGRA\symNumExtraBits}}{\symBase^{\symGRA}-1}\right)
      \sum_{\symMaxUpdateVal=-\infty}^\infty
      \symZ_\symMaxUpdateVal^{
        1+\frac{\symBase^{-\symNumExtraBits}}{\symBase-1}
      }
      \symBase^{-2\symGRA\symMaxUpdateVal}
      \\
       & \quad +
      2\symSomeConstant^2 \sum_{\symIndexS = 1}^\symNumExtraBits
      \symBase^{-\symGRA\symIndexS}
      \sum_{\symMaxUpdateVal=-\infty}^\infty
      \symZ_\symMaxUpdateVal^{1+\frac{\symBase^{-\symNumExtraBits}}{\symBase-1}+\symBase^{-\symIndexS}}
      \symBase^{-2\symGRA\symMaxUpdateVal}
      \\
       & =
      \frac{\symNumReg^{2\symGRA}\symSomeConstant^2\left(1 + \frac{2\symBase^{-\symGRA\symNumExtraBits}}{\symBase^{\symGRA}-1}\right)}{\symCardinality^{2\symGRA}(\symBase-1)^{2\symGRA}}
      \sum_{\symMaxUpdateVal=-\infty}^\infty
      \symZ_\symMaxUpdateVal^{
        1+\frac{\symBase^{-\symNumExtraBits}}{\symBase-1}
      }
      (-\ln \symZ_\symMaxUpdateVal)^{2\symGRA}
      \\
       & \quad +
      \frac{2\symNumReg^{2\symGRA}\symSomeConstant^2}{\symCardinality^{2\symGRA}(\symBase-1)^{2\symGRA}}
      \sum_{\symIndexS = 1}^\symNumExtraBits
      \symBase^{-\symGRA\symIndexS}
      \sum_{\symMaxUpdateVal=-\infty}^\infty
      \symZ_\symMaxUpdateVal^{1+\frac{\symBase^{-\symNumExtraBits}}{\symBase-1}+\symBase^{-\symIndexS}}
      (-\ln \symZ_\symMaxUpdateVal)^{2\symGRA}.
    \end{align*}
    The approximation given in \cref{lem:approximation} and the identity $\int_0^\infty e^{-\symX\symY}\symY^{2\symGRA-1}d\symY = \frac{\symGammaFunc(2\symGRA)}{\symX^{2\symGRA}}$ finally leads to
    \begin{align*}
       & \symExpectation((\symRegContrib(\symRegister))^2) \approx
      \frac{\symNumReg^{2\symGRA}\symSomeConstant^2\left(1 + \frac{2\symBase^{-\symGRA\symNumExtraBits}}{\symBase^{\symGRA}-1}\right)}{\symCardinality^{2\symGRA}(\symBase-1)^{2\symGRA}\ln\symBase}
      \int_0^\infty
      e^{-\symY(
          1+\frac{\symBase^{-\symNumExtraBits}}{\symBase-1}
          )}
      \symY^{2\symGRA-1}
      d\symY
      \\
       & \quad +
      \frac{2\symNumReg^{2\symGRA}\symSomeConstant^2}{\symCardinality^{2\symGRA}(\symBase-1)^{2\symGRA}\ln\symBase}
      \sum_{\symIndexS = 1}^\symNumExtraBits
      \symBase^{-\symGRA\symIndexS}
      \int_0^\infty
      e^{-\symY(1+\frac{\symBase^{-\symNumExtraBits}}{\symBase-1}+\symBase^{-\symIndexS})}
      \symY^{2\symGRA-1}
      d\symY
      \\
       & =
      \frac{\symNumReg^{2\symGRA}\symSomeConstant^2 \symGammaFunc(2\symGRA)}{\symCardinality^{2\symGRA}(\symBase-1)^{2\symGRA}\ln\symBase}
      \left(
      \frac{1 + \frac{2\symBase^{-\symGRA\symNumExtraBits}}{\symBase^{\symGRA}-1}}{ (1+\frac{\symBase^{-\symNumExtraBits}}{\symBase-1})^{2\symGRA}}
      +
      \sum_{\symIndexS = 1}^\symNumExtraBits
      \frac{2\symBase^{-\symGRA\symIndexS}}{(1+\frac{\symBase^{-\symNumExtraBits}}{\symBase-1}+\symBase^{-\symIndexS})^{2\symGRA}}\right)
      \\
       & = \frac{\symNumReg^{2\symGRA} \symGammaFunc(2\symGRA)\ln\symBase}{\symCardinality^{2\symGRA}(\symGammaFunc(\symGRA))^2}
      \left(
      1 + \frac{2\symBase^{-\symGRA\symNumExtraBits}}{\symBase^{\symGRA}-1}
      +
      \sum_{\symIndexS = 1}^\symNumExtraBits
      \frac{2\symBase^{-\symGRA\symIndexS}}{\left(1+\frac{(\symBase-1)\symBase^{-\symIndexS}}{\symBase-1+\symBase^{-\symNumExtraBits}}\right)^{\!2\symGRA}}\right).
    \end{align*}
  \end{proof}

  \begin{lemma}
    \label{lem:expectation_gamma_new}
    For a random variable $\symRegister$ that is distributed according to \eqref{equ:simple_register_pmf} with $\symNumExtraBits=2$, the expectation of $\symRegContrib(\symRegister) = \symBase^{-\symGRA\lfloor\symRegister/4\rfloor} \symContributionCoefficient_{\symRegister\bmod 4}$ with $\symContributionCoefficient_0$, $\symContributionCoefficient_1$, $\symContributionCoefficient_2$, $\symContributionCoefficient_3$ normalized by the condition
    \begin{equation}
      \label{equ:norm_condition}
      \symContributionCoefficient_0 \symNewConstant_0(\symGRA)
      +
      \symContributionCoefficient_1 \symNewConstant_1(\symGRA)
      +
      \symContributionCoefficient_2 \symNewConstant_2(\symGRA)
      +\symContributionCoefficient_3 \symNewConstant_3(\symGRA)
      =\frac{\ln\symBase}{\symGammaFunc(\symGRA)}
    \end{equation}
    can be approximated with \cref{lem:approximation} by
    \begin{equation*}
      \symExpectation(\symRegContrib(\symRegister))
      \approx
      \frac{\symNumReg^\symGRA}{\symCardinality^\symGRA}.
    \end{equation*}
    $\symNewConstant_0(\symGRA)$, $\symNewConstant_1(\symGRA)$, $\symNewConstant_2(\symGRA)$, and $\symNewConstant_3(\symGRA)$ are positive functions for $\symBase > 1, \symGRA>0$ and defined by
    \begin{align*}
      \symNewConstant_0(\symGRA)
       & := \textstyle \frac{1}{(\symBase^3-\symBase+1)^{\symGRA}}-\frac{1}{\symBase^{3\symGRA}},
      \\
      \symNewConstant_1(\symGRA)
       & := \textstyle \frac{1}{(\symBase^2-\symBase+1)^{\symGRA}}-\frac{1}{\symBase^{2\symGRA}}-\frac{1}{(\symBase^3-\symBase+1)^{\symGRA}}+\frac{1}{\symBase^{3\symGRA}}, \\
      \symNewConstant_2(\symGRA)
       & := \textstyle \frac{1}{(\symBase^3-\symBase^2+1)^{\symGRA}}-\frac{1}{(\symBase^3-\symBase^2+\symBase)^{\symGRA}}
      -\frac{1}{(\symBase^3-\symBase+1)^{\symGRA}}+\frac{1}{\symBase^{3\symGRA}},
      \\
      \symNewConstant_3(\symGRA)
       & :=
      \textstyle \frac{1}{(\symBase^3-\symBase+1)^{\symGRA}}
      -\frac{1}{(\symBase^3-\symBase^2+1)^{\symGRA}}+\frac{1}{(\symBase^3-\symBase^2+\symBase)^{\symGRA}} -\frac{1}{\symBase^{3\symGRA}}
      \\
       & \quad \textstyle- \frac{1}{(\symBase^2-\symBase+1)^{\symGRA}}+\frac{1}{\symBase^{2\symGRA}}
      +1-\frac{1}{\symBase^{\symGRA}}.
    \end{align*}
  \end{lemma}
  \begin{proof}
    For $\symNumExtraBits=2$ \eqref{equ:simple_register_pmf} can be written as
    \begin{equation*}
      \symDensityRegisterSimple(\symRegister\vert\symCardinality) =
      \symZ_{\symMaxUpdateVal}^{\frac{1}{\symBase-1}}
      (1-\symZ_{\symMaxUpdateVal})
      \symZ_{\symMaxUpdateVal-1}^{1-\symIndexBit_1}
      (1- \symZ_{\symMaxUpdateVal-1})^{ \symIndexBit_1}
      \symZ_{\symMaxUpdateVal-2}^{1-\symIndexBit_2}
      (1- \symZ_{\symMaxUpdateVal-2})^{ \symIndexBit_2}
    \end{equation*}
    with $\symRegister = 4\symMaxUpdateVal + \langle\symIndexBit_1\symIndexBit_{2}\rangle_2$.
    The expectation of $\symRegContrib(\symRegister)$ is
    \begin{align*}
       & \symExpectation(\symRegContrib(\symRegister)) = \sum_{\symRegister=-\infty}^\infty
      \symDensityRegisterSimple(\symRegister\vert\symCardinality)
      \symRegContrib(\symRegister)
      \\
       & =
      \sum_{\symMaxUpdateVal=-\infty}^\infty\sum_{\symIndexBit_1=0}^1\sum_{\symIndexBit_2=0}^1
      \symDensityRegisterSimple(4\symMaxUpdateVal + \langle\symIndexBit_1\symIndexBit_{2}\rangle_2\vert\symCardinality)
      \symRegContrib(4\symMaxUpdateVal + \langle\symIndexBit_1\symIndexBit_{2}\rangle_2)
      \\
       & =
      \sum_{\symMaxUpdateVal=-\infty}^\infty
      \symZ_{\symMaxUpdateVal}^{\frac{1}{\symBase-1}}
      (1-\symZ_{\symMaxUpdateVal})
      \symBase^{-\symGRA\symMaxUpdateVal}
      \left(\begin{array}{@{}l}
                \symZ_{\symMaxUpdateVal-1}
                \symZ_{\symMaxUpdateVal-2}
                \symContributionCoefficient_{0}
                +
                \symZ_{\symMaxUpdateVal-1}
                (1- \symZ_{\symMaxUpdateVal-2})
                \symContributionCoefficient_{1}
                \\
                +
                (1- \symZ_{\symMaxUpdateVal-1})
                \symZ_{\symMaxUpdateVal-2}
                \symContributionCoefficient_{2}
                \\
                +
                (1- \symZ_{\symMaxUpdateVal-1})
                (1- \symZ_{\symMaxUpdateVal-2})
                \symContributionCoefficient_{3}
              \end{array}\right)
      \\
       & =
      \frac{\symNumReg^\symGRA}{\symCardinality^\symGRA(\symBase-1)^\symGRA}\sum_{\symMaxUpdateVal=-\infty}^\infty
      \symZ_{\symMaxUpdateVal}^{\frac{1}{\symBase-1}}
      (1-\symZ_{\symMaxUpdateVal})(-\ln\symZ_{\symMaxUpdateVal})^\symGRA\cdot
      \\
       & \quad \cdot\left(\begin{array}{@{}l}
                              \symZ_{\symMaxUpdateVal}^{\symBase(\symBase+1)}
                              \symContributionCoefficient_{0}
                              +
                              \symZ_{\symMaxUpdateVal}^{\symBase}
                              (1- \symZ_{\symMaxUpdateVal}^{\symBase^2})
                              \symContributionCoefficient_{1}
                              \\
                              +
                              (1- \symZ_{\symMaxUpdateVal}^{\symBase})
                              \symZ_{\symMaxUpdateVal}^{\symBase^2}
                              \symContributionCoefficient_{2}
                              +
                              (1- \symZ_{\symMaxUpdateVal}^{\symBase})
                              (1- \symZ_{\symMaxUpdateVal}^{\symBase^2})
                              \symContributionCoefficient_{3}
                            \end{array}\right)
      \\
       & =
      \frac{\symNumReg^\symGRA}{\symCardinality^\symGRA(\symBase-1)^\symGRA}\sum_{\symMaxUpdateVal=-\infty}^\infty
      \symZ_{\symMaxUpdateVal}^{\frac{1}{\symBase-1}}
      (1-\symZ_{\symMaxUpdateVal})(-\ln\symZ_{\symMaxUpdateVal})^\symGRA\cdot
      \\
       & \cdot\left(
      \symZ_{\symMaxUpdateVal}^{\symBase(\symBase+1)}(
        \symContributionCoefficient_{0}-\symContributionCoefficient_{1}-\symContributionCoefficient_{2}+\symContributionCoefficient_{3})
      +
      \symZ_{\symMaxUpdateVal}^{\symBase^2}(\symContributionCoefficient_{2}-\symContributionCoefficient_{3})
      +
      \symZ_{\symMaxUpdateVal}^{\symBase}(\symContributionCoefficient_{1}-\symContributionCoefficient_{3})
      +
      \symContributionCoefficient_{3}\right).
    \end{align*}
    \cref{lem:approximation} with
    \begin{equation*}
      \symSomeFunc(\symX) = \symX^{\frac{1}{\symBase-1}}(-\ln \symX)^\symGRA(1-\symX)\left(\begin{array}{@{}l}
          \symX^{\symBase(\symBase+1)}(
          \symContributionCoefficient_{0}-\symContributionCoefficient_{1}-\symContributionCoefficient_{2}+\symContributionCoefficient_{3})
          \\
          +
          \symX^{\symBase^2}(\symContributionCoefficient_{2}-\symContributionCoefficient_{3})
          +
          \symX^{\symBase}(\symContributionCoefficient_{1}-\symContributionCoefficient_{3})
          +
          \symContributionCoefficient_{3}\end{array}\right)
    \end{equation*}
    yields
    \begin{align*}
       & \symExpectation(\symRegContrib(\symRegister)) \approx
      \frac{\symNumReg^\symGRA}{\symCardinality^\symGRA(\symBase-1)^\symGRA\ln\symBase}
      \int_0^\infty
      e^{-\symY\frac{1}{\symBase-1}}
      (1-e^{-\symY})\symY^{\symGRA-1}\cdot
      \\
       & \quad\cdot
      \left(\begin{array}{@{}l}
                e^{-\symBase(\symBase+1)\symY }(
                \symContributionCoefficient_{0}-\symContributionCoefficient_{1}-\symContributionCoefficient_{2}+\symContributionCoefficient_{3})
                \\
                +e^{-\symBase^2\symY }(\symContributionCoefficient_{2}-\symContributionCoefficient_{3})
                +e^{-\symBase\symY }(\symContributionCoefficient_{1}-\symContributionCoefficient_{3})
                +
                \symContributionCoefficient_{3}
              \end{array}\right)
      d\symY.
    \end{align*}
    Using $\int_0^\infty e^{-\symZ\symY}(1-e^{-\symY})\symY^{\symGRA-1}d\symY = \frac{\symGammaFunc(\symGRA)}{\symZ^\symGRA}-\frac{\symGammaFunc(\symGRA)}{(\symZ+1)^\symGRA}$ gives
    \begin{align*}
       & \symExpectation(\symRegContrib(\symRegister)) \approx
      \frac{\symNumReg^\symGRA\symGammaFunc(\symGRA) }{\symCardinality^\symGRA(\symBase-1)^\symGRA\ln\symBase}
      \cdot
      \\
       & \quad\cdot
      \left(\begin{array}{@{}l}
                (\symContributionCoefficient_{0}-\symContributionCoefficient_{1}-\symContributionCoefficient_{2}+\symContributionCoefficient_{3})\cdot
                \\
                \textstyle\quad\cdot ((\frac{1}{\symBase-1}+\symBase(\symBase+1))^{-\symGRA}-(\frac{\symBase}{\symBase-1}+\symBase(\symBase+1))^{-\symGRA})
                \\
                +(\symContributionCoefficient_{2}-\symContributionCoefficient_{3}) ((\frac{1}{\symBase-1}+\symBase^2)^{-\symGRA}-(\frac{\symBase}{\symBase-1}+\symBase^2)^{-\symGRA})
                \\
                +(\symContributionCoefficient_{1}-\symContributionCoefficient_{3}) ((\frac{1}{\symBase-1}+\symBase)^{-\symGRA}-(\frac{\symBase}{\symBase-1}+\symBase)^{-\symGRA})
                \\
                +\symContributionCoefficient_{3}((\frac{1}{\symBase-1})^{-\symGRA}-(\frac{\symBase}{\symBase-1} )^{-\symGRA})
              \end{array}\right)
      \\
       & =\frac{\symNumReg^\symGRA\symGammaFunc(\symGRA) }{\symCardinality^\symGRA\ln\symBase}
      \left(\begin{array}{@{}l}
                (\symContributionCoefficient_{0}-\symContributionCoefficient_{1}-\symContributionCoefficient_{2}+\symContributionCoefficient_{3}) ((\symBase^3-\symBase+1)^{-\symGRA}-\symBase^{-3\symGRA})
                \\
                +(\symContributionCoefficient_{2}-\symContributionCoefficient_{3}) ((\symBase^3-\symBase^2+1)^{-\symGRA}-(\symBase^3-\symBase^2+\symBase)^{-\symGRA})
                \\
                +(\symContributionCoefficient_{1}-\symContributionCoefficient_{3}) ((\symBase^2-\symBase+1)^{-\symGRA}-\symBase^{-2\symGRA})
                \\
                +\symContributionCoefficient_{3}(1-\symBase^{-\symGRA})
              \end{array}\right)
      \\
       & =\frac{\symNumReg^\symGRA\symGammaFunc(\symGRA) }{\symCardinality^\symGRA\ln\symBase}\left(
      \symContributionCoefficient_0 \symNewConstant_0(\symGRA)
      +
      \symContributionCoefficient_1 \symNewConstant_1(\symGRA)
      +
      \symContributionCoefficient_2 \symNewConstant_2(\symGRA)
      +\symContributionCoefficient_3 \symNewConstant_3(\symGRA)
      \right)
      =\frac{\symNumReg^\symGRA}{\symCardinality^\symGRA}.
    \end{align*}
    It remains to show that $\symNewConstant_0(\symGRA)$, $\symNewConstant_1(\symGRA)$, $\symNewConstant_2(\symGRA)$, and $\symNewConstant_3(\symGRA)$ are positive. The proof is trivial for $\symNewConstant_0(\symGRA)$. For $\symNewConstant_1(\symGRA)$ and $\symNewConstant_2(\symGRA)$ we use the strict convexity of $\symSomeFunc(\symX) = \symX^{-\symGRA}$ together with
    \begin{gather*}
      \symBase^3
      >
      \symBase^3 - \symBase + 1
      >
      \symBase^3 - \symBase^2 + \symBase
      >
      \symBase^3 - \symBase^2 + 1 > \symBase^2-\symBase +1>
      \symBase
      >
      1,
      \\
      \symBase^3
      >
      \symBase^3 - \symBase + 1
      >
      \symBase^3 - \symBase^2 + \symBase
      >
      \symBase^2
      >
      \symBase^2-\symBase +1
      >
      \symBase
      >
      1
    \end{gather*}
    to apply Karamata's inequality, which gives
    \begin{align*}
       & \symSomeFunc(\symBase^{3}) + \symSomeFunc(\symBase^2-\symBase+1) > \symSomeFunc(\symBase^3-\symBase+1) + \symSomeFunc(\symBase^2)
      \\
       & \Rightarrow
      \symBase^{-3\symGRA} + (\symBase^2-\symBase+1)^{-\symGRA} > (\symBase^3-\symBase+1)^{-\symGRA} + \symBase^{-2\symGRA}
      \\
       & \Rightarrow \symNewConstant_1(\symGRA) > 0
    \end{align*}
    and
    \begin{align*}
       & \symSomeFunc(\symBase^{3}) + \symSomeFunc(\symBase^3-\symBase^2+1) > \symSomeFunc(\symBase^3-\symBase+1) + \symSomeFunc(\symBase^3-\symBase^2+\symBase)
      \\
       & \Rightarrow
      \symBase^{-3\symGRA} + (\symBase^3-\symBase^2+1)^{-\symGRA} > (\symBase^3-\symBase+1)^{-\symGRA} + (\symBase^3-\symBase^2+\symBase)^{-\symGRA}
      \\
       & \Rightarrow \symNewConstant_2(\symGRA) > 0.
    \end{align*}
    For $\symNewConstant_3(\symGRA)$ we use the strict convexity of $\symSomeFunc(\symX) = \symX^{-\symGRA} - (\symX + \symBase^3-\symBase^2)^{-\symGRA}$ and apply again Karamata's inequality
    \begin{align*}
       & \symSomeFunc(\symBase^{2}) + \symSomeFunc(1) > \symSomeFunc(\symBase^2-\symBase+1) + \symSomeFunc(\symBase)
      \\
       & \Rightarrow\symBase^{-2\symGRA} - \symBase^{-3\symGRA} + 1 - (1 + \symBase^3-\symBase^2)^{-\symGRA}
      \\
       & \quad >(\symBase^2-\symBase+1)^{-\symGRA} - (\symBase^3-\symBase+1)^{-\symGRA} + \symBase^{-\symGRA} - (\symBase^3-\symBase^2+\symBase)^{-\symGRA}
      \\
       & \Rightarrow \symNewConstant_3(\symGRA) > 0.
    \end{align*}
  \end{proof}

  \begin{lemma}
    \label{lem:fgra_second_moment}
    For a random variable $\symRegister$ distributed according to \eqref{equ:simple_register_pmf}, the expectation of $(\symRegContrib(\symRegister))^2$ with $\symRegContrib$, $\symNewConstant_0$, $\symNewConstant_1$, $\symNewConstant_2$, and $\symNewConstant_3$ as defined in \cref{lem:expectation_gamma_new} can be approximated using \cref{lem:approximation} by
    \begin{equation}
      \label{equ:second_moment_new_contribution}
      \symExpectation((\symRegContrib(\symRegister))^2)
      \approx
      \textstyle
      \frac{\symNumReg^{2\symGRA}\symGammaFunc(2\symGRA) \ln\symBase}{\symCardinality^{2\symGRA}(\symGammaFunc(\symGRA))^2}
      \frac{ \symContributionCoefficient_0^2 \symNewConstant_0(2\symGRA)
        +
        \symContributionCoefficient_1^2 \symNewConstant_1(2\symGRA)
        +
        \symContributionCoefficient_2^2 \symNewConstant_2(2\symGRA)
        +\symContributionCoefficient_3^2 \symNewConstant_3(2\symGRA)
      }{(\symContributionCoefficient_0 \symNewConstant_0(\symGRA)
        +
        \symContributionCoefficient_1 \symNewConstant_1(\symGRA)
        +
        \symContributionCoefficient_2 \symNewConstant_2(\symGRA)
        +\symContributionCoefficient_3 \symNewConstant_3(\symGRA)
        )^2}.
    \end{equation}
  \end{lemma}
  \begin{proof}
    We define $\symGRA' := 2\symGRA$ and
    \begin{equation*}
      \symContributionCoefficient_\symIndexJ' := \symContributionCoefficient_\symIndexJ^2 \frac{(\symContributionCoefficient_0 \symNewConstant_0(\symGRA)
        +
        \symContributionCoefficient_1 \symNewConstant_1(\symGRA)
        +
        \symContributionCoefficient_2 \symNewConstant_2(\symGRA)
        +\symContributionCoefficient_3 \symNewConstant_3(\symGRA))^2
      }{ \symContributionCoefficient_0^2 \symNewConstant_0(2\symGRA)
        +
        \symContributionCoefficient_1^2 \symNewConstant_1(2\symGRA)
        +
        \symContributionCoefficient_2^2 \symNewConstant_2(2\symGRA)
        +\symContributionCoefficient_3^2 \symNewConstant_3(2\symGRA)
      }\frac{(\symGammaFunc(\symGRA))^2}{\symGammaFunc(2\symGRA)\ln \symBase}.
    \end{equation*}
    Since \eqref{equ:norm_condition} implies
    \begin{equation*}
      \symContributionCoefficient'_0 \symNewConstant_0(\symGRA')
      +
      \symContributionCoefficient'_1 \symNewConstant_1(\symGRA')
      +
      \symContributionCoefficient'_2 \symNewConstant_2(\symGRA')
      +\symContributionCoefficient'_3 \symNewConstant_3(\symGRA')
      =\frac{\ln\symBase}{\symGammaFunc(\symGRA')},
    \end{equation*}
    we can apply \cref{lem:expectation_gamma_new} giving
    \begin{equation*}
      \symExpectation(\symBase^{-\symGRA'\lfloor\symRegister/4\rfloor} \symContributionCoefficient'_{\symRegister\bmod 4}) \approx
      \frac{\symNumReg^{\symGRA'}}{\symCardinality^{\symGRA'}}
    \end{equation*}
    which is equivalent to
    \begin{multline*}
      \symExpectation(\symBase^{-2\symGRA\lfloor\symRegister/4\rfloor}\symContributionCoefficient^2_{\symRegister\bmod 4}) \cdot\\\cdot
      \frac{(\symContributionCoefficient_0 \symNewConstant_0(\symGRA)
        +
        \symContributionCoefficient_1 \symNewConstant_1(\symGRA)
        +
        \symContributionCoefficient_2 \symNewConstant_2(\symGRA)
        +\symContributionCoefficient_3 \symNewConstant_3(\symGRA))^2
      }{ \symContributionCoefficient_0^2 \symNewConstant_0(2\symGRA)
        +
        \symContributionCoefficient_1^2 \symNewConstant_1(2\symGRA)
        +
        \symContributionCoefficient_2^2 \symNewConstant_2(2\symGRA)
        +\symContributionCoefficient_3^2 \symNewConstant_3(2\symGRA)
      }\frac{(\symGammaFunc(\symGRA))^2}{\symGammaFunc(2\symGRA)\ln\symBase}
      \approx
      \frac{\symNumReg^{2\symGRA}}{\symCardinality^{2\symGRA}}
    \end{multline*}
    and finally leads to \eqref{equ:second_moment_new_contribution}.
  \end{proof}

  \begin{lemma}
    \label{lem:fgra_minimum}
    The approximation for $\symExpectation((\symRegContrib(\symRegister))^2)$ given in \cref{lem:fgra_second_moment} is minimized for
    \begin{equation}
      \label{equ:fgra_min_coefficients}
      \textstyle
      \symContributionCoefficient_{\symIndexJ} = \frac{\ln \symBase}{\symGammaFunc(\symGRA)}\frac{\symNewConstant_\symIndexJ(\symGRA) }{\symNewConstant_\symIndexJ(2\symGRA) }\left(
      \frac{\symNewConstant_0^2(\symGRA) }{\symNewConstant_0(2\symGRA)}
      +
      \frac{\symNewConstant_1^2(\symGRA) }{\symNewConstant_1(2\symGRA)}
      +
      \frac{\symNewConstant_2^2(\symGRA) }{\symNewConstant_2(2\symGRA)}
      +
      \frac{\symNewConstant_3^2(\symGRA) }{\symNewConstant_3(2\symGRA)}
      \right)^{\!-1}
    \end{equation}
    and the corresponding minimum is given by
    \begin{equation}
      \label{equ:fgra_min}
      \min_{\symContributionCoefficient_0,\symContributionCoefficient_1,\symContributionCoefficient_2,\symContributionCoefficient_3}
      \symExpectation((\symRegContrib(\symRegister))^2)
      \approx
      \frac{\symNumReg^{2\symGRA}}{\symCardinality^{2\symGRA}}\frac{\symGammaFunc(2\symGRA) \ln\symBase}{(\symGammaFunc(\symGRA))^2}
      \left(
      \sum_{\symIndexJ=0}^3
      \frac{\symNewConstant_\symIndexJ^2(\symGRA) }{\symNewConstant_\symIndexJ(2\symGRA)}
      \right)^{\!-1}.
    \end{equation}
  \end{lemma}

  \begin{proof}
    The Cauchy-Schwarz inequality
    \begin{multline*}
      \left(
      \sum_{\symIndexJ=0}^3\symContributionCoefficient_\symIndexJ^2 \symNewConstant_\symIndexJ(2\symGRA)
      \right)
      \left(
      \sum_{\symIndexJ=0}^3
      \frac{\symNewConstant_\symIndexJ^2(\symGRA) }{\symNewConstant_\symIndexJ(2\symGRA)}
      \right)
      \\
      =
      \left(
      \sum_{\symIndexJ=0}^3\left(\symContributionCoefficient_\symIndexJ \sqrt{\symNewConstant_\symIndexJ(2\symGRA)}\right)^2
      \right)
      \left(
      \sum_{\symIndexJ=0}^3
      \left(\frac{\symNewConstant_\symIndexJ(\symGRA) }{\sqrt{\symNewConstant_\symIndexJ(2\symGRA)}}\right)^2
      \right)
      \\
      \geq
      \left(
      \sum_{\symIndexJ=0}^3
      \left(\symContributionCoefficient_\symIndexJ \sqrt{\symNewConstant_\symIndexJ(2\symGRA)}\right)
      \left(\frac{\symNewConstant_\symIndexJ(\symGRA) }{\sqrt{\symNewConstant_\symIndexJ(2\symGRA)}}\right) \right)^2
      =
      \left(
      \sum_{\symIndexJ=0}^3
      \symContributionCoefficient_\symIndexJ\symNewConstant_\symIndexJ(\symGRA)
      \right)^2
    \end{multline*}
    gives
    \begin{equation*}
      \frac{
        \symContributionCoefficient_0^2 \symNewConstant_0(2\symGRA)
        +
        \symContributionCoefficient_1^2 \symNewConstant_1(2\symGRA)
        +
        \symContributionCoefficient_2^2 \symNewConstant_2(2\symGRA)
        +\symContributionCoefficient_3^2 \symNewConstant_3(2\symGRA)
      }{
        \left(
        \symContributionCoefficient_0 \symNewConstant_0(\symGRA)
        +
        \symContributionCoefficient_1 \symNewConstant_1(\symGRA)
        +
        \symContributionCoefficient_2 \symNewConstant_2(\symGRA)
        +\symContributionCoefficient_3 \symNewConstant_3(\symGRA)
        \right)^2
      }
      \geq
      \frac{1}{
        \sum_{\symIndexJ=0}^3
        \frac{\symNewConstant_\symIndexJ^2(\symGRA) }{\symNewConstant_\symIndexJ(2\symGRA)}
      }
    \end{equation*}
    with equality for
    \begin{equation}
      \label{equ:min_condition}
      \symContributionCoefficient_0 \frac{\symNewConstant_0(2\symGRA) }{\symNewConstant_0(\symGRA) }
      =
      \symContributionCoefficient_1 \frac{\symNewConstant_1(2\symGRA) }{\symNewConstant_1(\symGRA) }
      =
      \symContributionCoefficient_2 \frac{\symNewConstant_2(2\symGRA) }{\symNewConstant_2(\symGRA) }
      =
      \symContributionCoefficient_3\frac{ \symNewConstant_3(2\symGRA) }{\symNewConstant_3(\symGRA) }.
    \end{equation}
    Replacing the left-hand side expression of the inequality that also appears in \eqref{equ:second_moment_new_contribution} by the corresponding right-hand side leads to the claimed minimum \eqref{equ:fgra_min}. \eqref{equ:min_condition} defines the corresponding ratios between $\symContributionCoefficient_0$, $\symContributionCoefficient_1$, $\symContributionCoefficient_2$, and $\symContributionCoefficient_3$. Scaling those values to satisfy constraint \eqref{equ:norm_condition} finally gives \eqref{equ:fgra_min_coefficients}.
  \end{proof}

  \begin{lemma}
    \label{lem:cardinality_estimator}
    Consider $\symNumReg$ independent random variables $\symRegister_\symRegAddr$ that are distributed according to \eqref{equ:simple_register_pmf} and a function $\symRegContrib(\symRegister)$ for which $
      \symExpectation(\symRegContrib(\symRegister_\symRegAddr))
      =
      \frac{\symNumReg^\symGRA}{\symCardinality^\symGRA}$.
    The second-order delta method \cite{Casella2002} yields the estimator
    \begin{equation}
      \label{equ:lem_estimator}
      \symCardinalityEstimator
      =
      \symNumReg^{1+\frac{1}{\symGRA}}
      \cdot
      \left(\sum_{\symRegAddr=0}^{\symNumReg-1} \symRegContrib(\symRegister_\symRegAddr)
      \right)^{\!-\frac{1}{\symGRA}}
      \cdot
      \left(1 + \frac{1+\symGRA}{2}\frac{\symVarianceFactor}{\symNumReg}\right)^{\!-1}
    \end{equation}
    for $\symCardinality$ with relative variance
    \begin{equation*}
      \textstyle \symVariance(\frac{\symCardinalityEstimator}{\symCardinality})
      \approx
      \frac{\symVarianceFactor}{\symNumReg} + \symComplexity(\frac{1}{\symNumReg^2})
    \end{equation*}
    and
    \begin{equation}
      \label{equ:lem_variance_factor}
      \symVarianceFactor
      =
      \frac{1}{\symGRA^2}\frac{\symCardinality^{2\symGRA}}{\symNumReg^{2\symGRA}}
      \symVariance(\symRegContrib(\symRegister_\symRegAddr))
      =
      \frac{1}{\symGRA^2}
      \left(
      \frac{\symCardinality^{2\symGRA}}{\symNumReg^{2\symGRA}}\symExpectation((\symRegContrib(\symRegister_\symRegAddr))^2)-1
      \right).
    \end{equation}
  \end{lemma}

  \begin{proof}
    Consider the statistic $\symStatisticX_\symNumReg := \frac{1}{\symNumReg^{1+\symGRA}}\sum_{\symRegAddr=0}^{\symNumReg-1}\symRegContrib(\symRegister_\symRegAddr)$ for which $\symExpectation(\symStatisticX_\symNumReg) = \frac{1}{\symCardinality^{\symGRA}}$
    and $\symVariance(\symStatisticX_\symNumReg) = \frac{\symVarianceFactor \symGRA^2}{\symNumReg\symCardinality^{2\symGRA}}$. The second-order delta method gives
    \begin{equation*}
      \textstyle\symTransformation(\symStatisticX_\symNumReg) - \symTransformation(\frac{1}{\symCardinality^{\symGRA}})
      \approx
      (\symStatisticX_\symNumReg - \frac{1}{\symCardinality^{\symGRA}})\
      \symTransformation'(\frac{1}{\symCardinality^{\symGRA}})
      +
      \frac{1}{2}
      (\symStatisticX_\symNumReg - \frac{1}{\symCardinality^{\symGRA}})^{2}
      \symTransformation''(\frac{1}{\symCardinality^{\symGRA}})
    \end{equation*}
    with the transformation $\symTransformation(\symX) := \symX^{-\frac{1}{\symGRA}}$ and its derivatives $\symTransformation'(\symX) = -\frac{1}{\symGRA}\symX^{-\frac{1}{\symGRA}-1}$ and $\symTransformation''(\symX) = \frac{1+\symGRA}{\symGRA^2}\symX^{-\frac{1}{\symGRA}-2}$. Taking the expectation on both sides gives
    \begin{equation*}
      \symExpectation(\symStatisticX_\symNumReg^{-\frac{1}{\symGRA}}) - \symCardinality
      \approx
      \frac{1+\symGRA}{2\symGRA^2}
      \symVariance(\symStatisticX_\symNumReg) \symCardinality^{2\symGRA+1}
      \approx
      \symCardinality
      \frac{1+\symGRA}{2}
      \frac{\symVarianceFactor}{\symNumReg}
    \end{equation*}
    and therefore
    \begin{equation*}
      \symCardinality
      \approx
      \frac{\symExpectation(\symStatisticX_\symNumReg^{-\frac{1}{\symGRA}})}{1 +
        \frac{1+\symGRA}{2}\frac{\symVarianceFactor}{\symNumReg}
      }
      =
      \symExpectation\left(
      \symNumReg^{1+\frac{1}{\symGRA}}
      \cdot
      \left(\sum_{\symRegAddr=0}^{\symNumReg-1} \symRegContrib(\symRegister_\symRegAddr)
        \right)^{\!-\frac{1}{\symGRA}}
      \cdot
      \left(1 + \frac{1+\symGRA}{2}\frac{\symVarianceFactor}{\symNumReg}\right)^{\!-1}
      \right)
    \end{equation*}
    which directly leads to \eqref{equ:lem_estimator}.

    For the variance we use
    \begin{equation*}
      \left(\symTransformation(\symStatisticX_\symNumReg) - \symTransformation(\symCardinality^{-\symGRA})\right)^2
      \approx
      \left(\symStatisticX_\symNumReg - \symCardinality^{-\symGRA}\right)^2
      (\symTransformation'(\symCardinality^{-\symGRA}))^2
    \end{equation*}
    and again take the expectation on both sides which leads to
    \begin{equation*}
      \symExpectation((\symStatisticX_\symNumReg^{-\frac{1}{\symGRA}} - \symCardinality)^2)
      \approx
      \symVariance(\symStatisticX_\symNumReg)
      (\symTransformation'(\symCardinality^{-\symGRA}))^2
      =
      \frac{\symVarianceFactor\symCardinality^2}{\symNumReg}.
    \end{equation*}
    This gives for the relative variance of estimator $\symCardinalityEstimator$
    \begin{align*}
       & \symVariance\left(\frac{\symCardinalityEstimator}{\symCardinality}\right)
      =
      \frac{1}{\symCardinality^2}
      \symExpectation\left(\left(\frac{1}{1 +
          \frac{1+\symGRA}{2}\frac{\symVarianceFactor}{\symNumReg}
        }
        \symStatisticX_\symNumReg^{-\frac{1}{\symGRA}}
        -\symCardinality\right)^2\right)
      \\
       & =
      \frac{1}{\symCardinality^2}
      \symExpectation\left(\frac{(\symStatisticX_\symNumReg^{-\frac{1}{\symGRA}} - \symCardinality)^2 -2(\symCardinality\frac{1+\symGRA}{2}\frac{\symVarianceFactor}{\symNumReg})(\symStatisticX_\symNumReg^{-\frac{1}{\symGRA}} - \symCardinality)+ (\symCardinality\frac{1+\symGRA}{2}\frac{\symVarianceFactor}{\symNumReg})^2}{(1 +
          \frac{1+\symGRA}{2}\frac{\symVarianceFactor}{\symNumReg})^2}
      \right)
      \\
       & \approx
      \frac{1}{\symCardinality^2}
      \left(\frac{\frac{\symVarianceFactor\symCardinality^2}{\symNumReg} - (\symCardinality\frac{1+\symGRA}{2}\frac{\symVarianceFactor}{\symNumReg})^2}{(1 +
          \frac{1+\symGRA}{2}\frac{\symVarianceFactor}{\symNumReg})^2}
      \right)
      =
      \frac{\symVarianceFactor}{\symNumReg}\frac{1 - (\frac{1+\symGRA}{2})^2\frac{\symVarianceFactor}{\symNumReg}}{(1 +
        \frac{1+\symGRA}{2}\frac{\symVarianceFactor}{\symNumReg})^2}
      =
      {\textstyle\frac{\symVarianceFactor}{\symNumReg} + \symComplexity(\frac{1}{\symNumReg^2})}.
    \end{align*}
  \end{proof}

  \begin{lemma}
    \label{lem:conditional_pmf}
    Let $\symUpdateEvent_\symUpdateVal$ be independent events that occur with probability $\symProbability(\symUpdateEvent_\symUpdateVal) =
      1 - \symZ_\symUpdateVal$ with $\symZ_\symUpdateVal := e^{-\frac{\symCardinality}{\symNumReg 2^{\symUpdateVal}}}$. Consider $\symRegisterSimple = 4\symMaxUpdateValSimple + \langle\symIndexBitSimple_1\symIndexBitSimple_2\rangle_2$ defined by $\symMaxUpdateValSimple = \max \lbrace \symUpdateVal \vert \symUpdateEvent_\symUpdateVal\rbrace$, $\symIndexBitSimple_1 = [\symUpdateEvent_{\symMaxUpdateValSimple-1}]$, and $\symIndexBitSimple_2 = [\symUpdateEvent_{\symMaxUpdateValSimple-2}]$ where we used the Iverson bracket notation.
    Furthermore, consider $\symRegister = 4 \symMaxUpdateVal + \langle\symIndexBit_1\symIndexBit_2\rangle_2$ defined by $\symMaxUpdateVal = \max(0, \min(\symMaxUpdateValMax, \symMaxUpdateValSimple))$, $\symIndexBit_1 = [\symUpdateEvent_{\symMaxUpdateVal-1}\wedge \symMaxUpdateVal > 1]$, and $\symIndexBit_2 = [\symUpdateEvent_{\symMaxUpdateVal-2}\wedge \symMaxUpdateVal > 2]$ with $\symMaxUpdateValMax\geq 3$. Then the conditional \ac{PMF} of $\symRegisterSimple$ for known $\symRegister$ is given by
    \begin{equation}
      \label{equ:cond_pmf}
      \scriptstyle\symDensityRegisterSimple(\symRegisterSimple\vert \symRegister,\symCardinality)
      =
      \begin{cases}
        \scriptstyle
        \frac{\symZ_{\symMaxUpdateVal}
          (1-\symZ_{\symMaxUpdateVal})
          \prod_{\symIndexJ = 1}^2
          \symZ_{\symMaxUpdateVal-\symIndexJ}^{1-\symIndexBit_\symIndexJ}
          (1- \symZ_{\symMaxUpdateVal-\symIndexJ})^{ \symIndexBit_\symIndexJ}}{\symZ_{0}}
          & \scriptstyle
        \symRegister=0,\ \symRegisterSimple = 4\symMaxUpdateVal + \langle\symIndexBit_1\symIndexBit_2\rangle_2,\ \symMaxUpdateVal\leq 0,
        \\ \scriptstyle
        \symZ_{0}^{1-\symIndexBit_1}
        (1- \symZ_{0})^{ \symIndexBit_1}
        \symZ_{-1}^{1-\symIndexBit_2}
        (1- \symZ_{-1})^{ \symIndexBit_2}
          & \scriptstyle
        \symRegister=4,\ \symRegisterSimple = 4 + \langle\symIndexBit_1\symIndexBit_2\rangle_2,
        \\ \scriptstyle
        \symZ_{0}^{1-\symIndexBit_1}
        (1- \symZ_{0})^{ \symIndexBit_1}
          & \scriptstyle
        \symRegister=8,\ \symRegisterSimple = 8 + \langle\symIndexBit_1\rangle_2,
        \\ \scriptstyle
        \symZ_{0}^{1-\symIndexBit_1}
        (1- \symZ_{0})^{ \symIndexBit_1}
          & \scriptstyle
        \symRegister=10,\ \symRegisterSimple = 10 + \langle\symIndexBit_1\rangle_2,
        \\ \scriptstyle
        1
          & \scriptstyle
        12 \leq \symRegister < 4\symMaxUpdateValMax,\ \symRegisterSimple=\symRegister,
        \\ \scriptstyle
        \frac{\symZ_\symMaxUpdateValMax(1-\symZ_\symMaxUpdateValMax)}{1-\symZ_{\symMaxUpdateValMax-1}},
          & \scriptstyle
        \symRegister \geq 4\symMaxUpdateValMax,\ \symRegisterSimple=\symRegister,
        \\ \scriptstyle
        \frac{\symZ_{\symMaxUpdateValMax+1}(1-\symZ_{\symMaxUpdateValMax+1})\symZ_{\symMaxUpdateValMax}}{1-\symZ_{\symMaxUpdateValMax-1}},
          & \scriptstyle
        \symRegister = 4\symMaxUpdateValMax + \langle\symIndexBit_1\symIndexBit_2\rangle_2,\ \symRegisterSimple=4\symMaxUpdateValMax+4 + \symIndexBit_1,
        \\ \scriptstyle
        \frac{\symZ_{\symMaxUpdateValMax+1}(1-\symZ_{\symMaxUpdateValMax+1})(1-\symZ_{\symMaxUpdateValMax})}{1-\symZ_{\symMaxUpdateValMax-1}},
          & \scriptstyle
        \symRegister = 4\symMaxUpdateValMax + \langle\symIndexBit_1\symIndexBit_2\rangle_2,\ \symRegisterSimple=4\symMaxUpdateValMax+6 + \symIndexBit_1,
        \\ \scriptstyle
        \frac{
          \symZ_{\symMaxUpdateVal}
          (1-\symZ_{\symMaxUpdateVal})
          \prod_{\symIndexJ = 1}^2
          \symZ_{\symMaxUpdateVal-\symIndexJ}^{1-\symIndexBit_\symIndexJ}
          (1- \symZ_{\symMaxUpdateVal-\symIndexJ})^{ \symIndexBit_\symIndexJ}}{1-\symZ_{\symMaxUpdateValMax-1}}
          & \scriptstyle
        \symRegister \geq 4\symMaxUpdateValMax,\ \symRegisterSimple= 4\symMaxUpdateVal+ \langle\symIndexBit_1\symIndexBit_2\rangle_2,\ \symMaxUpdateVal\geq \symMaxUpdateValMax+2,
        \\ \scriptstyle
        0 & \scriptstyle
        \text{else}.
      \end{cases}
    \end{equation}
  \end{lemma}

  \begin{proof}
    The \ac{PMF} of $\symRegisterSimple$ is given by
    \begin{equation*}
      \scriptstyle\symDensityRegisterSimple(\symRegisterSimple\vert\symCardinality) =
      \begin{cases}
        \scriptstyle\symProbability(\overline\symUpdateEvent_{\symMaxUpdateVal-2} \wedge \overline\symUpdateEvent_{\symMaxUpdateVal-1} \wedge \symUpdateEvent_\symMaxUpdateVal\wedge \bigwedge_{\symUpdateVal=\symMaxUpdateVal+1}^\infty \overline\symUpdateEvent_{\symUpdateVal})
         &
        \scriptstyle\symRegisterSimple= 4\symMaxUpdateVal,
        \\
        \scriptstyle\symProbability(\symUpdateEvent_{\symMaxUpdateVal-2} \wedge \overline\symUpdateEvent_{\symMaxUpdateVal-1} \wedge \symUpdateEvent_\symMaxUpdateVal\wedge \bigwedge_{\symUpdateVal=\symMaxUpdateVal+1}^\infty \overline\symUpdateEvent_{\symUpdateVal})
         &
        \scriptstyle\symRegisterSimple= 4\symMaxUpdateVal +1,
        \\
        \scriptstyle\symProbability(\overline \symUpdateEvent_{\symMaxUpdateVal-2} \wedge \symUpdateEvent_{\symMaxUpdateVal-1} \wedge \symUpdateEvent_\symMaxUpdateVal\wedge \bigwedge_{\symUpdateVal=\symMaxUpdateVal+1}^\infty \overline\symUpdateEvent_{\symUpdateVal})
         &
        \scriptstyle\symRegisterSimple= 4\symMaxUpdateVal +2,
        \\
        \scriptstyle\symProbability(\symUpdateEvent_{\symMaxUpdateVal-2} \wedge \symUpdateEvent_{\symMaxUpdateVal-1} \wedge \symUpdateEvent_\symMaxUpdateVal\wedge \bigwedge_{\symUpdateVal=\symMaxUpdateVal+1}^\infty \overline\symUpdateEvent_{\symUpdateVal})
         &
        \scriptstyle\symRegisterSimple= 4\symMaxUpdateVal +3,
      \end{cases}
    \end{equation*}
    and that of $\symRegister$ is given by
    \begin{equation*}
      \scriptstyle\symDensityRegister(\symRegister\vert\symCardinality) =
      \begin{cases}
        \scriptstyle\symProbability(\bigwedge_{\symUpdateVal=1}^\infty \overline\symUpdateEvent_{\symUpdateVal} )
         &
        \scriptstyle\symRegister=0,
        \\
        \scriptstyle\symProbability(\symUpdateEvent_{1} \wedge \bigwedge_{\symUpdateVal=2}^\infty \overline\symUpdateEvent_{\symUpdateVal} )
         &
        \scriptstyle\symRegister=4,
        \\
        \scriptstyle\symProbability(\overline \symUpdateEvent_{1} \wedge \symUpdateEvent_{2} \wedge \bigwedge_{\symUpdateVal=3}^\infty \overline\symUpdateEvent_{\symUpdateVal} )
         &
        \scriptstyle\symRegister=8,
        \\
        \scriptstyle\symProbability(\symUpdateEvent_{1}\wedge \symUpdateEvent_{2} \wedge \bigwedge_{\symUpdateVal=3}^\infty \overline\symUpdateEvent_{\symUpdateVal} )
         &
        \scriptstyle\symRegister=10,
        \\
        \scriptstyle\symProbability(\overline\symUpdateEvent_{\symMaxUpdateVal-2} \wedge \overline\symUpdateEvent_{\symMaxUpdateVal-1} \wedge \symUpdateEvent_\symMaxUpdateVal\wedge \bigwedge_{\symUpdateVal=\symMaxUpdateVal+1}^\infty \overline\symUpdateEvent_{\symUpdateVal})
         &
        \scriptstyle\symRegister= 4\symMaxUpdateVal,\ 3 \leq \symMaxUpdateVal < \symMaxUpdateValMax,
        \\
        \scriptstyle\symProbability(\symUpdateEvent_{\symMaxUpdateVal-2} \wedge \overline\symUpdateEvent_{\symMaxUpdateVal-1} \wedge \symUpdateEvent_\symMaxUpdateVal\wedge \bigwedge_{\symUpdateVal=\symMaxUpdateVal+1}^\infty \overline\symUpdateEvent_{\symUpdateVal})
         &
        \scriptstyle\symRegister= 4\symMaxUpdateVal +1,\ 3 \leq \symMaxUpdateVal < \symMaxUpdateValMax,
        \\
        \scriptstyle\symProbability(\overline \symUpdateEvent_{\symMaxUpdateVal-2} \wedge \symUpdateEvent_{\symMaxUpdateVal-1} \wedge \symUpdateEvent_\symMaxUpdateVal\wedge \bigwedge_{\symUpdateVal=\symMaxUpdateVal+1}^\infty \overline\symUpdateEvent_{\symUpdateVal})
         &
        \scriptstyle\symRegister= 4\symMaxUpdateVal +2,\ 3 \leq \symMaxUpdateVal < \symMaxUpdateValMax,
        \\
        \scriptstyle\symProbability(\symUpdateEvent_{\symMaxUpdateVal-2} \wedge \symUpdateEvent_{\symMaxUpdateVal-1} \wedge \symUpdateEvent_\symMaxUpdateVal\wedge \bigwedge_{\symUpdateVal=\symMaxUpdateVal+1}^\infty \overline\symUpdateEvent_{\symUpdateVal})
         &
        \scriptstyle\symRegister= 4\symMaxUpdateVal +3,\ 3 \leq \symMaxUpdateVal < \symMaxUpdateValMax,
        \\
        \scriptstyle\symProbability(\overline\symUpdateEvent_{\symMaxUpdateValMax-2} \wedge \overline\symUpdateEvent_{\symMaxUpdateValMax-1} \wedge \bigvee_{\symUpdateVal=\symMaxUpdateValMax}^\infty \symUpdateEvent_{\symUpdateVal})
         &
        \scriptstyle\symRegister= 4\symMaxUpdateValMax,
        \\
        \scriptstyle\symProbability(\symUpdateEvent_{\symMaxUpdateValMax-2} \wedge \overline\symUpdateEvent_{\symMaxUpdateValMax-1} \wedge \bigvee_{\symUpdateVal=\symMaxUpdateValMax}^\infty \symUpdateEvent_{\symUpdateVal})
         &
        \scriptstyle\symRegister= 4\symMaxUpdateValMax +1,
        \\
        \scriptstyle\symProbability(\overline\symUpdateEvent_{\symMaxUpdateValMax-2} \wedge \symUpdateEvent_{\symMaxUpdateValMax-1} \wedge \bigvee_{\symUpdateVal=\symMaxUpdateValMax}^\infty \symUpdateEvent_{\symUpdateVal})
         &
        \scriptstyle\symRegister= 4\symMaxUpdateValMax +2,
        \\
        \scriptstyle\symProbability(\symUpdateEvent_{\symMaxUpdateValMax-2} \wedge \symUpdateEvent_{\symMaxUpdateValMax-1} \wedge \bigvee_{\symUpdateVal=\symMaxUpdateValMax}^\infty \symUpdateEvent_{\symUpdateVal})
         &
        \scriptstyle\symRegister= 4\symMaxUpdateValMax +3,
        \\
        \scriptstyle 0,
         &
        \scriptstyle\text{else}.
      \end{cases}
    \end{equation*}
    Using the identity $\symProbability(\symEventX\vert\symEventY) = \symProbability(\symEventX\wedge\symEventY)/\symProbability(\symEventY)$ and considering all cases of both \acp{PMF} gives for the \ac{PMF} of $\symRegisterSimple$ conditioned on $\symRegister$
    \begin{equation*}
      \scriptstyle\symDensityRegisterSimple(\symRegisterSimple\vert \symRegister,\symCardinality)=
      \begin{cases}
        \scriptstyle\symProbability(\overline\symUpdateEvent_{\symMaxUpdateVal-2} \wedge \overline\symUpdateEvent_{\symMaxUpdateVal-1} \wedge \symUpdateEvent_\symMaxUpdateVal\wedge \bigwedge_{\symUpdateVal=\symMaxUpdateVal+1}^0 \overline\symUpdateEvent_{\symUpdateVal})
         &
        \scriptstyle\symRegister=0,\ \symRegisterSimple= 4\symMaxUpdateVal,\ \symMaxUpdateVal\leq 0,
        \\
        \scriptstyle\symProbability(\symUpdateEvent_{\symMaxUpdateVal-2} \wedge \overline\symUpdateEvent_{\symMaxUpdateVal-1} \wedge \symUpdateEvent_\symMaxUpdateVal\wedge \bigwedge_{\symUpdateVal=\symMaxUpdateVal+1}^0 \overline\symUpdateEvent_{\symUpdateVal})
         &
        \scriptstyle\symRegister=0,\ \symRegisterSimple= 4\symMaxUpdateVal+1,\ \symMaxUpdateVal\leq 0,
        \\
        \scriptstyle\symProbability(\overline\symUpdateEvent_{\symMaxUpdateVal-2} \wedge\symUpdateEvent_{\symMaxUpdateVal-1} \wedge \symUpdateEvent_\symMaxUpdateVal\wedge \bigwedge_{\symUpdateVal=\symMaxUpdateVal+1}^0 \overline\symUpdateEvent_{\symUpdateVal})
         &
        \scriptstyle\symRegister=0,\ \symRegisterSimple= 4\symMaxUpdateVal+2,\ \symMaxUpdateVal\leq 0,
        \\
        \scriptstyle\symProbability(\symUpdateEvent_{\symMaxUpdateVal-2} \wedge\symUpdateEvent_{\symMaxUpdateVal-1} \wedge \symUpdateEvent_\symMaxUpdateVal\wedge \bigwedge_{\symUpdateVal=\symMaxUpdateVal+1}^0 \overline\symUpdateEvent_{\symUpdateVal})
         &
        \scriptstyle\symRegister=0,\ \symRegisterSimple= 4\symMaxUpdateVal+3,\ \symMaxUpdateVal\leq 0,
        \\
        \scriptstyle\symProbability(\overline\symUpdateEvent_{-1}\wedge\overline\symUpdateEvent_{0})
         &
        \scriptstyle\symRegister=4,\ \symRegisterSimple= 4,
        \\
        \scriptstyle\symProbability(\symUpdateEvent_{-1}\wedge\overline\symUpdateEvent_{0})
         &
        \scriptstyle\symRegister=4,\ \symRegisterSimple= 5,
        \\
        \scriptstyle\symProbability(\overline\symUpdateEvent_{-1}\wedge\symUpdateEvent_{0})
         &
        \scriptstyle\symRegister=4,\ \symRegisterSimple= 6,
        \\
        \scriptstyle\symProbability(\symUpdateEvent_{-1}\wedge\symUpdateEvent_{0})
         &
        \scriptstyle\symRegister=4,\ \symRegisterSimple= 7,
        \\
        \scriptstyle\symProbability(\overline\symUpdateEvent_{0})
         &
        \scriptstyle\symRegister=8,\ \symRegisterSimple= 8,
        \\
        \scriptstyle\symProbability(\symUpdateEvent_{0})
         &
        \scriptstyle\symRegister=8,\ \symRegisterSimple= 9,
        \\
        \scriptstyle\symProbability(\overline\symUpdateEvent_{0})
         &
        \scriptstyle\symRegister=10,\ \symRegisterSimple= 10,
        \\
        \scriptstyle\symProbability(\symUpdateEvent_{0})
         &
        \scriptstyle\symRegister=10,\ \symRegisterSimple= 11,
        \\
        \scriptstyle 1
         &
        \scriptstyle 12 \leq \symRegister < 4\symMaxUpdateValMax,\ \symRegisterSimple=\symRegister,
        \\
        \scriptstyle\frac{
          \symProbability(\symUpdateEvent_{\symMaxUpdateValMax}\wedge \bigwedge_{\symUpdateVal=\symMaxUpdateValMax+1}^\infty\overline\symUpdateEvent_{\symUpdateVal})}
        {\symProbability(\bigvee_{\symUpdateVal=\symMaxUpdateValMax}^\infty\symUpdateEvent_{\symUpdateVal})}
         &
        \scriptstyle\symRegister = 4\symMaxUpdateValMax,\ \symRegisterSimple=4\symMaxUpdateValMax,
        \\
        \scriptstyle\frac{
          \symProbability(\symUpdateEvent_{\symMaxUpdateValMax}\wedge \bigwedge_{\symUpdateVal=\symMaxUpdateValMax+1}^\infty\overline\symUpdateEvent_{\symUpdateVal})}
        {\symProbability(\bigvee_{\symUpdateVal=\symMaxUpdateValMax}^\infty\symUpdateEvent_{\symUpdateVal})}
         &
        \scriptstyle\symRegister = 4\symMaxUpdateValMax+1,\ \symRegisterSimple=4\symMaxUpdateValMax+1,
        \\
        \scriptstyle\frac{
          \symProbability(\symUpdateEvent_{\symMaxUpdateValMax}\wedge \bigwedge_{\symUpdateVal=\symMaxUpdateValMax+1}^\infty\overline\symUpdateEvent_{\symUpdateVal})}
        {\symProbability(\bigvee_{\symUpdateVal=\symMaxUpdateValMax}^\infty\symUpdateEvent_{\symUpdateVal})}
         &
        \scriptstyle\symRegister = 4\symMaxUpdateValMax+2,\ \symRegisterSimple=4\symMaxUpdateValMax+2,
        \\
        \scriptstyle\frac{
          \symProbability(\symUpdateEvent_{\symMaxUpdateValMax}\wedge \bigwedge_{\symUpdateVal=\symMaxUpdateValMax+1}^\infty\overline\symUpdateEvent_{\symUpdateVal})}
        {\symProbability(\bigvee_{\symUpdateVal=\symMaxUpdateValMax}^\infty\symUpdateEvent_{\symUpdateVal})}
         &
        \scriptstyle\symRegister = 4\symMaxUpdateValMax+3,\ \symRegisterSimple=4\symMaxUpdateValMax+3,
        \\
        \scriptstyle\frac{
          \symProbability(\overline\symUpdateEvent_{\symMaxUpdateValMax}\wedge \symUpdateEvent_{\symMaxUpdateValMax+1}\wedge \bigwedge_{\symUpdateVal=\symMaxUpdateValMax+2}^\infty\overline\symUpdateEvent_{\symUpdateVal})}
        {\symProbability(\bigvee_{\symUpdateVal=\symMaxUpdateValMax}^\infty\symUpdateEvent_{\symUpdateVal})}
         &
        \scriptstyle\symRegister = 4\symMaxUpdateValMax,\ \symRegisterSimple=4\symMaxUpdateValMax+4,
        \\
        \scriptstyle\frac{
          \symProbability(\overline\symUpdateEvent_{\symMaxUpdateValMax}\wedge \symUpdateEvent_{\symMaxUpdateValMax+1}\wedge \bigwedge_{\symUpdateVal=\symMaxUpdateValMax+2}^\infty\overline\symUpdateEvent_{\symUpdateVal})}
        {\symProbability(\bigvee_{\symUpdateVal=\symMaxUpdateValMax}^\infty\symUpdateEvent_{\symUpdateVal})}
         &
        \scriptstyle\symRegister = 4\symMaxUpdateValMax+1,\ \symRegisterSimple=4\symMaxUpdateValMax+4,
        \\
        \scriptstyle\frac{
          \symProbability(\overline\symUpdateEvent_{\symMaxUpdateValMax}\wedge \symUpdateEvent_{\symMaxUpdateValMax+1}\wedge \bigwedge_{\symUpdateVal=\symMaxUpdateValMax+2}^\infty\overline\symUpdateEvent_{\symUpdateVal})}
        {\symProbability(\bigvee_{\symUpdateVal=\symMaxUpdateValMax}^\infty\symUpdateEvent_{\symUpdateVal})}
         &
        \scriptstyle\symRegister = 4\symMaxUpdateValMax+2,\ \symRegisterSimple=4\symMaxUpdateValMax+5,
        \\
        \scriptstyle\frac{
          \symProbability(\overline\symUpdateEvent_{\symMaxUpdateValMax}\wedge \symUpdateEvent_{\symMaxUpdateValMax+1}\wedge \bigwedge_{\symUpdateVal=\symMaxUpdateValMax+2}^\infty\overline\symUpdateEvent_{\symUpdateVal})}
        {\symProbability(\bigvee_{\symUpdateVal=\symMaxUpdateValMax}^\infty\symUpdateEvent_{\symUpdateVal})}
         &
        \scriptstyle\symRegister = 4\symMaxUpdateValMax+3,\ \symRegisterSimple=4\symMaxUpdateValMax+5,
        \\
        \scriptstyle\frac{
          \symProbability(\symUpdateEvent_{\symMaxUpdateValMax}\wedge \symUpdateEvent_{\symMaxUpdateValMax+1}\wedge \bigwedge_{\symUpdateVal=\symMaxUpdateValMax+2}^\infty\overline\symUpdateEvent_{\symUpdateVal})}
        {\symProbability(\bigvee_{\symUpdateVal=\symMaxUpdateValMax}^\infty\symUpdateEvent_{\symUpdateVal})}
         &
        \scriptstyle\symRegister = 4\symMaxUpdateValMax,\ \symRegisterSimple=4\symMaxUpdateValMax+6,
        \\
        \scriptstyle\frac{
          \symProbability(\symUpdateEvent_{\symMaxUpdateValMax}\wedge \symUpdateEvent_{\symMaxUpdateValMax+1}\wedge \bigwedge_{\symUpdateVal=\symMaxUpdateValMax+2}^\infty\overline\symUpdateEvent_{\symUpdateVal})}
        {\symProbability(\bigvee_{\symUpdateVal=\symMaxUpdateValMax}^\infty\symUpdateEvent_{\symUpdateVal})}
         &
        \scriptstyle\symRegister = 4\symMaxUpdateValMax+1,\ \symRegisterSimple=4\symMaxUpdateValMax+6,
        \\
        \scriptstyle\frac{
          \symProbability(\symUpdateEvent_{\symMaxUpdateValMax}\wedge \symUpdateEvent_{\symMaxUpdateValMax+1}\wedge \bigwedge_{\symUpdateVal=\symMaxUpdateValMax+2}^\infty\overline\symUpdateEvent_{\symUpdateVal})}
        {\symProbability(\bigvee_{\symUpdateVal=\symMaxUpdateValMax}^\infty\symUpdateEvent_{\symUpdateVal})}
         &
        \scriptstyle\symRegister = 4\symMaxUpdateValMax+2,\ \symRegisterSimple=4\symMaxUpdateValMax+7,
        \\
        \scriptstyle\frac{
          \symProbability(\symUpdateEvent_{\symMaxUpdateValMax}\wedge \symUpdateEvent_{\symMaxUpdateValMax+1}\wedge \bigwedge_{\symUpdateVal=\symMaxUpdateValMax+2}^\infty\overline\symUpdateEvent_{\symUpdateVal})}
        {\symProbability(\bigvee_{\symUpdateVal=\symMaxUpdateValMax}^\infty\symUpdateEvent_{\symUpdateVal})}
         &
        \scriptstyle\symRegister = 4\symMaxUpdateValMax+3,\ \symRegisterSimple=4\symMaxUpdateValMax+7,
        \\
        \scriptstyle\frac{
          \symProbability(\overline\symUpdateEvent_{\symMaxUpdateVal-2}\wedge\overline\symUpdateEvent_{\symMaxUpdateVal-1}\wedge \symUpdateEvent_{\symMaxUpdateVal}\wedge \bigwedge_{\symUpdateVal=\symMaxUpdateVal+1}^\infty\overline\symUpdateEvent_{\symUpdateVal})}
        {\symProbability(\bigvee_{\symUpdateVal=\symMaxUpdateValMax}^\infty\symUpdateEvent_{\symUpdateVal})}
         &
        \scriptstyle\symRegister \geq 4\symMaxUpdateValMax,\ \symRegisterSimple=4\symMaxUpdateVal,\ \symMaxUpdateVal\geq \symMaxUpdateValMax+2,
        \\
        \scriptstyle\frac{
          \symProbability(\symUpdateEvent_{\symMaxUpdateVal-2}\wedge\overline\symUpdateEvent_{\symMaxUpdateVal-1}\wedge \symUpdateEvent_{\symMaxUpdateVal}\wedge \bigwedge_{\symUpdateVal=\symMaxUpdateVal+1}^\infty\overline\symUpdateEvent_{\symUpdateVal})}
        {\symProbability(\bigvee_{\symUpdateVal=\symMaxUpdateValMax}^\infty\symUpdateEvent_{\symUpdateVal})}
         &
        \scriptstyle\symRegister \geq 4\symMaxUpdateValMax,\ \symRegisterSimple=4\symMaxUpdateVal+1,\ \symMaxUpdateVal\geq \symMaxUpdateValMax+2,
        \\
        \scriptstyle\frac{
          \symProbability(\overline\symUpdateEvent_{\symMaxUpdateVal-2}\wedge\symUpdateEvent_{\symMaxUpdateVal-1}\wedge \symUpdateEvent_{\symMaxUpdateVal}\wedge \bigwedge_{\symUpdateVal=\symMaxUpdateVal+1}^\infty\overline\symUpdateEvent_{\symUpdateVal})}
        {\symProbability(\bigvee_{\symUpdateVal=\symMaxUpdateValMax}^\infty\symUpdateEvent_{\symUpdateVal})}
         &
        \scriptstyle\symRegister \geq 4\symMaxUpdateValMax,\ \symRegisterSimple=4\symMaxUpdateVal+2,\ \symMaxUpdateVal\geq \symMaxUpdateValMax+2,
        \\
        \scriptstyle\frac{
          \symProbability(\symUpdateEvent_{\symMaxUpdateVal-2}\wedge\symUpdateEvent_{\symMaxUpdateVal-1}\wedge \symUpdateEvent_{\symMaxUpdateVal}\wedge \bigwedge_{\symUpdateVal=\symMaxUpdateVal+1}^\infty\overline\symUpdateEvent_{\symUpdateVal})}
        {\symProbability(\bigvee_{\symUpdateVal=\symMaxUpdateValMax}^\infty\symUpdateEvent_{\symUpdateVal})}
         &
        \scriptstyle\symRegister \geq 4\symMaxUpdateValMax,\ \symRegisterSimple=4\symMaxUpdateVal+3,\ \symMaxUpdateVal\geq \symMaxUpdateValMax+2,
        \\
        \scriptstyle 0
         &
        \scriptstyle\text{else}.
      \end{cases}
    \end{equation*}
    The identities
    \begin{align*}
       & \textstyle\symProbability(\symUpdateEvent_\symUpdateVal) =
      1 - \symZ_\symUpdateVal,
      \\
       & \textstyle\symProbability(\overline\symUpdateEvent_\symUpdateVal) = \symZ_\symUpdateVal,
      \\
       & \textstyle\symProbability(\bigwedge_{\symUpdateVal=\symMaxUpdateVal+1}^0 \overline\symUpdateEvent_{\symUpdateVal})=\prod_{\symUpdateVal=\symMaxUpdateVal+1}^0 \symZ_\symUpdateVal=\symZ_\symMaxUpdateVal/\symZ_0 \quad \text{for $\symMaxUpdateVal\leq 0$},
      \\
       & \textstyle\symProbability(\bigwedge_{\symUpdateVal=\symMaxUpdateValMax+2}^\infty\overline\symUpdateEvent_{\symUpdateVal}) = \prod_{\symUpdateVal=\symMaxUpdateValMax+2}^\infty \symZ_\symUpdateVal = \symZ_{\symMaxUpdateValMax+1},
      \\
       & \textstyle\symProbability(\bigwedge_{\symUpdateVal=\symMaxUpdateVal+1}^\infty\overline\symUpdateEvent_{\symUpdateVal}) = \prod_{\symUpdateVal=\symMaxUpdateVal+1}^\infty \symZ_\symUpdateVal = \symZ_\symMaxUpdateVal,
      \\
       & \textstyle\symProbability(\bigvee_{\symUpdateVal=\symMaxUpdateValMax}^\infty\symUpdateEvent_{\symUpdateVal}) = 1 - \symProbability(\bigwedge_{\symUpdateVal=\symMaxUpdateValMax}^\infty\overline\symUpdateEvent_{\symUpdateVal}) = 1 - \prod_{\symUpdateVal=\symMaxUpdateValMax}^\infty \symZ_\symUpdateVal = 1 - \symZ_{\symMaxUpdateValMax-1}
    \end{align*}
    finally lead to \eqref{equ:cond_pmf}.
  \end{proof}

  \begin{lemma}
    \label{lem:corr_reg_contrib}
    The corrected register contributions $\symRegContribCorr(\symRegister)$ defined by
    \begin{equation*}
      \symRegContribCorr(\symRegister) = \sum_{\symRegisterSimple=-\infty}^\infty \symDensityRegisterSimple(\symRegisterSimple\vert \symRegister,\symCardinality = \symCardinalityEstimatorAlt)\symRegContrib(\symRegisterSimple)
    \end{equation*}
    for
    $\symRegContrib(\symRegister) := 2^{-\symGRA\lfloor\symRegister/4\rfloor} \symContributionCoefficient_{\symRegister \bmod 4}$
    and $\symDensityRegisterSimple$ as given by \eqref{equ:cond_pmf}
    results in \eqref{equ:corrected_contributions}.
  \end{lemma}

  \begin{proof}
    We use $\symZEstimate_0$ and $\symZEstimate_\symMaxUpdateValMax$ to denote $\symZEstimate_0:=e^{-\frac{\symCardinalityEstimatorAlt}{\symNumReg}}$ and $\symZEstimate_\symMaxUpdateValMax:=e^{-\frac{\symCardinalityEstimatorAlt}{\symNumReg 2^\symMaxUpdateValMax}}$, respectively. The cases $\symRegister\in\lbrace 0,4,8,10\rbrace$ are shown by
    \begin{align*}
       &
      \symRegContribCorr(0)
      =
      \sum_{\symRegisterSimple=-\infty}^\infty \symDensityRegisterSimple(\symRegisterSimple \vert 0,\symCardinality = \symCardinalityEstimatorAlt)\symRegContrib(\symRegisterSimple)
      \\
       & =
      \frac{1}{\symZ_{0}}
      \sum_{\symMaxUpdateVal=-\infty}^0
      \symZEstimate_0^{2^{-\symMaxUpdateVal}}
      (1-\symZEstimate_0^{2^{-\symMaxUpdateVal}})
      2^{-\symGRA \symMaxUpdateVal}
      \left(
      \begin{array}{@{}l}
          \symContributionCoefficient_0 \symZEstimate_0^{2^{-\symMaxUpdateVal+1}}\symZEstimate_0^{2^{-\symMaxUpdateVal+2}}
          \\
          +
          \symContributionCoefficient_1 \symZEstimate_0^{2^{-\symMaxUpdateVal+1}}(1-\symZEstimate_0^{2^{-\symMaxUpdateVal+2}})
          \\
          +
          \symContributionCoefficient_2 (1-\symZEstimate_0^{2^{-\symMaxUpdateVal+1}})\symZEstimate_0^{2^{-\symMaxUpdateVal+2}}
          \\
          +
          \symContributionCoefficient_3 (1-\symZEstimate_0^{2^{-\symMaxUpdateVal+1}})(1-\symZEstimate_0^{2^{-\symMaxUpdateVal+2}})
        \end{array}
      \right)
      \\
       & =
      \frac{1}{\symZ_{0}}
      \sum_{\symMaxUpdateVal=0}^\infty
      2^{\symGRA \symMaxUpdateVal}
      (\symZEstimate_0^{2^{\symMaxUpdateVal}}-\symZEstimate_0^{2^{\symMaxUpdateVal+1}})
      \symCubicContribFunc(\symZEstimate_0^{2^{\symMaxUpdateVal+1}})
      =
      \symSmallRangeCorrFunc(\symZEstimate_0),
      \\
       &
      \symRegContribCorr(4) =
      \sum_{\symRegisterSimple=-\infty}^\infty \symDensityRegisterSimple(\symRegisterSimple \vert 4,\symCardinality = \symCardinalityEstimatorAlt)\symRegContrib(\symRegisterSimple)
      \\
       & =
      2^{-\symGRA}(
      \symContributionCoefficient_0 \symZEstimate_0^3
      +
      \symContributionCoefficient_1 \symZEstimate_0(1-\symZEstimate_0^2)
      +
      \symContributionCoefficient_2 (1- \symZEstimate_0)\symZEstimate_0^2
      +
      \symContributionCoefficient_3 (1- \symZEstimate_0)(1-\symZEstimate_0^2))
      \\
       & =
      2^{-\symGRA}\symCubicContribFunc(\symZEstimate_0),
      \\
       &
      \symRegContribCorr(8) =
      \sum_{\symRegisterSimple=-\infty}^\infty \symDensityRegisterSimple(\symRegisterSimple \vert 8,\symCardinality = \symCardinalityEstimatorAlt)\symRegContrib(\symRegisterSimple)
      \\
       & =
      2^{-2\symGRA}(
      \symContributionCoefficient_0 \symZEstimate_0 + \symContributionCoefficient_1 (1-\symZEstimate_0))
      = 4^{-\symGRA}(\symZEstimate_0(\symContributionCoefficient_0-\symContributionCoefficient_1) + \symContributionCoefficient_1),
      \\
       &
      \symRegContribCorr(10) =
      \sum_{\symRegisterSimple=-\infty}^\infty \symDensityRegisterSimple(\symRegisterSimple \vert 10,\symCardinality = \symCardinalityEstimatorAlt)\symRegContrib(\symRegisterSimple)
      \\
       & =
      2^{-2\symGRA}(
      \symContributionCoefficient_2 \symZEstimate_0 + \symContributionCoefficient_3 (1-\symZEstimate_0))
      = 4^{-\symGRA}(\symZEstimate_0(\symContributionCoefficient_2-\symContributionCoefficient_3) + \symContributionCoefficient_3).
    \end{align*}
    The case $12 \leq \symRegister < 4\symMaxUpdateValMax$ is trivial. The remaining case $\symRegister = 4\symMaxUpdateValMax + \langle \symIndexBit_1\symIndexBit_2\rangle_2$ is proven by
    \begin{align*}
       &
      \symRegContribCorr(4\symMaxUpdateValMax+ \langle \symIndexBit_1\symIndexBit_2\rangle_2)
      =
      \sum_{\symRegisterSimple=-\infty}^\infty \symDensityRegisterSimple(\symRegisterSimple \vert 4\symMaxUpdateValMax+ \langle \symIndexBit_1\symIndexBit_2\rangle_2,\symCardinality = \symCardinalityEstimatorAlt)\symRegContrib(\symRegisterSimple)
      \\
       & =
      2^{-\symGRA \symMaxUpdateValMax}\frac{\symZEstimate_\symMaxUpdateValMax(1-\symZEstimate_\symMaxUpdateValMax)}{1-\symZEstimate_{\symMaxUpdateValMax}^2}\symContributionCoefficient_{2\symIndexBit_1+\symIndexBit_2}
      \\
       & \quad + 2^{-\symGRA(\symMaxUpdateValMax+1)} \frac{\sqrt{\symZEstimate_{\symMaxUpdateValMax}}(1-\sqrt{\symZEstimate_{\symMaxUpdateValMax}})\symZEstimate_{\symMaxUpdateValMax}}{1-\symZEstimate_{\symMaxUpdateValMax}^2}\symContributionCoefficient_{\symIndexBit_1}
      \\
       &
      \quad + 2^{-\symGRA(\symMaxUpdateValMax+1)} \frac{\sqrt{\symZEstimate_{\symMaxUpdateValMax}}(1-\sqrt{\symZEstimate_{\symMaxUpdateValMax}})(1-\symZEstimate_{\symMaxUpdateValMax})}{1-\symZEstimate_{\symMaxUpdateValMax}^2}\symContributionCoefficient_{2+\symIndexBit_1}
      \\
       &
      \quad +
      \sum_{\symMaxUpdateVal=\symMaxUpdateValMax + 2}^\infty
      \frac{
      \symZEstimate_{\symMaxUpdateValMax}^{2^{\symMaxUpdateValMax-\symMaxUpdateVal}}
      (1-\symZEstimate_{\symMaxUpdateValMax}^{2^{\symMaxUpdateValMax-\symMaxUpdateVal}})
      }{2^{\symGRA\symMaxUpdateVal}(
      1-\symZEstimate_{\symMaxUpdateValMax}^2)}
      \left(
      \begin{array}{@{}l}
          \symContributionCoefficient_0 \symZEstimate_{\symMaxUpdateValMax}^{2^{\symMaxUpdateValMax-\symMaxUpdateVal+1}}\symZEstimate_{\symMaxUpdateValMax}^{2^{\symMaxUpdateValMax-\symMaxUpdateVal+2}}
          \\
          +
          \symContributionCoefficient_1 \symZEstimate_{\symMaxUpdateValMax}^{2^{\symMaxUpdateValMax-\symMaxUpdateVal+1}}(1-\symZEstimate_{\symMaxUpdateValMax}^{2^{\symMaxUpdateValMax-\symMaxUpdateVal+2}})
          \\
          +
          \symContributionCoefficient_2 (1-\symZEstimate_{\symMaxUpdateValMax}^{2^{\symMaxUpdateValMax-\symMaxUpdateVal+1}})\symZEstimate_{\symMaxUpdateValMax}^{2^{\symMaxUpdateValMax-\symMaxUpdateVal+2}}
          \\
          +
          \symContributionCoefficient_3 (1-\symZEstimate_{\symMaxUpdateValMax}^{2^{\symMaxUpdateValMax-\symMaxUpdateVal+1}})(1-\symZEstimate_{\symMaxUpdateValMax}^{2^{\symMaxUpdateValMax-\symMaxUpdateVal+2}})
        \end{array}
      \right)
      \\
       & =
      2^{-\symGRA \symMaxUpdateValMax}\frac{\symZEstimate_\symMaxUpdateValMax(1+\sqrt{\symZEstimate_\symMaxUpdateValMax})}{(1+\sqrt{\symZEstimate_{\symMaxUpdateValMax}})(1+\symZEstimate_{\symMaxUpdateValMax})}\symContributionCoefficient_{2\symIndexBit_1+\symIndexBit_2}
      \\
       & \quad + 2^{-\symGRA(\symMaxUpdateValMax+1)} \frac{\sqrt{\symZEstimate_{\symMaxUpdateValMax}}\symZEstimate_{\symMaxUpdateValMax}}{(1+\sqrt{\symZEstimate_{\symMaxUpdateValMax}})(1+\symZEstimate_{\symMaxUpdateValMax})}\symContributionCoefficient_{\symIndexBit_1}
      \\
       &
      \quad + 2^{-\symGRA(\symMaxUpdateValMax+1)} \frac{\sqrt{\symZEstimate_{\symMaxUpdateValMax}}(1-\symZEstimate_{\symMaxUpdateValMax})}{(1+\sqrt{\symZEstimate_{\symMaxUpdateValMax}})(1+\symZEstimate_{\symMaxUpdateValMax})}\symContributionCoefficient_{2+\symIndexBit_1}
      \\
       &
      \quad +
      \sum_{\symMaxUpdateVal=0}^\infty
      \textstyle\frac{
      (\sqrt{\symZEstimate_{\symMaxUpdateValMax}}^{2^{-\symMaxUpdateVal-1}} - \sqrt{\symZEstimate_{\symMaxUpdateValMax}}^{2^{-\symMaxUpdateVal}})
      }{2^{\symGRA(\symMaxUpdateValMax + 2 + \symMaxUpdateVal)}(
      1-\symZEstimate_{\symMaxUpdateValMax}^2)}
      \left(
      \begin{array}{@{}l}
          \symContributionCoefficient_0 (\sqrt{\symZEstimate_{\symMaxUpdateValMax}}^{2^{-\symMaxUpdateVal}})^3
          \\
          +
          \symContributionCoefficient_1 (\sqrt{\symZEstimate_{\symMaxUpdateValMax}}^{2^{-\symMaxUpdateVal}})(1-(\sqrt{\symZEstimate_{\symMaxUpdateValMax}}^{2^{-\symMaxUpdateVal}})^2)
          \\
          +
          \symContributionCoefficient_2 (1-\sqrt{\symZEstimate_{\symMaxUpdateValMax}}^{2^{-\symMaxUpdateVal}})(\sqrt{\symZEstimate_{\symMaxUpdateValMax}}^{2^{-\symMaxUpdateVal}})^2
          \\
          +
          \symContributionCoefficient_3 (1-\sqrt{\symZEstimate_{\symMaxUpdateValMax}}^{2^{-\symMaxUpdateVal}})(1-(\sqrt{\symZEstimate_{\symMaxUpdateValMax}}^{2^{-\symMaxUpdateVal}})^2)
        \end{array}
      \right)
      \\
       & =
      \textstyle
      \frac{\symZEstimate_\symMaxUpdateValMax(1+\sqrt{\symZEstimate_\symMaxUpdateValMax}) \symContributionCoefficient_{2\symIndexBit_1+\symIndexBit_2}
        + 2^{-\symGRA}\sqrt{\symZEstimate_{\symMaxUpdateValMax}}
        (\symZEstimate_{\symMaxUpdateValMax} (\symContributionCoefficient_{\symIndexBit_1} -\symContributionCoefficient_{2+\symIndexBit_1})
        + \symContributionCoefficient_{2+\symIndexBit_1})
        + \symLargeRangeCorrFunc(\sqrt{\symZEstimate_{\symMaxUpdateValMax}})
      }{2^{\symGRA\symMaxUpdateValMax}(1+\sqrt{\symZEstimate_{\symMaxUpdateValMax}})(1+\symZEstimate_{\symMaxUpdateValMax})}.
    \end{align*}
  \end{proof}

  \begin{lemma}
    \label{lem:small_range_estimator}
    Assume $\symNumReg$ independent and identically distributed values $\symRegister_\symRegAddr$ which take the values from $\lbrace 0, 4, 8, 10 \rbrace$ with probabilities
    (cf. \eqref{equ:register_pmf} for $\symBase=2$ and $\symNumExtraBits=2$)
    \begin{align*}
       & \symProbability(\symRegister_\symRegAddr = 0) = \symZ_0 = \symZ_2^4,
      \\
       & \symProbability(\symRegister_\symRegAddr = 4) = \symZ_1 (1-\symZ_1) = \symZ_2^2 (1-\symZ_2)(1+\symZ_2),
      \\
       & \symProbability(\symRegister_\symRegAddr = 8) = \symZ_2(1 - \symZ_2)\symZ_1 = \symZ_2^3(1 - \symZ_2),
      \\
       & \symProbability(\symRegister_\symRegAddr = 10) =\symZ_2(1 - \symZ_2)(1-\symZ_1)=\symZ_2(1 - \symZ_2)^2(1+\symZ_2),
    \end{align*}
    respectively. $\symZ_\symMaxUpdateVal$ is defined as $\symZ_\symMaxUpdateVal := e^{-\frac{\symCardinality}{\symNumReg 2^{\symMaxUpdateVal}}}$. If the corresponding observed frequencies are $\symCount_0, \symCount_4, \symCount_8, \symCount_{10}$, the \ac{ML} estimator
    for $\symZ_0 = e^{-\frac{\symCardinality}{\symNumReg}}$ is given by
    \begin{equation*}
      \symZEstimate_0 = e^{-\frac{\symCardinalityEstimatorLow }{\symNumReg}} =
      \textstyle
      \left(
      \frac{
        \sqrt{\symBeta^2 + 4\symAlpha \symGamma} - \symBeta
      }{
        2 \symAlpha
      }
      \right)^{\!4}
    \end{equation*}
    with
    \begin{align*}
      \symAlpha & := \symNumReg + 3(\symCount_0 + \symCount_4 + \symCount_8 + \symCount_{10}), \\
      \symBeta  & := \symNumReg - \symCount_0 - \symCount_4,                                   \\
      \symGamma & := 4\symCount_0 + 2\symCount_4 + 3\symCount_8 + \symCount_{10}.
    \end{align*}
  \end{lemma}
  \begin{proof}
    Together with $\symProbability(\symRegister_\symRegAddr \notin \lbrace0,4,8,10\rbrace) = 1 - \symProbability(\symRegister_\symRegAddr = 0) -\symProbability(\symRegister_\symRegAddr = 4)-\symProbability(\symRegister_\symRegAddr = 8)-\symProbability(\symRegister_\symRegAddr = 10) = 1 - \symZ_2$ we can write the likelihood function $\symLikelihood$ of the $\symNumReg$ registers as
    \begin{align*}
       & \symLikelihood
      =
      \symProbability(\symRegister_\symRegAddr = 0)^{\symCount_{0}}
      \symProbability(\symRegister_\symRegAddr = 4)^{\symCount_{4}}
      \symProbability(\symRegister_\symRegAddr = 8)^{\symCount_{8}}
      \symProbability(\symRegister_\symRegAddr = 10)^{\symCount_{10}}
      \cdot
      \\
       & \quad\cdot
      \symProbability(\symRegister_\symRegAddr \notin \lbrace0,4,8,10\rbrace)^{\symNumReg-\symCount_0-\symCount_4-\symCount_8-\symCount_{10}}
      \\
       & =
      \symZ_2^{4\symCount_0 + 2\symCount_4 + 3\symCount_8+\symCount_{10}}
      (1-\symZ_2)^{\symCount_{10} +\symNumReg-\symCount_0 }
      (1+\symZ_2)^{\symCount_{4} +\symCount_{10}}
    \end{align*}
    and the log-likelihood function as
    \begin{multline*}
      \ln \symLikelihood = (4\symCount_0 + 2\symCount_4 + 3\symCount_8+\symCount_{10}) \ln(\symZ_2) + (\symCount_{10} +\symNumReg-\symCount_0 ) \ln(1-\symZ_2)
      \\
      + (\symCount_{4} +\symCount_{10})\ln(1+\symZ_2).
    \end{multline*}
    The \ac{ML} estimate $\symZEstimate_2$ for $\symZ_2$ can be found by setting the first derivative
    \begin{multline*}
      (\ln \symLikelihood)' = \frac{4\symCount_0 + 2\symCount_4 + 3\symCount_8+\symCount_{10}}{\symZ_2} + \frac{\symCount_0-\symCount_{10} -\symNumReg}{1-\symZ_2} + \frac{\symCount_{4} +\symCount_{10}}{1+\symZ_2}
    \end{multline*}
    equal to 0
    \begin{align*}
       & \frac{4\symCount_0 + 2\symCount_4 + 3\symCount_8+\symCount_{10}}{\symZEstimate_2} + \frac{\symCount_0-\symCount_{10} -\symNumReg}{1-\symZEstimate_2} + \frac{\symCount_{4} +\symCount_{10}}{1+\symZEstimate_2} = 0
      \\
       & \Rightarrow \symAlpha \symZEstimate_2^2 +\symBeta \symZEstimate_2 -\symGamma = 0
      \Rightarrow \symZEstimate_2 = \textstyle
      \frac{
        \sqrt{\symBeta^2 + 4\symAlpha \symGamma} - \symBeta
      }{
        2 \symAlpha
      }.
    \end{align*}
    Using the invariance property of \ac{ML} estimators we can estimate $\symZ_0 = \symZ_2^4$ by $\symZEstimate_0 = \symZEstimate_2^4$, which completes the proof.
  \end{proof}

  \begin{lemma}
    \label{lem:large_range_estimator}
    Assume $\symNumReg$ independent and identically distributed values $\symRegister_\symRegAddr$ which take the values from $\lbrace 4\symMaxUpdateValMax, 4\symMaxUpdateValMax+1, 4\symMaxUpdateValMax+2, 4\symMaxUpdateValMax+3 \rbrace$ with probabilities (cf. \eqref{equ:register_pmf} for $\symBase=2$ and $\symNumExtraBits=2$)
    \begin{align*}
       & \symProbability(\symRegister_\symRegAddr = 4\symMaxUpdateValMax)= (1-\symZ_{\symMaxUpdateValMax-1})\symZ_{\symMaxUpdateValMax-1}\symZ_{\symMaxUpdateValMax-2}
      =(1-\symZ_{\symMaxUpdateValMax-1})\symZ_{\symMaxUpdateValMax-1}^3,
      \\
       & \symProbability(\symRegister_\symRegAddr = 4\symMaxUpdateValMax+1)= (1-\symZ_{\symMaxUpdateValMax-1})\symZ_{\symMaxUpdateValMax-1}(1-\symZ_{\symMaxUpdateValMax-2})
      \\
       & \quad = (1-\symZ_{\symMaxUpdateValMax-1})^2\symZ_{\symMaxUpdateValMax-1}(1+\symZ_{\symMaxUpdateValMax-1}),
      \\
       & \symProbability(\symRegister_\symRegAddr = 4\symMaxUpdateValMax+2) = (1-\symZ_{\symMaxUpdateValMax-1})^2\symZ_{\symMaxUpdateValMax-2}=(1-\symZ_{\symMaxUpdateValMax-1})^2\symZ_{\symMaxUpdateValMax-1}^2,
      \\
       & \symProbability(\symRegister_\symRegAddr = 4\symMaxUpdateValMax+3) =(1-\symZ_{\symMaxUpdateValMax-1})^2(1-\symZ_{\symMaxUpdateValMax-2})
      \\
       & \quad
      =(1-\symZ_{\symMaxUpdateValMax-1})^3(1+\symZ_{\symMaxUpdateValMax-1}),
    \end{align*}
    respectively. $\symZ_\symMaxUpdateVal$ is defined as $\symZ_\symMaxUpdateVal := e^{-\frac{\symCardinality}{\symNumReg 2^{\symMaxUpdateVal}}}$. If the corresponding observed frequencies are $\symCount_{ 4\symMaxUpdateValMax}, \symCount_{ 4\symMaxUpdateValMax+1}, \symCount_{ 4\symMaxUpdateValMax+2}, \symCount_{ 4\symMaxUpdateValMax+3}$, the \ac{ML} estimator
    for $\symZ_\symMaxUpdateValMax = e^{-\frac{\symCardinality}{\symNumReg 2^\symMaxUpdateValMax}}$ is given by
    \begin{equation*}
      \symZEstimate_\symMaxUpdateValMax = \textstyle
      \sqrt{
        \frac{
          \sqrt{\symBeta^2 + 4\symAlpha \symGamma} - \symBeta
        }{
          2 \symAlpha
        }}
    \end{equation*}
    with $\symAlpha$, $\symBeta$, $\symGamma$
    \begin{align*}
      \symAlpha & := \symNumReg + 3(\symCount_{4\symMaxUpdateValMax} + \symCount_{4\symMaxUpdateValMax+1} + \symCount_{4\symMaxUpdateValMax+2} +\symCount_{4\symMaxUpdateValMax+3}), \\
      \symBeta  & := \symCount_{4\symMaxUpdateValMax} + \symCount_{4\symMaxUpdateValMax+1} + 2\symCount_{4\symMaxUpdateValMax+2} +2\symCount_{4\symMaxUpdateValMax+3},               \\
      \symGamma & := \symNumReg + 2\symCount_{4\symMaxUpdateValMax} + \symCount_{4\symMaxUpdateValMax+2} - \symCount_{4\symMaxUpdateValMax+3}.
    \end{align*}
  \end{lemma}
  \begin{proof}
    Together with $\symProbability(\symRegister_\symRegAddr < 4 \symMaxUpdateValMax) = 1 - \symProbability(\symRegister_\symRegAddr = 4 \symMaxUpdateValMax) -\symProbability(\symRegister_\symRegAddr = 4 \symMaxUpdateValMax+1)-\symProbability(\symRegister_\symRegAddr = 4 \symMaxUpdateValMax+2)-\symProbability(\symRegister_\symRegAddr = 4 \symMaxUpdateValMax+3) = \symZ_{\symMaxUpdateValMax-1}$ we can write the likelihood function $\symLikelihood$ of the $\symNumReg$ registers as
    \begin{align*}
       & \symLikelihood
      =
      \symProbability(\symRegister_\symRegAddr = 4 \symMaxUpdateValMax)^{ \symCount_{ 4\symMaxUpdateValMax}}
      \symProbability(\symRegister_\symRegAddr = 4 \symMaxUpdateValMax+1)^{ \symCount_{ 4\symMaxUpdateValMax+1}}
      \symProbability(\symRegister_\symRegAddr = 4 \symMaxUpdateValMax+2)^{ \symCount_{ 4\symMaxUpdateValMax+2}}
      \cdot
      \\
       & \quad\cdot
      \symProbability(\symRegister_\symRegAddr = 4 \symMaxUpdateValMax+3)^{ \symCount_{ 4\symMaxUpdateValMax+3}}
      \symProbability(\symRegister_\symRegAddr < 4 \symMaxUpdateValMax)^{\symNumReg-\symCount_{ 4\symMaxUpdateValMax}-\symCount_{ 4\symMaxUpdateValMax+1}-\symCount_{ 4\symMaxUpdateValMax+2}-\symCount_{ 4\symMaxUpdateValMax+3}}
      \\
       & =
      \symZ_{\symMaxUpdateValMax-1}^{\symNumReg + 2\symCount_{4\symMaxUpdateValMax} + \symCount_{4\symMaxUpdateValMax+2}-\symCount_{4\symMaxUpdateValMax+3}}
      (1-\symZ_{\symMaxUpdateValMax-1})^{\symCount_{4\symMaxUpdateValMax}+2\symCount_{4\symMaxUpdateValMax+1}+2\symCount_{4\symMaxUpdateValMax+2}+3\symCount_{4\symMaxUpdateValMax+3}}
      \cdot
      \\
       & \quad \cdot
      (1+\symZ_{\symMaxUpdateValMax-1})^{\symCount_{4\symMaxUpdateValMax+1} +\symCount_{4\symMaxUpdateValMax+3}}
    \end{align*}
    and the log-likelihood function as
    \begin{multline*}
      \ln \symLikelihood = (\symNumReg + 2\symCount_{4\symMaxUpdateValMax} + \symCount_{4\symMaxUpdateValMax+2}-\symCount_{4\symMaxUpdateValMax+3}) \ln(\symZ_{\symMaxUpdateValMax-1})
      \\
      + (\symCount_{4\symMaxUpdateValMax}+2\symCount_{4\symMaxUpdateValMax+1}+2\symCount_{4\symMaxUpdateValMax+2}+3\symCount_{4\symMaxUpdateValMax+3}) \ln(1-\symZ_{\symMaxUpdateValMax-1})
      \\
      + (\symCount_{4\symMaxUpdateValMax+1} +\symCount_{4\symMaxUpdateValMax+3})\ln(1+\symZ_{\symMaxUpdateValMax-1}).
    \end{multline*}
    The \ac{ML} estimate $\symZEstimate_{\symMaxUpdateValMax-1}$ for $\symZ_{\symMaxUpdateValMax-1}$ can be found by setting the first derivative
    \begin{multline*}
      (\ln \symLikelihood)' = \frac{\symNumReg + 2\symCount_{4\symMaxUpdateValMax} + \symCount_{4\symMaxUpdateValMax+2}-\symCount_{4\symMaxUpdateValMax+3}}{\symZ_{\symMaxUpdateValMax-1}}
      \\
      + \frac{\symCount_{4\symMaxUpdateValMax}+2\symCount_{4\symMaxUpdateValMax+1}+2\symCount_{4\symMaxUpdateValMax+2}+3\symCount_{4\symMaxUpdateValMax+3}}{\symZ_{\symMaxUpdateValMax-1}-1} + \frac{\symCount_{4\symMaxUpdateValMax+1} +\symCount_{4\symMaxUpdateValMax+3}}{1+\symZ_{\symMaxUpdateValMax-1}}
    \end{multline*}
    equal to 0
    \begin{align*}
       & \frac{\symNumReg + 2\symCount_{4\symMaxUpdateValMax} + \symCount_{4\symMaxUpdateValMax+2}-\symCount_{4\symMaxUpdateValMax+3}}{\symZEstimate_{\symMaxUpdateValMax-1}}
      \\
       & \quad + \frac{\symCount_{4\symMaxUpdateValMax}+2\symCount_{4\symMaxUpdateValMax+1}+2\symCount_{4\symMaxUpdateValMax+2}+3\symCount_{4\symMaxUpdateValMax+3}}{\symZEstimate_{\symMaxUpdateValMax-1}-1} + \frac{\symCount_{4\symMaxUpdateValMax+1} +\symCount_{4\symMaxUpdateValMax+3}}{1+\symZEstimate_{\symMaxUpdateValMax-1}}=0
      \\
       & \Rightarrow \symAlpha \symZEstimate_{\symMaxUpdateValMax-1}^2 +\symBeta \symZEstimate_{\symMaxUpdateValMax-1} -\symGamma = 0
      \Rightarrow \symZEstimate_{\symMaxUpdateValMax-1} = \textstyle
      \frac{
        \sqrt{\symBeta^2 + 4\symAlpha \symGamma} - \symBeta
      }{
        2 \symAlpha
      }.
    \end{align*}
    Using the invariance property of \ac{ML} estimators we can estimate $\symZ_{\symMaxUpdateValMax} = \sqrt{\symZ_{\symMaxUpdateValMax-1}}$ by $\symZEstimate_{\symMaxUpdateValMax} = \sqrt{\symZEstimate_{\symMaxUpdateValMax-1}}$, which completes the proof.
  \end{proof}

  \begin{lemma}
    \label{lem:martingale_contrib}
    The probability that a register $\symRegister_\symRegAddr$ of an \ac{ULL} with $\symNumReg$ registers is changed with the next new distinct element is given by $\symRegMartingale(\symRegister_\symRegAddr)$ where
    \begin{equation*}
      \symRegMartingale(\symRegister) =
      \begin{cases}
        \frac{1}{\symNumReg}                                                            & \symRegister = 0,                                                                                                              \\
        \frac{1}{2\symNumReg}                                                           & \symRegister = 4,                                                                                                              \\
        \frac{3}{4\symNumReg}                                                           & \symRegister = 8,                                                                                                              \\
        \frac{1}{4\symNumReg}                                                           & \symRegister = 10,                                                                                                             \\
        \frac{7 -2\symIndexBit_1 - 4\symIndexBit_2}{2^\symMaxUpdateVal\symNumReg}       & \symRegister = 4\symMaxUpdateVal + \langle\symIndexBit_1\symIndexBit_2\rangle_2, 3\leq \symMaxUpdateVal < \symMaxUpdateValMax, \\
        \frac{3 - \symIndexBit_1 -2\symIndexBit_2}{2^{\symMaxUpdateValMax-1}\symNumReg} & \symRegister = 4\symMaxUpdateValMax + \langle\symIndexBit_1\symIndexBit_2\rangle_2.
      \end{cases}
    \end{equation*}
  \end{lemma}
  \begin{proof}
    The probability that register $\symRegister_\symRegAddr$ is changed is given by the probability that the register is selected which is $\frac{1}{\symNumReg}$ multiplied by the probability that the register is changed by the update value $\symUpdateVal$ whose \ac{PMF} is given by
    \begin{equation*}
      \symDensityUpdate(\symUpdateVal) =
      \begin{cases}
        2^{-\symUpdateVal}        & 1 \leq \symUpdateVal < \symMaxUpdateValMax, \\
        2^{1-\symMaxUpdateValMax} & \symUpdateVal = \symMaxUpdateValMax.
      \end{cases}
    \end{equation*}
    This formula corresponds to \eqref{equ:geometric} with $\symBase=2$ and truncating the maximum update value at $\symMaxUpdateValMax$. If the register is still in its initial state ($\symRegister_\symRegAddr=0$), any update value will change the register and therefore $\symRegMartingale(\symRegister_\symRegAddr) = \frac{1}{\symNumReg}$. For $\symRegister_\symRegAddr=4$, any update value $\symUpdateVal\geq 2$ will change the register which results in $\symRegMartingale(\symRegister_\symRegAddr) = \frac{1}{2\symNumReg}$.
    If $\symRegister_\symRegAddr=8$, update values $\symUpdateVal \geq 3$ or $\symUpdateVal=1$ will change the register leading to $\symRegMartingale(\symRegister_\symRegAddr) = \frac{1}{\symNumReg}(\frac{1}{2}+\frac{1}{4}) = \frac{3}{4\symNumReg}$.
    If $\symRegister_\symRegAddr=10$, update values $\symUpdateVal \geq 3$ will change the register leading to $\symRegMartingale(\symRegister_\symRegAddr) = \frac{1}{4\symNumReg}$.
    If $\symRegister_\symRegAddr=4\symMaxUpdateVal + \langle\symIndexBit_1\symIndexBit_2\rangle_2$ with $3\leq \symMaxUpdateVal < \symMaxUpdateValMax$, update values $\symUpdateVal$ satisfying $(\symUpdateVal > \symMaxUpdateVal )\vee ((\symUpdateVal=\symMaxUpdateVal-1) \wedge (\symIndexBit_1= 0))\vee ((\symUpdateVal=\symMaxUpdateVal-2) \wedge (\symIndexBit_2 = 0))$ will change the register leading to $\symRegMartingale(\symRegister_\symRegAddr) = \frac{1}{\symNumReg}(\frac{1}{2^\symMaxUpdateVal} + \frac{1-\symIndexBit_1}{2^{\symMaxUpdateVal-1}} + \frac{1-\symIndexBit_2}{2^{\symMaxUpdateVal-2}} ) = \frac{7 -2\symIndexBit_1 - 4\symIndexBit_2}{2^\symMaxUpdateVal\symNumReg}$.
    If $\symRegister_\symRegAddr=4\symMaxUpdateValMax + \langle\symIndexBit_1\symIndexBit_2\rangle_2$, update values $\symUpdateVal$ satisfying $((\symUpdateVal=\symMaxUpdateValMax-1) \wedge (\symIndexBit_1= 0))\vee ((\symUpdateVal=\symMaxUpdateValMax-2) \wedge (\symIndexBit_2 = 0))$ will change the register leading to $\symRegMartingale(\symRegister_\symRegAddr) = \frac{1}{\symNumReg}(\frac{1-\symIndexBit_1}{2^{\symMaxUpdateValMax-1}} + \frac{1-\symIndexBit_2}{2^{\symMaxUpdateValMax-2}} ) = \frac{3 - \symIndexBit_1 -2\symIndexBit_2}{2^{\symMaxUpdateValMax-1}\symNumReg}$.
  \end{proof}

  \begin{lemma}
    \label{lem:mvp_martingale}
    The relative variance of the martingale estimator \cite{Ting2014, Cohen2015, Pettie2021a} applied to the data structure described in \cref{sec:data_structure} can be approximated by
    \begin{equation*}
      \symVariance(\symCardinalityEstimator/\symCardinality)
      \approx
      \frac{\ln(\symBase)(1+\frac{\symBase^{-\symNumExtraBits}}{\symBase-1})}{2\symNumReg}
    \end{equation*}
    if the simplified model as described in \cref{sec:stat_model} can be assumed and the distinct count $\symCardinality$ as well as the number of registers $\symNumReg$ are sufficiently large.
  \end{lemma}
  \begin{proof}
    According to \cite{Pettie2021a} the variance of the relative estimation error can be expressed as
    \begin{equation*}
      \symVariance(\symCardinalityEstimator/\symCardinality) =
      \frac{1}{\symCardinality^2}
      \left(
      \left(\sum_{\symCardinality'=1}^\symCardinality \symExpectation(1/\symStateChangeProbability(\symRegister_0, \ldots, \symRegister_{\symNumReg-1})\vert \symCardinality')\right)
      -
      \symCardinality
      \right).
    \end{equation*}
    $\symStateChangeProbability(\symRegister_0, \ldots, \symRegister_{\symNumReg-1})$ is the probability that a new distinct element triggers a state change, if the current state is given by registers $\symRegister_0, \ldots, \symRegister_{\symNumReg-1}$.
    The probability that a register is changed by some update value $\symUpdateVal$, is given by the probability that the register was never updated with the same value or a value greater than $\symUpdateVal+\symNumExtraBits$ before. The probability that this is the case after adding $\symCardinality$ distinct elements is given by $\symProbability(\overline\symUpdateEvent_\symUpdateVal \wedge \bigwedge_{\symMaxUpdateVal=\symUpdateVal+\symNumExtraBits+1}^\infty \overline\symUpdateEvent_\symMaxUpdateVal ) = \symZ_\symUpdateVal^{1+\frac{\symBase^{-\symNumExtraBits}}{\symBase-1}}$ with $\symZ_\symUpdateVal=e^{-\frac{\symCardinality(\symBase-1)}{\symNumReg\symBase^{\symUpdateVal}}}$. Therefore, the probability of
    the expectation of $\symStateChangeProbability(\symRegister_0, \ldots, \symRegister_{\symNumReg-1})$ conditioned on the distinct count $\symCardinality$ is
    \begin{align*}
       & \symExpectation(\symStateChangeProbability(\symRegister_0, \ldots, \symRegister_{\symNumReg-1}) \vert \symCardinality)
      \approx
      \sum_{\symUpdateVal=1}^\infty
      \symDensityUpdate(\symUpdateVal)
      \symZ_\symUpdateVal^{1+\frac{\symBase^{-\symNumExtraBits}}{\symBase-1}}
      \\
       & \approx
      \sum_{\symUpdateVal=1}^\infty
      (\symBase - 1)\symBase^{-\symUpdateVal}
      \symZ_\symUpdateVal^{1+\frac{\symBase^{-\symNumExtraBits}}{\symBase-1}}
      \stackrel{\symCardinality\rightarrow\infty}{\approx}
      \sum_{\symUpdateVal=-\infty}^\infty
      (\symBase - 1)\symBase^{-\symUpdateVal}
      \symZ_\symUpdateVal^{1+\frac{\symBase^{-\symNumExtraBits}}{\symBase-1}}
      \\
       & =
      -\frac{\symNumReg}{\symCardinality}
      \sum_{\symUpdateVal=-\infty}^\infty
      \symZ_\symUpdateVal^{1+\frac{\symBase^{-\symNumExtraBits}}{\symBase-1}}
      \ln \symZ_\symUpdateVal
      \approx
      \frac{\symNumReg}{\symCardinality}
      \frac{1}{\ln \symBase}\int_0^1 \symZ^{\frac{\symBase^{-\symNumExtraBits}}{\symBase-1}} d\symZ
      \\
       & =
      \frac{\symNumReg}{\symCardinality}
      \frac{1}{\ln \symBase}\frac{1}{1 + \frac{\symBase^{-\symNumExtraBits}}{\symBase-1}}.
    \end{align*}
    Here we used \cref{lem:approximation} to approximate the sum by an integral. Since the state change probability is the sum of the individual register change probabilities, the delta method can be used for $\symNumReg\rightarrow \infty$ to approximate
    $\symExpectation(1/\symStateChangeProbability(\symRegister_0, \ldots, \symRegister_{\symNumReg-1})\vert \symCardinality)$ by
    \begin{multline*}
      \symExpectation(1/\symStateChangeProbability(\symRegister_0, \ldots, \symRegister_{\symNumReg-1})\vert \symCardinality)
      \approx
      \frac{1}{\symExpectation(\symStateChangeProbability(\symRegister_0, \ldots, \symRegister_{\symNumReg-1}) \vert \symCardinality)}
      \\
      \approx
      \frac{\symCardinality}{\symNumReg}
      \ln(\symBase) \left(1 + \frac{\symBase^{-\symNumExtraBits}}{\symBase-1}\right)
    \end{multline*}
    which finally leads to
    \begin{multline*}
      \symVariance(\symCardinalityEstimator/\symCardinality)
      \stackrel{\symCardinality,\symNumReg\rightarrow\infty}{\approx}
      \frac{1}{\symCardinality^2}
      \left(
      \left(\sum_{\symCardinality'=1}^\symCardinality \frac{\symCardinality'}{\symNumReg}
        \ln(\symBase) \left(1 + \frac{\symBase^{-\symNumExtraBits}}{\symBase-1}\right)\right)
      -
      \symCardinality
      \right)
      \\
      =
      \frac{1}{\symCardinality^2}
      \frac{\symCardinality (\symCardinality+1)}{2\symNumReg}
      \ln(\symBase) \left(1 + \frac{\symBase^{-\symNumExtraBits}}{\symBase-1}\right)
      -
      \frac{1}{\symCardinality}
      \stackrel{\symCardinality\rightarrow\infty}{\approx}
      \frac{\ln(\symBase)(1+\frac{\symBase^{-\symNumExtraBits}}{\symBase-1})}{2\symNumReg}.
    \end{multline*}
  \end{proof}

  \section{Additional Figures}
  \label{app:figures}

  \begin{figure}[h]
    \centering
    \includegraphics[width=\linewidth]{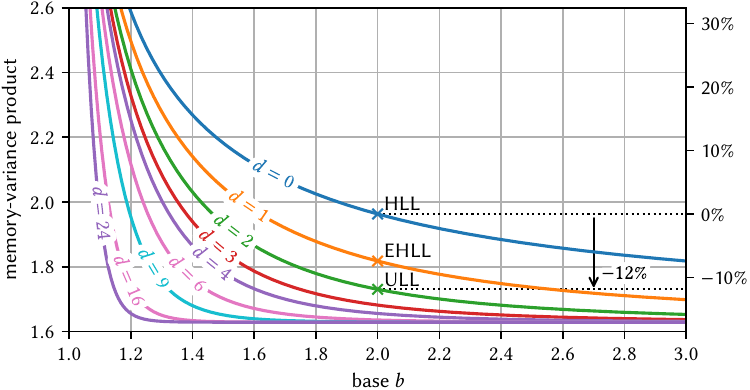}
    \caption{\boldmath The theoretical asymptotic \acf*{MVP} over the base $\symBase$ for various values of $\symNumExtraBits$ when using martingale estimation under the assumption of optimal lossless compression. The \acs*{MVP} of \acf*{ULL} is 12\% smaller than that of \acf*{HLL}. The theoretical limit of $1.63$ \cite{Pettie2021a} is reached for $\symBase^{-\symNumExtraBits}/(\symBase-1)\rightarrow 0$.}
    \label{fig:mvp_compressed_martingale}
  \end{figure}

\fi

\end{document}